\RequirePackage{snapshot}
\documentclass[11pt,twoside,british]{article}
\usepackage[T1]{fontenc}
\usepackage[latin9]{inputenc}
\usepackage[a4paper]{geometry}
\usepackage{fancyhdr}
\usepackage{color}
\usepackage{babel}
\usepackage{verbatim}
\usepackage{refstyle}
\usepackage{mathrsfs}
\usepackage{mathtools}
\usepackage{enumitem}
\usepackage{amsmath}
\usepackage{amsthm}
\usepackage{amssymb}
\usepackage{graphicx}
\usepackage{setspace}
\usepackage[authoryear,longnamesfirst]{natbib}
\usepackage{xargs}[2008/03/08]
\usepackage[pdfusetitle,
 bookmarks=true,bookmarksnumbered=false,bookmarksopen=false,
 breaklinks=false,pdfborder={0 0 0},pdfborderstyle={},backref=false,colorlinks=false]
 {hyperref}

\makeatletter

\AtBeginDocument{\providecommand\eqref[1]{\ref{eq:#1}}}
\AtBeginDocument{\providecommand\subsecref[1]{\ref{subsec:#1}}}
\AtBeginDocument{\providecommand\secref[1]{\ref{sec:#1}}}
\AtBeginDocument{\providecommand\lemref[1]{\ref{lem:#1}}}
\AtBeginDocument{\providecommand\appref[1]{\ref{app:#1}}}
\AtBeginDocument{\providecommand\propref[1]{\ref{prop:#1}}}
\AtBeginDocument{\providecommand\thmref[1]{\ref{thm:#1}}}
\AtBeginDocument{\providecommand\assref[1]{\ref{ass:#1}}}
\AtBeginDocument{\providecommand\defref[1]{\ref{def:#1}}}
\AtBeginDocument{\providecommand\corref[1]{\ref{cor:#1}}}
\AtBeginDocument{\providecommand\figref[1]{\ref{fig:#1}}}
\AtBeginDocument{\providecommand\enuref[1]{\ref{enu:#1}}}
\AtBeginDocument{\providecommand\remref[1]{\ref{rem:#1}}}
\RS@ifundefined{subsecref}
  {\newref{subsec}{name = \RSsectxt}}
  {}
\RS@ifundefined{thmref}
  {\def\RSthmtxt{theorem~}\newref{thm}{name = \RSthmtxt}}
  {}
\RS@ifundefined{lemref}
  {\def\RSlemtxt{lemma~}\newref{lem}{name = \RSlemtxt}}
  {}

\numberwithin{equation}{section}
\numberwithin{figure}{section}
\numberwithin{table}{section}
\theoremstyle{remark}
\newtheorem*{notation*}{\protect\notationname}
\theoremstyle{plain}
\newtheorem{lem}{\protect\lemmaname}[section]
\theoremstyle{remark}
\newtheorem{rem}{\protect\remarkname}[section]
\theoremstyle{definition}
\newtheorem{defn}{\protect\definitionname}[section]
\theoremstyle{plain}
\newtheorem{assumption}{\protect\assumptionname}
\theoremstyle{plain}
\newtheorem{prop}{\protect\propositionname}[section]
\theoremstyle{plain}
\newtheorem{thm}{\protect\theoremname}[section]
\theoremstyle{plain}
\newtheorem{cor}{\protect\corollaryname}[section]
\theoremstyle{definition}
\newtheorem{example}{\protect\examplename}[section]

\setlist[enumerate,1]{label=\upshape{(\roman*)}, ref=(\roman*)}
\setlist[enumerate,2]{label=\upshape{(\alph*)}, ref=(\alph*)}
\setlist[enumerate,3]{label=\upshape{\roman*.}, ref=\roman*}

\makeatletter
  \@ifpackageloaded{refstyle}{}{\usepackage{refstyle}}
\makeatother

\newref{thm}{name = Theorem~}
\newref{prop}{name = Proposition~}
\newref{lem}{name = Lemma~}
\newref{cor}{name = Corollary~}
\newref{ass}{name = }
\newref{exa}{name = Example~}
\newref{rem}{name = Remark~}
\newref{def}{name = Definition~}

\newref{eq}{refcmd = (\ref{#1})}

\newref{sec}{name = Section~}
\newref{sub}{name = Section~}
\newref{subsec}{name = Section~}
\newref{app}{name = Appendix~}
\newref{supp}{name = Online Appendix~}

\newref{fig}{name = Figure~}
\newref{tbl}{name = Table~}

\makeatletter
  \@ifpackageloaded{mathrsfs}{}{\usepackage{mathrsfs}}
\makeatother

\date{}

\usepackage{nextpage}

\allowdisplaybreaks[1]

\usepackage{scalefnt}
\usepackage[group-separator={,}]{siunitx}
\newcommand\smaller[2][0.85]{{\scalefont{#1}#2}}
\newcommand{\ass}[1]{{\upshape{\smaller[0.76]{#1}}}}

\newcommand{\assumpname}[1]{%
  \renewcommand{\theassumption}{\ass{#1}}%
}

\makeatletter
\newsavebox{\@brx}
\newcommand{\dbllangle}[1][]{\savebox{\@brx}{\(\m@th{#1\langle}\)}%
  \mathopen{\copy\@brx\kern-0.5\wd\@brx\usebox{\@brx}}}
\newcommand{\dblrangle}[1][]{\savebox{\@brx}{\(\m@th{#1\rangle}\)}%
  \mathclose{\copy\@brx\kern-0.5\wd\@brx\usebox{\@brx}}}
\makeatother

\usepackage{changepage}

\theoremstyle{definition}
\newtheorem{examplex}{Example}
\renewenvironment{example}
  {\pushQED{\qed}\examplex}
  {\popQED\endexamplex}

\numberwithin{examplex}{section}

\newcounter{subremark}[rem]
\renewcommand{\thesubremark}{(\roman{subremark})}

\newcommand{\subremark}{%
  \refstepcounter{subremark}%
  \thesubremark{}.%
}

\newcounter{savedexnumber}

\usepackage{mathdots}

\usepackage[usestackEOL]{stackengine}
\setstackgap{S}{5pt}
\newcommand{\authaffil}[2]{\Shortunderstack{#1\\\small{#2}}}

\usepackage{needspace}

\newref{lr}{name = \ass{LR.}}
\newref{lrp}{name = \ass{LR.}}

\makeatother

\providecommand{\assumptionname}{Assumption}
\providecommand{\corollaryname}{Corollary}
\providecommand{\definitionname}{Definition}
\providecommand{\examplename}{Example}
\providecommand{\lemmaname}{Lemma}
\providecommand{\notationname}{Notation}
\providecommand{\propositionname}{Proposition}
\providecommand{\remarkname}{Remark}
\providecommand{\theoremname}{Theorem}

\begin{document}
% Macros, Version 4
% Updated: 2020-11-19

% FORMATTING

\global\long\def\uwrite#1#2{\underset{#2}{\underbrace{#1}} }%

\global\long\def\blw#1{\ensuremath{\underline{#1}}}%

\global\long\def\abv#1{\ensuremath{\overline{#1}}}%

\global\long\def\vect#1{\mathbf{#1}}%

% SETS AND SEQUENCES

\global\long\def\smlseq#1{\{#1\} }%

\global\long\def\seq#1{\left\{  #1\right\}  }%

\global\long\def\smlsetof#1#2{\{#1\mid#2\} }%

\global\long\def\setof#1#2{\left\{  #1\mid#2\right\}  }%

% LIMITS

\global\long\def\goesto{\ensuremath{\rightarrow}}%

\global\long\def\ngoesto{\ensuremath{\nrightarrow}}%

\global\long\def\uto{\ensuremath{\uparrow}}%

\global\long\def\dto{\ensuremath{\downarrow}}%

\global\long\def\uuto{\ensuremath{\upuparrows}}%

\global\long\def\ddto{\ensuremath{\downdownarrows}}%

\global\long\def\ulrto{\ensuremath{\nearrow}}%

\global\long\def\dlrto{\ensuremath{\searrow}}%

% FUNCTIONS AND FUNCTION SPACES

\global\long\def\setmap{\ensuremath{\rightarrow}}%

\global\long\def\elmap{\ensuremath{\mapsto}}%

\global\long\def\compose{\ensuremath{\circ}}%

\global\long\def\cont{C}%

\global\long\def\cadlag{D}%

\global\long\def\Ellp#1{\ensuremath{\mathcal{L}^{#1}}}%

% SETS OF NUMBERS

\global\long\def\naturals{\ensuremath{\mathbb{N}}}%

\global\long\def\reals{\mathbb{R}}%

\global\long\def\complex{\mathbb{C}}%

\global\long\def\rationals{\mathbb{Q}}%

\global\long\def\integers{\mathbb{Z}}%

% NORMS, MODULI, INNER PRODUCTS

\global\long\def\abs#1{\ensuremath{\left|#1\right|}}%

\global\long\def\smlabs#1{\ensuremath{\lvert#1\rvert}}%
 
\global\long\def\bigabs#1{\ensuremath{\bigl|#1\bigr|}}%
 
\global\long\def\Bigabs#1{\ensuremath{\Bigl|#1\Bigr|}}%
 
\global\long\def\biggabs#1{\ensuremath{\biggl|#1\biggr|}}%

\global\long\def\norm#1{\ensuremath{\left\Vert #1\right\Vert }}%

\global\long\def\smlnorm#1{\ensuremath{\lVert#1\rVert}}%
 
\global\long\def\bignorm#1{\ensuremath{\bigl\|#1\bigr\|}}%
 
\global\long\def\Bignorm#1{\ensuremath{\Bigl\|#1\Bigr\|}}%
 
\global\long\def\biggnorm#1{\ensuremath{\biggl\|#1\biggr\|}}%

\global\long\def\floor#1{\left\lfloor #1\right\rfloor }%
\global\long\def\smlfloor#1{\lfloor#1\rfloor}%

\global\long\def\ceil#1{\left\lceil #1\right\rceil }%
\global\long\def\smlceil#1{\lceil#1\rceil}%

% SET OPERATIONS

\global\long\def\Union{\ensuremath{\bigcup}}%

\global\long\def\Intsect{\ensuremath{\bigcap}}%

\global\long\def\union{\ensuremath{\cup}}%

\global\long\def\intsect{\ensuremath{\cap}}%

\global\long\def\pset{\ensuremath{\mathcal{P}}}%

\global\long\def\clsr#1{\ensuremath{\overline{#1}}}%

\global\long\def\symd{\ensuremath{\Delta}}%

\global\long\def\intr{\operatorname{int}}%

\global\long\def\cprod{\otimes}%

\global\long\def\Cprod{\bigotimes}%

% HILBERT SPACES

\global\long\def\smlinprd#1#2{\ensuremath{\langle#1,#2\rangle}}%

\global\long\def\inprd#1#2{\ensuremath{\left\langle #1,#2\right\rangle }}%

\global\long\def\orthog{\ensuremath{\perp}}%

\global\long\def\dirsum{\ensuremath{\oplus}}%

% LINEAR ALGEBRA

\global\long\def\spn{\operatorname{sp}}%

\global\long\def\rank{\operatorname{rk}}%

\global\long\def\proj{\operatorname{proj}}%

\global\long\def\tr{\operatorname{tr}}%

\global\long\def\vek{\operatorname{vec}}%

\global\long\def\diag{\operatorname{diag}}%

\global\long\def\col{\operatorname{col}}%

% PROBABILITY SPACES AND SIGMA-FIELDS

\global\long\def\smpl{\ensuremath{\Omega}}%

\global\long\def\elsmp{\ensuremath{\omega}}%

\global\long\def\sigf#1{\mathcal{#1}}%

\global\long\def\sigfield{\ensuremath{\mathcal{F}}}%
\global\long\def\sigfieldg{\ensuremath{\mathcal{G}}}%

\global\long\def\flt#1{\mathcal{#1}}%

\global\long\def\filt{\mathcal{F}}%
\global\long\def\filtg{\mathcal{G}}%

\global\long\def\Borel{\ensuremath{\mathcal{B}}}%

\global\long\def\cyl{\ensuremath{\mathcal{C}}}%

\global\long\def\nulls{\ensuremath{\mathcal{N}}}%

\global\long\def\gauss{\mathfrak{g}}%

\global\long\def\leb{\mathfrak{m}}%

% Probability and expectation

\global\long\def\prob{\ensuremath{\mathbb{P}}}%

\global\long\def\Prob{\ensuremath{\mathbb{P}}}%

\global\long\def\Probs{\mathcal{P}}%

\global\long\def\PROBS{\mathcal{M}}%

\global\long\def\expect{\ensuremath{\mathbb{E}}}%

\global\long\def\Expect{\ensuremath{\mathbb{E}}}%

\global\long\def\probspc{\ensuremath{(\smpl,\filt,\Prob)}}%

% Distributions and stochastic convergence

\global\long\def\iid{\ensuremath{\textnormal{i.i.d.}}}%

\global\long\def\as{\ensuremath{\textnormal{a.s.}}}%

\global\long\def\asp{\ensuremath{\textnormal{a.s.p.}}}%

\global\long\def\io{\ensuremath{\ensuremath{\textnormal{i.o.}}}}%

\newcommand\independent{\protect\mathpalette{\protect\independenT}{\perp}}
\def\independenT#1#2{\mathrel{\rlap{$#1#2$}\mkern2mu{#1#2}}}

\global\long\def\indep{\independent}%

\global\long\def\distrib{\ensuremath{\sim}}%

\global\long\def\distiid{\ensuremath{\sim_{\iid}}}%

\global\long\def\asydist{\ensuremath{\overset{a}{\distrib}}}%

\global\long\def\inprob{\ensuremath{\overset{p}{\goesto}}}%

\global\long\def\inprobu#1{\ensuremath{\overset{#1}{\goesto}}}%

\global\long\def\inas{\ensuremath{\overset{\as}{\goesto}}}%

\global\long\def\eqas{=_{\as}}%

\global\long\def\inLp#1{\ensuremath{\overset{\Ellp{#1}}{\goesto}}}%

\global\long\def\indist{\ensuremath{\overset{d}{\goesto}}}%

\global\long\def\eqdist{=_{d}}%

\global\long\def\wkc{\ensuremath{\rightsquigarrow}}%

\global\long\def\wkcu#1{\overset{#1}{\ensuremath{\rightsquigarrow}}}%

\global\long\def\plim{\operatorname*{plim}}%

% Moments

\global\long\def\var{\operatorname{var}}%

\global\long\def\lrvar{\operatorname{lrvar}}%

\global\long\def\cov{\operatorname{cov}}%

\global\long\def\corr{\operatorname{corr}}%

\global\long\def\bias{\operatorname{bias}}%

\global\long\def\MSE{\operatorname{MSE}}%

\global\long\def\med{\operatorname{med}}%

\global\long\def\avar{\operatorname{avar}}%

\global\long\def\se{\operatorname{se}}%

\global\long\def\sd{\operatorname{sd}}%

% Testing

\global\long\def\nullhyp{H_{0}}%

\global\long\def\althyp{H_{1}}%

\global\long\def\ci{\mathcal{C}}%

% STOCHASTIC PROCESSES

\global\long\def\simple{\mathcal{R}}%

\global\long\def\sring{\mathcal{A}}%

\global\long\def\sproc{\mathcal{H}}%

\global\long\def\Wiener{\ensuremath{\mathbb{W}}}%

\global\long\def\sint{\bullet}%

\global\long\def\cv#1{\left\langle #1\right\rangle }%

\global\long\def\smlcv#1{\langle#1\rangle}%

\global\long\def\qv#1{\left[#1\right]}%

\global\long\def\smlqv#1{[#1]}%

% MISCELLANEOUS

\global\long\def\trans{\mathsf{T}}%

\global\long\def\indic{\ensuremath{\mathbf{1}}}%

\global\long\def\Lagr{\mathcal{L}}%

\global\long\def\grad{\nabla}%

\global\long\def\pmin{\ensuremath{\wedge}}%
\global\long\def\Pmin{\ensuremath{\bigwedge}}%

\global\long\def\pmax{\ensuremath{\vee}}%
\global\long\def\Pmax{\ensuremath{\bigvee}}%

\global\long\def\sgn{\operatorname{sgn}}%

\global\long\def\argmin{\operatorname*{argmin}}%

\global\long\def\argmax{\operatorname*{argmax}}%

\global\long\def\Rp{\operatorname{Re}}%

\global\long\def\Ip{\operatorname{Im}}%

\global\long\def\deriv{\ensuremath{\mathrm{d}}}%

\global\long\def\diffnspc{\ensuremath{\deriv}}%

\global\long\def\diff{\ensuremath{\,\deriv}}%

\global\long\def\i{\ensuremath{\mathrm{i}}}%

\global\long\def\e{\mathrm{e}}%

\global\long\def\sep{,\ }%

\global\long\def\defeq{\coloneqq}%

\global\long\def\eqdef{\eqqcolon}%

%% CROSS REFERENCING

\begin{comment}
\global\long\def\objlabel#1#2{\{\}}%

\global\long\def\objref#1{\{\}}%
\end{comment}

\global\long\def\err{\varepsilon}%

\global\long\def\mset#1{\mathcal{#1}}%

\global\long\def\largedec#1{\mathbf{#1}}%

\global\long\def\z{\largedec z}%

\newcommandx\Ican[1][usedefault, addprefix=\global, 1=]{I_{#1}^{\ast}}%

\global\long\def\jsr{{\scriptstyle \mathrm{JSR}}}%

\global\long\def\cjsr{{\scriptstyle \mathrm{CJSR}}}%

\global\long\def\rsr{{\scriptstyle \mathrm{RJSR}}}%

\global\long\def\ctspc{\mathscr{M}}%

\global\long\def\b#1{\boldsymbol{#1}}%

\global\long\def\pseudy{\text{\ensuremath{\abv y}}}%

\global\long\def\regcoef{\kappa}%

\global\long\def\adj{\operatorname{adj}}%

\global\long\def\llangle{\dbllangle}%

\global\long\def\rrangle{\dblrangle}%

\global\long\def\smldblangle#1{\ensuremath{\llangle#1\rrangle}}%

\newcommand{\casecens}{\upshape{(i)}}

\newcommand{\caseclas}{\upshape{(ii)}}

\newcommand{\casestat}{\upshape{(iii)}}

\global\long\def\delcens{\mathrm{(i)}}%

\global\long\def\delclas{\mathrm{(ii)}}%

\global\long\def\smlf{f}%

\global\long\def\bigf{\b f}%

\global\long\def\fe{f}%

\global\long\def\ga{g}%

\global\long\def\co{\chi}%

\global\long\def\CO{\mathrm{X}}%

\global\long\def\EQ{\Theta}%

\global\long\def\eq{\theta}%

\global\long\def\ct{\psi}%

\global\long\def\cn{\mu}%

\global\long\def\cvar{{\cal M}}%

\global\long\def\codf{\chi}%

\global\long\def\eqdf{\chi_{2}}%

\global\long\def\ctdf{\chi_{1}}%

\global\long\def\set#1{\mathscr{#1}}%

\global\long\def\srp{\mathrm{H}}%

\global\long\def\srfn{h}%

\global\long\def\ch{\operatorname{co}}%

\global\long\def\im{\operatorname{im}}%

\global\long\def\lin{\mathrm{lin}}%

\global\long\def\state{\mathfrak{z}}%

\global\long\def\fespc{\mathscr{F}}%

\global\long\def\denspc{\mathscr{R}}%

\global\long\def\den{\varrho}%

\global\long\def\cden{\rho}%

\title{Common Trends and Long-Run Identification \\ in Nonlinear Structural
VARs}
\author{\authaffil{James A.\ Duffy\footnotemark[1]{}}{University of Oxford}\hspace{3cm}
\authaffil{Sophocles Mavroeidis\footnotemark[2]{}}{University
of Oxford}}
\date{\vspace*{0.3cm}September 2024}

\maketitle
\renewcommand*{\thefootnote}{\fnsymbol{footnote}}

\footnotetext[1]{Department\ of Economics and Corpus Christi College;
\texttt{james.duffy@economics.ox.ac.uk}}

\footnotetext[2]{Department of Economics and University College;
\texttt{sophocles.mavroeidis@economics.ox.ac.uk}}

\renewcommand*{\thefootnote}{\arabic{footnote}}

\setcounter{footnote}{0}
\begin{abstract}
\noindent While it is widely recognised that linear (structural) VARs
may fail to capture important aspects of economic time series, the
use of nonlinear SVARs has to date been almost entirely confined to
the modelling of stationary time series, because of a lack of understanding
as to how common stochastic trends may be accommodated within nonlinear
models. This has unfortunately circumscribed the range of series to
which such models can be applied -- and/or required that these series
be first transformed to stationarity, a potential source of misspecification
-- and prevented the use of long-run identifying restrictions in
these models. To address these problems, we develop a flexible class
of additively time-separable nonlinear SVARs, which subsume models
with threshold-type endogenous regime switching, both of the piecewise
linear and smooth transition varieties. We extend the Granger--Johansen
representation theorem to this class of models, obtaining conditions
that specialise exactly to the usual ones when the model is linear.
We further show that, as a corollary, these models are capable of
supporting the same kinds of long-run identifying restrictions as
are available in linearly cointegrated SVARs.
\end{abstract}
\vfill

\noindent First version: April 2024. We thank B.~Beare, J.~Dolado,
A.\ Escribano, S.~Johansen, Y.~Lu, B.\ Nielsen, B.~Rossi, and
seminar participants at UC3 Madrid, Oxford, Sydney, the 2024 BSE Summer
Forum and the 26th Dynamic Econometrics Conference, for their helpful
comments on earlier drafts of this paper.

\thispagestyle{plain}

\pagenumbering{roman}

\newpage{}

\thispagestyle{plain}

\setcounter{tocdepth}{2}

{\singlespacing

\tableofcontents{}

}

\newpage{}

\newpage{}

\pagenumbering{arabic}

\section{Introduction}

For more than four decades, the linear structural VAR (SVAR) model,
\begin{align}
\Phi_{0}z_{t} & =c+\sum_{i=1}^{k}\Phi_{i}z_{t-i}+u_{t} & u_{t} & =\Upsilon\err_{t}\label{eq:svar-intro}
\end{align}
has played a central role in empirical macroeconomics. Structural
VARs provide a satisfyingly coherent framework within which the problem
of identifying causal relations, in the presence of simultaneity,
may be approached by regarding the $p$ endogenous variables $z_{t}$
as generated by an underlying sequence of $p$ exogenous structural
shocks $\err_{t}$. Viewed through the lens of these models, the identification
problem reduces to one of identifying these structural shocks and
their associated impulse responses: equivalently, of identifying the
mapping $\Upsilon$ between $\err_{t}$ and $u_{t}$.

The early literature on structural VARs, following the seminal contribution
of \citet{Sims80}, developed contemporaneously with major advances
in the modelling of common (stochastic) trends in macroeconomic time
series. This latter work culminated in the development of the theory
of the cointegrated VAR, and in particular the Granger--Johansen
representation theorem (GJRT), which showed how a VAR could be configured
so as to accommodate the presence of common stochastic trends, and
the cointegrating relations that are dual to them (\citealp{EG87Ecta};
\citealp{Joh91Ecta}). This importantly demonstrated that the SVAR
model could be reconciled with some of the evident properties of the
(levels of) nonstationary macroeconomic time series, obviating the
need to difference these series to stationarity prior to analysis
(indeed, an implication of this work was that such pre-filtering was
not only unnecessary, but could induce model misspecification). An
important corollary to the GJRT was that by uncovering the mapping
between the common trends and the endogenous variables, it made available
a class of long-run identifying restrictions, relating to whether
certain structural shocks were regarded as having permanent, or merely
transient, effects (\citealp{BlanchardQuah89}; \citealp{KPSW91AER}).

Subsequent work has sought to extend the (S)VAR model so as to allow
for various forms of nonlinearity: typically as modelled via some
sort of regime switching, possibly with smooth transitions between
these regimes (see e.g.\ \citealp{Chan09}; \citealp{TTG10}; \citealp{HT13}).
More recently, interest in structural macroeconomic models that incorporate
occasionally binding constraints, of which the zero lower bound on
interest rates is a leading instance, has had its counterpart in the
development of new class of regime-switching nonlinear VARs. These
newer models are notably distinguished from the earlier nonlinear
VAR literature by being able to accommodate \emph{endogenous} regime
switching, i.e.\ of allowing the current regime to depend on the
current values of the endogenous variables, rather than being pre-determined,
or being dependent on some exogenous switching process (\citealp{SM21};
\citealp{AMSV21}).

However, there remains a significant disconnect between the literature
on nonlinear (S)VARs, and that on cointegration and common trends.
In this respect, the situation is reminiscent of that of the literature
on linear VARs in the early 1980s, with there being little theoretical
understanding of how common trends may be accommodated within nonlinear
VARs, such that the prior removal of common trends by pre-filtering
remains a necessity. One area where significant efforts have been
made to remedy this pertain to a certain class of `nonlinear VECM'
models. These are VAR models that are specified directly in VECM form,
with a linear cointegrating space, but in which the stationary components
of the model (the equilibrium errors and the lagged differences) are
permitted to enter nonlinearly (see e.g.\ \citealp{EM02JTSA}; \citealp{BR04EctJ};
\citealp{Saik05JoE,Saik08ET}; \citealp{KR10JoE,KR13ET}). This class
of models has proved sufficiently tractable to facilitate a nonlinear
extension of the GJRT, but the assumptions made in this literature
are such as to leave a substantial portion of the space of nonlinear
VAR models unexplored.

In this paper, we generalise \eqref{svar-intro} in the direction
of the additively time-separable SVAR
\begin{align}
\fe_{0}(z_{t}) & =c+\sum_{i=1}^{k}\fe_{i}(z_{t-i})+u_{t} & u_{t} & =\Upsilon\err_{t}\label{eq:SVAR-intro}
\end{align}
where $\fe_{0}:\reals^{p}\setmap\reals^{p}$ is continuous and has
a continuous inverse (i.e.\ it is a homeomorphism), and each $\fe_{i}:\reals^{p}\setmap\reals^{p}$
is continuous.\footnote{For a detailed explanation of why the structural shocks enter \eqref{SVAR-intro}
in an apparently similar manner to \eqref{svar-intro}, see \subsecref{identification}
below.} We aim not only to extend the GJRT, but also to preserve exactly
the kind of long-run identifying restrictions that are yielded by
the GJRT in a linear setting, such that it remains possible to fruitfully
discriminate between structural shocks on the basis of the permanence
(or transience) of their effects. Via an associated VECM representation,
this leads us to consider SVARs in which the \emph{image} of the map
\[
\pi(z)\defeq\sum_{i=1}^{k}\fe_{i}(z)-\fe_{0}(z)
\]
is an $r$-dimensional subspace of $\reals^{p}$, a requirement that
we term the \emph{common row space condition} (CRSC); effectively,
this entails a decomposition of the form $\pi(z)=\alpha\eq(z)$. This
reverses a key assumption of the nonlinear VECM literature, where
$\eq(z)=\beta^{\trans}z$ is a fixed linear function, but $\alpha$
is allowed to vary (albeit only as a function of $\beta^{\trans}z_{t-1}$
and $\Delta z_{t-1}$, or their lags). As we discuss, the principal
motivation for the CRSC is that it ensures that only $q=p-r$ structural
shocks may have permanent effects -- whereas if the CRSC is violated,
all $p$ structural shocks may have permanent effects, even though
$z_{t}$ has only $q<p$ common stochastic trends.

Our conditions on the nonlinear SVAR \eqref{SVAR-intro} are otherwise
rather general, and unlike the preceding literature allow for common
stochastic trends to enter $z_{t}$ \emph{nonlinearly}, thereby giving
rise to nonlinear cointegrating relations between elements of $z_{t}$.
(A rigorous discussion of `cointegration' in this nonlinear setting
is deferred to a companion paper, in which the limiting distribution
of the standardised process $n^{-1/2}z_{\smlfloor{n\lambda}}$ is
derived; for a discussion of this in the context of the censored and
kinked SVAR, see \citealp{DMW22}.) These conditions are presented
in \secref{model}, followed by a discussion of the structural shock
identification problem that arises in this model. We further show
that our conditions reduce to exactly those of the linear cointegrated
VAR, when each $\{\fe_{i}\}_{i=0}^{k}$ is linear: and thus they do
not appear to be particularly restrictive.

The principal contribution of this paper is to provide an extension
of the GJRT to the setting of a nonlinear VAR (\secref{common-trends}),
demonstrating how the model \eqref{SVAR-intro} may be configured
so as to generate common stochastic trends. This in turn permits an
analysis of the stability of the steady state equilibria of the model,
and a characterisation of the directions in which the shocks $u_{t}$
have only transient effects (\subsecref{attractors}), and thereby
yields long-run identifying restrictions on $\Upsilon$ (\subsecref{long-run-id}).
The GJRT also facilitates the calculation of the implied long-run
multipliers associated with the shocks (\secref{lr-multipliers}).
We subsequently return, in \secref{without-the-crsc}, to a discussion
of the role played by the CRSC in ensuring the availability of long-run
identifying restrictions. Finally, \subsecref{piecewise-affine} elaborates
on a class of models in which each $\fe_{i}$ is piecewise affine
(or a smooth, convex combination of piecewise affine functions) showing
how the abstract conditions introduced in \secref{model} may be specialised
to this class of models. \secref{conclusion} concludes. Proofs of
all results appear in the appendices.

\begin{notation*}
$e_{m,i}$ denotes the $i$th column of an $m\times m$ identity matrix;
when $m$ is clear from the context, we write this simply as $e_{i}$.
$\smlnorm{\cdot}$ denotes the Euclidean norm on $\reals^{m}$. Matrix
norms are always those induced by the corresponding vector norm. For
$X$ a random vector and $p\geq1$, $\smlnorm X_{p}\defeq(\expect\smlnorm X^{p})^{1/p}$.
For a matrix $\alpha\in\reals^{p\times r}$ with $\rank\alpha=r$,
$\alpha_{\perp}$ denotes some $p\times(p-r)$ matrix, with $\rank\alpha_{\perp}=p-r$,
such that $\alpha_{\perp}^{\trans}\alpha=0$. For $\{A_{i}\}_{i=1}^{k}$
a collection of $m\times n$ matrices, $[A_{i}]_{i=1}^{k}$ denotes
the $mk\times n$ matrix formed as $(A_{1}^{\trans},\ldots,A_{m}^{\trans})^{\trans}$.
The \emph{convex hull} of $\{A_{i}\}_{i=1}^{k}$ is denoted $\ch\{A_{i}\}_{i=1}^{k}\defeq\{\sum_{i=1}^{k}\lambda_{i}A_{i}\mid\lambda_{i}\geq0\sep\sum_{i=1}^{k}\lambda_{i}=1\}$.
For a continuously differentiable function $f:\reals^{\ell}\setmap\reals^{m}$,
$\partial_{x}f(x_{0})$ denotes the Jacobian of $f$ at $x=x_{0}$.
\end{notation*}

\section{Model}

\label{sec:model}

\subsection{Framework: a nonlinear VAR($k$)}

\label{subsec:framework}

We consider an additively time-separable nonlinear VAR($k$) model
for a $p$-dimensional time series $\{z_{t}\}_{t\in\integers}$, of
the form
\begin{equation}
\fe_{0}(z_{t})=c+\sum_{i=1}^{k}\fe_{i}(z_{t-i})+u_{t}\label{eq:nlVAR}
\end{equation}
where $\fe_{i}:\reals^{p}\setmap\reals^{p}$ for $i\in\{0,1,\ldots,k\}$,
and $\{u_{t}\}_{t\in\integers}$ is a sequence of `shocks' taking
values in $\reals^{p}$. (As a convenient location normalisation,
we shall suppose throughout that $\fe_{i}(0)=0$ for all $i\in\{0,\ldots,k\}$.)
By defining
\begin{equation}
\ga_{j}(z)\defeq-\sum_{i=j+1}^{k}\fe_{i}(z)\label{eq:ge}
\end{equation}
for $j\in\{0,\ldots,k-1\}$, we can write the model in nonlinear
VECM form as
\begin{equation}
\Delta\fe_{0}(z_{t})=c+\pi(z_{t-1})+\sum_{i=1}^{k-1}\Delta\ga_{i}(z_{t-i})+u_{t}\label{eq:VECM-first}
\end{equation}
where
\begin{equation}
\pi(z)\defeq-[\ga_{0}(z)+\fe_{0}(z)]=-\fe_{0}(z)+\sum_{i=1}^{k}\fe_{i}(z)\label{eq:pi-def}
\end{equation}
(see also \lemref{vecm} below).

In the linear VAR, where $\pi(z)=\Pi z$ is linear, the key assumption
required for the model to generate common stochastic trends (i.e.\ cointegrated
$I(1)$ processes) is that $\rank\Pi=r=p-q$, where $q$ is the number
of unit roots in the autoregressive polynomial. Equivalently, this
condition may be expressed in terms of the image of the map $z\elmap\Pi z$,
\[
\im\Pi=\{\Pi z\mid z\in\reals^{p}\}=\spn\Pi,
\]
being an $r$-dimensional (linear) subspace of $\reals^{p}$. In extending
the GJRT to the setting of the nonlinear VAR \eqref{nlVAR}, the role
of $\Pi$ (or more precisely, of $z\elmap\Pi z$) will now be played
by $\pi(z)$, and our high-level assumptions below will ensure that
\[
\im\pi=\{\pi(z)\mid z\in\reals^{p}\}
\]
remains an $r$-dimensional (linear) subspace of $\reals^{p}$, exactly
as in the linear case. We term this the \emph{common row space} \emph{condition}
(CRSC). (Note that there is no corresponding restriction on $\ker\pi$;
in particular, this will \emph{not} be required to be a $q$-dimensional
subspace of $\reals^{p}$.) For a further discussion of this condition
-- in particular, of the role that it plays in ensuring that the
long-run identifying restrictions familiar from a linear cointegrated
VAR are also available in the present setting -- see \secref{without-the-crsc}
below.

\subsection{High-level assumptions}

\label{subsec:high-level}

To state our main assumptions on the model \eqref{nlVAR}, we first
provide a more compact representation for \eqref{VECM-first}; its
proof appears in \appref{auxproof}.
\begin{lem}
\label{lem:vecm}Suppose $\{z_{t}\}$ follows \eqref{nlVAR}. Then
\eqref{VECM-first} holds. Defining $\b{\zeta}_{t}\defeq[\zeta_{j,t}]_{j=1}^{k-1}$,
where
\[
\zeta_{j,t}\defeq\sum_{i=j}^{k-1}\ga_{i}(z_{t-i+j}),
\]
and letting $\b z_{t}\defeq(z_{t}^{\trans},\b{\zeta}_{t-1}^{\trans})^{\trans}$,
we also have
\begin{align}
\begin{bmatrix}\Delta\fe_{0}(z_{t})\\
\Delta\b{\zeta}_{t-1}
\end{bmatrix} & =\begin{bmatrix}c\\
0_{p(k-1)}
\end{bmatrix}+\begin{bmatrix}\pi(z_{t-1})+\ga_{1}(z_{t-1})\\
\b{\ga}(z_{t-1})
\end{bmatrix}+\begin{bmatrix}0 & D_{1}\\
0 & D
\end{bmatrix}\begin{bmatrix}z_{t-1}\\
\b{\zeta}_{t-2}
\end{bmatrix}+\begin{bmatrix}u_{t}\\
0_{p(k-1)}
\end{bmatrix}\nonumber \\
 & \eqdef\b c+\b{\pi}(z_{t-1})+\b D\b z_{t-1}+\b u_{t}.\label{eq:VECM-system}
\end{align}
for $\b{\ga}(z)\defeq[\ga_{j}(z)]_{j=1}^{k-1}$,
\[
D\defeq\begin{bmatrix}-I_{p} & I_{p}\\
 & \ddots & \ddots\\
 &  & -I_{p} & I_{p}\\
 &  &  & -I_{p}
\end{bmatrix}\in\reals^{p(k-1)\times p(k-1)}
\]
and $D_{1}$ the first $p$ rows of $D$.
\end{lem}
\begin{rem}
\label{rem:vecm-representation}\subremark{}\label{subrem:higher-order}
The preceding holds exactly as stated when $k\geq3$; the model with
$k\in\{1,2\}$ may be handled as a special case of $k=3$ with at
least one of $\fe_{2}$ and $\fe_{3}$ being identically zero. Alternatively,
we may specialise the above to $k\in\{1,2\}$ by taking $D$ to be
an `empty' matrix when $k=1$, and $D=-I_{p}$ when $k=2$. (In
the proofs and statements of all results, to avoid these kinds of
complications, we shall implicitly maintain throughout, without loss
of generality, that $k\geq3$.)

\subremark{} \eqref{VECM-system} differs notably from the companion
form representation that is typically employed in the analysis of
linear VECMs. In particular, the state vector is not 
\[
\vec{z}_{t}\defeq(z_{t}^{\trans},\ldots,z_{t-k+1}^{\trans})^{\trans},
\]
but rather
\begin{equation}
\b z_{t}=\begin{bmatrix}z_{t}\\
\b{\zeta}_{t-1}
\end{bmatrix}=\begin{bmatrix}z_{t}\\
\zeta_{1,t-1}\\
\vdots\\
\zeta_{k-1,t-1}
\end{bmatrix}=\begin{bmatrix}z_{t}\\
\sum_{i=1}^{k-1}\ga_{i}(z_{t-i})\\
\vdots\\
\ga_{k-1}(z_{t-1})
\end{bmatrix}=\begin{bmatrix}z_{t}\\
\sum_{i=1}^{k-1}\Gamma_{i}z_{t-i}\\
\vdots\\
\Gamma_{k-1}z_{t-1}
\end{bmatrix},\label{eq:bz-in-linear}
\end{equation}
where the final equality holds only in a \emph{linear} VAR. By comparison
with $\vec{z}_{t}$, we have defined $\b z_{t}$ so that the r.h.s.\ of
\eqref{VECM-system} involves a nonlinear transformation of only $z_{t-1}$,
and is otherwise linear in $\b{\zeta}_{t-2}$. This permits a major
simplification of the analysis of the model when $k\geq2$; for this
the assumption of additive separability is crucially important.
\end{rem}

We next define a class $\cvar_{r}$ of nonlinear VAR models, that
will be shown to be consistent with the presence of $q=p-r$ common
stochastic trends in $\{z_{t}\}$, and $r$ `cointegrating relations'
between those series. For some $\theta:\reals^{p}\setmap\reals^{r}$,
define
\begin{align}
\b{\eq}(z) & \defeq\begin{bmatrix}\eq(z)\\
\b{\ga}(z)
\end{bmatrix} & \b D_{0} & \defeq\begin{bmatrix}0_{r\times p} & 0\\
0 & D
\end{bmatrix} & E & \defeq\begin{bmatrix}I_{p}\\
0_{p(k-2)\times p}
\end{bmatrix},\label{eq:bold-theta}
\end{align}
and for $\alpha\in\reals^{p\times r}$ with $\rank\alpha=r$,
\begin{align}
\b{\alpha} & \defeq\begin{bmatrix}\alpha & E^{\trans}\\
0 & I_{p(k-1)}
\end{bmatrix} & \b{\alpha}_{\perp} & \defeq\begin{bmatrix}I_{p}\\
-E
\end{bmatrix}\alpha_{\perp},\label{eq:bold-alpha}
\end{align}
Recall that a function $f:{\cal X\setmap{\cal Y}}$ is a \emph{homeomorphism}
if it is continuous, bijective, and has a continuous inverse. $f$
is \emph{bi-Lipschitz} if there exists a $C_{0}<\infty$ such that
$C_{0}^{-1}\smlnorm{x-x^{\prime}}\leq\smlnorm{f(x)-f(x^{\prime})}\leq C_{0}\smlnorm{x-x^{\prime}}$
for all $x,x^{\prime}\in{\cal X}$; we term any such $C_{0}$ a \emph{bi-Lipschitz
constant} for $f$. For $\mathcal{A}\subset\reals^{m\times m}$ a
bounded collection of matrices, let $\rho_{\jsr}(\mathcal{A})\defeq\limsup_{t\goesto\infty}\sup_{B\in\mathcal{A}^{t}}\rho(B)^{1/t}$
denote its joint spectral radius (JSR; e.g.\ \citealp{Jungers09},
Defn.\ 1.1), where $\mathcal{A}^{t}\defeq\{\prod_{s=1}^{t}M_{s}\mid M_{s}\in\mathcal{A}\}$
is the set of $t$-fold products of matrices in ${\cal A}$, and $\rho(B)$
the spectral radius of $B$.

\needspace{3cm}
\begin{defn}
\label{def:cvar}Let $r\in\{0,\ldots,p\}$, $\alpha\in\reals^{p\times r}$
have $\rank\alpha=r$, $\abv b\in\reals$, and $\abv{\rho}\in[0,1)$.
We say that the VAR($k$) model \eqref{nlVAR}, parametrised by $c\in\reals^{p}$
and $\{\fe_{i}\}_{i=0}^{k}$, belongs to class $\cvar_{r}(\alpha,\abv b,\abv{\rho})$
if:
\begin{enumerate}[label=($\cvar$.\roman*), leftmargin=1.75cm]
\item \label{enu:cvar:homeo}$\fe_{0}:\reals^{p}\setmap\reals^{p}$ is
a homeomorphism;
\item \label{enu:cvar:rank} there exists $\cn\in\reals^{r}$ and $\theta:\reals^{p}\setmap\reals^{r}$
such that 
\begin{align*}
c & =\alpha\cn & \pi(z) & =\alpha\eq(z)
\end{align*}
for all $z\in\reals^{p}$;
\item \label{enu:cvar:jsr}for every $\b z=(z^{\trans},\b{\zeta}^{\trans})^{\trans}$
and $\b z^{\prime}=(z^{\prime\trans},\b{\zeta}^{\prime\trans})^{\trans}$
in $\reals^{p}\times\reals^{p(k-1)}$, there exists a $\b{\beta}\in{\cal B}$
(possibly depending on $\b z$ and $\b z^{\prime}$) such that
\begin{equation}
[(\b{\eq}\compose\fe_{0}^{-1})(z)+\b D_{0}\b z]-[(\b{\eq}\compose\fe_{0}^{-1})(z^{\prime})+\b D_{0}\b z^{\prime}]=\b{\beta}^{\trans}(\b z-\b z^{\prime})\label{eq:gradient}
\end{equation}
where ${\cal B}\subset\reals^{p\times[p(k-1)+r]}$ is closed with
$\max_{\b{\beta}\in{\cal B}}\smlnorm{\b{\beta}}\leq\abv b$ and
\begin{equation}
\rho_{\jsr}(\{I_{p(k-1)+r}+\b{\beta}^{\trans}\b{\alpha}\mid\b{\beta}\in{\cal B}\})\leq\abv{\rho}.\label{eq:jsr}
\end{equation}
\end{enumerate}
Let $\cvar_{r}^{\ast}(\alpha,\abv b,\abv{\rho})$ denote those models
in $\cvar_{r}(\alpha,\abv b,\abv{\rho})$ for which $\fe_{0}=I_{p}$,
the identity on $\reals^{p}$. We denote by $\cvar_{r}$ (resp.\ $\cvar_{r}^{\ast}$)
the union of all classes $\cvar_{r}(\alpha,\abv b,\abv{\rho})$ (resp.\ $\cvar_{r}^{\ast}(\alpha,\abv b,\abv{\rho})$)
with $\alpha\in\reals^{p\times r}$ having $\rank\alpha=r$, $\abv b<\infty$
and $\abv{\rho}<1$.
\end{defn}
\begin{rem}
\label{rem:cvar}\subremark{} We allow $\fe_{0}$ to be non-trivial
to permit (linear and nonlinear) \emph{structural} VARs to be accommodated
within the present setting: see \subsecref{identification} for a
further discussion. In the linear case, $\fe_{0}^{-1}(z)=\Phi_{0}^{-1}z$
is trivially continuous; requiring $\fe_{0}$ to be a homeomorphism
extends this to a nonlinear setting. When deriving the properties
of the series generated by the models in class $\cvar_{r}$, we shall
first rewrite the system in terms of $z_{t}^{\ast}\defeq\fe_{0}(z_{t})$,
which is itself generated by a model in class $\cvar_{r}^{\ast}$
(see \lemref{unstruct} in \appref{auxiliary}).

\subremark{} It follows from the preceding conditions, and principally
from \ref{enu:cvar:jsr}, that $\fe_{i}$ is continuous for each $i\in\{1,\ldots,k\}$;
indeed, these functions are Lipschitz when $\fe_{0}=I_{p}$ (see \lemref{co-cts}
in \appref{auxiliary}).

\subremark{} A fundamental role will be played throughout the following
by the map $\co:\reals^{p}\setmap\reals^{p}$ defined by
\begin{equation}
\co(z)\defeq\begin{bmatrix}\ct(z)\\
\eq(z)
\end{bmatrix},\label{eq:co-map}
\end{equation}
where
\begin{equation}
\ct(z)\defeq\alpha_{\perp}^{\trans}\left[\fe_{0}(z)-\sum_{i=1}^{k-1}\ga_{i}(z)\right]\eqdef\alpha_{\perp}^{\trans}\srfn(z).\label{eq:ct-def}
\end{equation}
As developed in \subsecref{gjrt} below, if $\{z_{t}\}$ is generated
by a model from class $\cvar_{r}$, then $\ct(z_{t})$ and $\eq(z_{t})$
respectively provide a representation for the `common trends' and
`equilibrium errors' in the system. Since $\co$ is guaranteed
to be a homeomorphism under our assumptions (see \lemref{co-cts}
in \appref{auxiliary}), these maps provide an exhaustive description
of $z_{t}$. As a further consequence of the invertibility of $\co$,
we have $\im\eq=\reals^{r}$ and hence $\im\pi=\spn\alpha$, an $r$-dimensional
subspace of $\reals^{p}$; models in class $\cvar_{r}$ thus satisfy
the common row space condition introduced in \subsecref{framework}
above.

\subremark{} The case $r=0$ corresponds to $c=0$ and $\pi(z)=0$
for all $z\in\reals^{p}$. The results in this paper will be seen
to cover this case, if the appropriate adjustments are made by deleting
$\eq$ and $\cn$ from those objects in which they appear, so that
e.g.\ now $\b{\eq}=\b{\ga}$ and $\co=\ct$. (To keep this paper
to a manageable length, we do not treat this case explicitly in our
proofs.) At the other extreme, the case $r=p$ corresponds to a stationary
system with no common trends. For this reason, we generally restrict
the statements of our results to the case where $r\in\{0,\ldots,p-1\}$.

\subremark{}\label{subrem:cvar:shortmem} \ref{enu:cvar:jsr} is
our principal assumption on the stability of the weakly dependent
components of the system, which comprise the `equilibrium errors'
$\xi_{t}\defeq\cn+\eq(z_{t})$, and the `lagged differences' $\Delta\zeta_{j,t}=\sum_{i=j}^{k-1}\Delta\ga_{i}(z_{t-i+j})$
for $j\in\{1,\ldots,k-1\}$. In a linear cointegrated VAR, these reduce
to the familiar $\xi_{t}=\cn+\beta^{\trans}z_{t}$ and $\Delta\zeta_{j,t}=\sum_{i=j}^{k-1}\Gamma_{j}\Delta z_{t-i+j}$,
which are stationary (up to initialisation) by the GJRT. Though in
the present setting, $\xi_{t}$ and $\{\Delta\zeta_{j,t}\}_{j=1}^{k-1}$
will generally \emph{not} be stationary, \eqref{jsr} nonetheless
ensures that these are of strictly smaller order than $z_{t}$ itself,
just as in a linearly cointegrated system. As noted in \subsecref{linear:model}
below, in a linear VAR \ref{enu:cvar:jsr} reduces to precisely the
familiar conditions on the roots of the associated autoregressive
polynomial; whereas the form in which it is stated here is applicable
to models which depart so far from linearity that no useful notion
of `autoregressive roots' is available.

\subremark{}\label{subrem:cvar:lipschitz} The requirement that \eqref{gradient}
hold with $\smlnorm{\b{\beta}}\leq\abv b$ implies $\b{\eq}\compose\fe_{0}^{-1}$
is Lipschitz, with a Lipschitz constant that is bounded by $\abv b$;
the converse is also true (see \lemref{lipcondition} in \appref{auxiliary}).
Excepting cases where $\fe_{0}$ and $\b{\eq}$ interact in a very
particular way, this condition signals that our concern here is principally
with models in which the equilibrium errors are defined by a function
$\b{\eq}$ for which $\b{\eq}(z)=O(\smlnorm z)$ as $\smlnorm z\goesto\infty$,
i.e.\ which exhibits at most linear asymptotic growth. For example,
if $\fe_{0}=I_{p}$ and $\eq(z)=y-m(x)$, for $z=(y^{\trans},x^{\trans})^{\trans}$,
then \ref{enu:cvar:jsr} requires $m$ to be Lipschitz, and thus that
$m(x)=O(\smlnorm x)$ as $\smlnorm x\goesto\infty$, for $r\in\{0,\ldots,p-1\}$.
\end{rem}

Our assumptions on the data generating process may now be stated.

\assumpname{CVAR}
\begin{assumption}
\label{ass:cvar} Given $\vec{z}_{0}=(z_{0}^{\trans},\ldots,z_{-k+1}^{\trans})^{\trans}$,
$\{z_{t}\}_{t\in\integers}$ follows \eqref{nlVAR} for $t\geq1$,
where $(c,\{\fe_{i}\}_{i=0}^{k})$ belong to class $\cvar_{r}$.
\end{assumption}
For the purposes of much of this paper, the sequence of innovations
$\{u_{t}\}$ may be treated as a given non-random sequence in $\reals^{p}$,
since our results are either non-asymptotic, or take limits under
the assumption that the shocks cease at some finite horizon $\tau$.
However, it will occasionally be useful to indicate the implications
of our results in cases where the following holds.

\assumpname{ERR}
\begin{assumption}
\label{ass:coerr}$\{u_{t}\}_{t\in\integers}$ is a random sequence
in $\reals^{p}$, such that $\sup_{t\in\integers}\smlnorm{u_{t}}_{2+\delta_{u}}<\infty$
for some $\delta_{u}>0$, and
\begin{equation}
U_{n}(\lambda)\defeq n^{-1/2}\sum_{t=1}^{\smlfloor{n\lambda}}u_{t}\wkc U(\lambda)\label{eq:Unwkc}
\end{equation}
on $D[0,1]$, where $U$ is a $p$-dimensional Brownian motion with
positive-definite variance $\Sigma$.
\end{assumption}
\needspace{3cm}

\subsection{Relationship to the linear VAR}

\label{subsec:linear:model}

To help cast further light on the defining conditions of class $\cvar_{r}$,
we now show that when \eqref{nlVAR} is specialised to a linear VAR,
these conditions reduce to exactly the usual conditions for a linear
VAR to generate common $I(1)$ components, as per the GJRT (\citealp{Joh95},
Ch.\ 4). This provides a clear indication that our conditions on
the general, nonlinear VAR are not overly restrictive. By a \emph{linear
VAR}, we mean one in which $\fe_{i}(z)=\Phi_{i}z$ for some $\Phi_{i}\in\reals^{p\times p}$,
for $i\in\{0,\ldots,k\}$, so that \eqref{nlVAR} becomes
\begin{equation}
\Phi_{0}z_{t}=c+\sum_{i=1}^{k}\Phi_{i}z_{t-i}+u_{t}.\label{eq:linearVAR}
\end{equation}
Let $\Phi(\lambda)\defeq\Phi_{0}-\sum_{i=1}^{k}\Phi_{i}\lambda^{i}$
denote the associated autoregressive polynomial.
\begin{prop}
\label{prop:lin-rep-var}Suppose $(c,\{\fe_{i}\}_{i=0}^{k})$ is a
linear VAR, where for some $r\in\{0,\ldots,p\}$,
\begin{enumerate}[label=\ass{LIN.\arabic*}, leftmargin=1.75cm]
\item \label{enu:lin:homeo}$\Phi_{0}$ is invertible;
\item \label{enu:lin:rank}$c\in\spn\Phi(1)$ and $\rank\Phi(1)=r=p-q$;
and
\item \label{enu:lin:roots}$\Phi(\lambda)$ has $q$ roots at unity, and
all others outside the unit circle.
\end{enumerate}
Then $(c,\{f_{i}\}_{i=1}^{k})$ belongs to class $\cvar_{r}$, with:
$\ga_{j}(z)=\Gamma_{j}z$, where $\Gamma_{j}\defeq-\sum_{i=j+1}^{k}\Phi_{i}$
for $j\in\{1,\ldots,k-1\}$; and
\[
\co(z)=\begin{bmatrix}\ct(z)\\
\eq(z)
\end{bmatrix}=\begin{bmatrix}\alpha_{\perp}^{\trans}(\Phi_{0}-\sum_{i=1}^{k-1}\Gamma_{i})\\
\beta^{\trans}
\end{bmatrix}z\eqdef\begin{bmatrix}\alpha_{\perp}^{\trans}\srp\\
\beta^{\trans}
\end{bmatrix}z\eqdef\CO z
\]
for $\alpha,\beta\in\reals^{p\times r}$ such that $\alpha\beta^{\trans}=-\Phi(1)$.
\end{prop}
The proof of \propref{lin-rep-var} demonstrates that the factorisation
required of $\pi(z)$ in \ref{enu:cvar:rank} has its direct counterpart
in condition \ref{enu:lin:rank} on the rank of $\Phi(1)$, while
the high-level condition \ref{enu:cvar:jsr}, which here involves
only the matrices
\begin{align}
\b{\beta}^{\trans} & \defeq\begin{bmatrix}\beta^{\trans}\Phi_{0}^{-1} & 0\\
\b{\Gamma}\Phi_{0}^{-1} & D
\end{bmatrix}, & \b{\alpha} & \defeq\begin{bmatrix}\alpha & E^{\trans}\\
0 & I_{p(k-1)}
\end{bmatrix}, & \b{\Gamma} & \defeq[\Gamma_{i}]_{i=1}^{k-1}\defeq\left[\begin{smallmatrix}\Gamma_{1}\\
\vdots\\
\Gamma_{k-1}
\end{smallmatrix}\right],\label{eq:lin-boldbeta}
\end{align}
is equivalent to the familiar condition \ref{enu:lin:roots} on the
autoregressive roots.

\subsection{The structural nonlinear VAR($k$)}

\label{subsec:identification}

In macroeconomics, a conventional starting point for empirical work
is the \emph{structural} \emph{(linear) VAR} model, in which the observed
series $\{z_{t}\}$ are regarded as being generated by an underlying
$p$-dimensional i.i.d.\ sequence of \emph{structural shocks} $\{\err_{t}\}$,
whose elements are mutually orthogonal, and each of which have an
economic interpretation (as e.g.\ an aggregate supply shock, a monetary
policy shock, etc.).\footnote{In some treatments (e.g.\ \citealp{LMS17JoE}; \citealp{GMR20REStud})
the components of $\err_{t}$ are taken to be not merely uncorrelated,
but mutually independent. Under non-Gaussianity, this is known to
yield certain additional identifying restrictions. We emphasise that
in the present work, only orthogonality between the components of
$\err_{t}$ is maintained.} This perspective may also be adopted within the framework of \eqref{nlVAR},
reproduced here as
\begin{equation}
\fe_{0}(z_{t})=c+\sum_{i=1}^{k}\fe_{i}(z_{t-i})+u_{t},\tag{\ref{eq:nlVAR}}
\end{equation}
by supposing that there is a parametrisation $(c^{\err},\{\fe_{i}^{\err}\}_{i=0}^{k})$
of the model such that $\{z_{t}\}$ satisfies
\begin{equation}
\fe_{0}^{\err}(z_{t})=c^{\err}+\sum_{i=1}^{k}\fe_{i}^{\err}(z_{t-i})+\err_{t}.\label{eq:struct-form}
\end{equation}
We may suppose, as a convenient location and scale normalisation,
that each of the elements of $\err_{t}$ have mean zero and unit variance,
so that $\{\err_{t}\}$ is i.i.d.\ with $\expect\err_{t}=0$ and
$\expect\err_{t}\err_{t}^{\trans}=I_{p}$. For the purposes of this
section, we shall also require $\err_{t}$ to be continuously distributed,
with density $\den^{\err}$. (For simplicity of presentation, we also
treat the order $k$ of the VAR as known, though only a known finite
upper bound on the order of the VAR is required for the discussion
that follows.)

The question then arises as to whether, or to what extent, the structural
shocks \{$\err_{t}\}$ are identified by observation of $\{z_{t}\}$.
To phrase this question more precisely, we first clarify how the model
\eqref{nlVAR} is parametrised. For $i\in\{0,\ldots,k\}$, let $\fespc_{i}$
denote a collection of functions $\reals^{p}\setmap\reals^{p}$ to
which each $\fe_{i}$ respectively belongs. We also specialise \ref{ass:coerr}
to the case where $\{u_{t}\}$ is i.i.d., with a density $\den$ on
$\reals^{p}$ that belongs to some set $\denspc$, normalised to have
$\expect u_{t}=0$ and $\expect u_{t}u_{t}^{\trans}=I_{p}$. The parameter
space for \eqref{nlVAR} is then
\[
\set P\defeq\{(\tilde{c},\{\tilde{\fe}_{i}\}_{i=0}^{k},\tilde{\den})\mid\tilde{c}\in\reals\sep\tilde{\fe}_{i}\in\fespc_{i}\text{ for }i\in\{0,\ldots,k\}\sep\tilde{\den}\in\denspc\}.
\]
Under the regularity conditions given in \appref{identification},
each $(\tilde{c},\{\tilde{\fe}_{i}\}_{i=0}^{k},\tilde{\den})\in\set P$
completely specifies the density of $z_{t}$ conditional on $\vec{z}_{t-1}=(z_{t-1}^{\trans},\ldots,z_{t-k}^{\trans})^{\trans}$.
Following \citet[pp.~949f.]{Matz09Ecta}, we say that two points in
$\set P$ are \emph{observationally equivalent} if they imply exactly
the same conditional density function, and thus the same likelihood
for $\{z_{t}\}_{t=1}^{n}$, conditional on the initial values $\vec{z}_{0}$
(which as per \ref{ass:cvar} we take to be arbitrary, and thus uninformative
about the model parameters).

Let $\Upsilon\in\reals^{p\times p}$ be an orthogonal matrix, and
suppose that $\denspc$ is such that the density $\tilde{\den}$ of
$\tilde{u_{t}}\defeq\Upsilon\err_{t}$ also lies in $\denspc$. Then
it is evident from pre-multiplying \eqref{struct-form} through by
$\Upsilon$, to obtain
\[
\tilde{\fe}_{0}(z_{t})\defeq\Upsilon\fe_{0}^{\err}(z_{t})=\Upsilon c^{\err}+\sum_{i=1}^{k}\Upsilon\fe_{i}^{\err}(z_{t-i})+\Upsilon\err_{t}\eqdef\tilde{c}+\sum_{i=1}^{k}\tilde{\fe}_{i}(z_{t-i})+\tilde{u}_{t},
\]
that $(\tilde{c},\{\tilde{\fe}_{i}\}_{i=0}^{k},\tilde{\den})$ must
be observationally equivalent to $(c^{\err},\{\fe_{i}^{\err}\}_{i=0}^{k},\den^{\err})$.
Thus, as is familiar from the \emph{linear} structural VAR, the model
parameters and the structural shocks are at best identified up to
an unknown orthogonal transformation.

Remarkably, despite the much broader class of nonlinear transformations
permitted by \eqref{nlVAR}, such orthogonal transformations delineate
the \emph{only} loci of non-identification in the nonlinear VAR. By
\thmref{shockid} in \appref{identification}, under certain conditions
on the sets $\{\fespc_{i}\}_{i=0}^{k}$ and $\denspc$, and the parameters
generating $\{z_{t}\}$, there exists a $\den\in\denspc$ such that
the parameters $(c,\{\fe_{i}\}_{i=0}^{k},\den)$ of \eqref{nlVAR}
are observationally equivalent to those $(c^{\err},\{\fe_{i}^{\err}\}_{i=0}^{k},\den^{\err})$
of \eqref{struct-form}, if \emph{and only if} there is an orthogonal
matrix $\Upsilon\in\reals^{p\times p}$ such that
\begin{equation}
c=\Upsilon c^{\err}\text{ and }\fe_{i}(z)=\Upsilon\fe_{i}^{\err}(z)\sep\forall z\in\reals^{p}\sep i\in\{0,\ldots,k\},\label{eq:identified-set}
\end{equation}
and thus the shocks $\{u_{t}\}$ in \eqref{nlVAR} satisfy
\begin{equation}
u_{t}=\fe_{0}(z_{t})-c-\sum_{i=1}^{k}\fe_{i}(z_{t-i})=\Upsilon\left[\fe_{0}^{\err}(z_{t})-c^{\err}-\sum_{i=1}^{k}\fe_{i}^{\err}(z_{t-i})\right]=\Upsilon\err_{t}.\label{eq:u-to-epsilon}
\end{equation}
(For a discussion of the conditions under which the preceding holds,
see \appref{identification}: these are essentially weak smoothness
conditions on the elements of $\{\fespc_{i}\}_{i=0}^{k}$ and $\denspc$,
together with the requirement that $\{\fe_{i}^{\err}\}_{i=1}^{k}$be
such as to ensure sufficient dependence of the r.h.s.\ of the model
on lags of $z_{t}$.) Conversely, should \eqref{identified-set} fail
to hold, then there will be at least some realisations of $\{z_{t}\}$
for which the likelihoods of $(c,\{\fe_{i}\}_{i=0}^{k},\den)$ and
$(c^{\err},\{\fe_{i}^{\err}\}_{i=0}^{k},\den^{\err})$ will be distinct,
and so the data is to this extent informative about these two alternative
parametrisations of the model.\footnote{We do not\emph{ }claim, on the basis of \thmref{shockid}, that the
parameters of the model are consistently estimable up to an orthogonal
transformation. While it seems reasonable to suppose that consistent
nonparametric estimation of the model would be possible (under regularity
conditions) if $\{z_{t}\}$ is stationary and ergodic, the usual connection
between identification and consistent estimation is attenuated when
$\{z_{t}\}$ possesses some stochastic (or indeed, deterministic)
trends, because of the non-recurrence of those trends in higher dimensions
(see \citealp[Sec.~6]{Bing01Hdbk}; \citealp[p.~62]{GP13JoE}). Consistent
estimation of the model parameters (up to $\Upsilon$) would in such
cases likely require further restrictions on the functional form of
$\fe_{i}$.}

Accordingly, when subsequently discussing the problem of structural
shock identification, we shall suppose that the shocks $u_{t}$ appearing
in \eqref{nlVAR} are related to $\err_{t}$ in precisely the manner
of \eqref{u-to-epsilon}, yielding the \emph{structural nonlinear
VAR model}
\begin{align}
\fe_{0}(z_{t}) & =c+\sum_{i=1}^{k}\fe_{i}(z_{t-i})+u_{t} & u_{t} & =\Upsilon\err_{t}\label{eq:nlSVAR}
\end{align}
where $\Upsilon\in\reals^{p\times p}$ is an orthogonal matrix. In
other words, on the strength of \thmref{shockid}, we shall regard
observation of $\{z_{t}\}$ as being sufficient to identify the model
parameters up to (and only up to) $\Upsilon$, and thence the structural
shocks up to transformation by this same matrix. To identify the structural
shocks, we thus seek additional restrictions sufficient to pin down
(at least some of) the unknown elements of $\Upsilon$, such as will
be provided by the long-run identifying restrictions developed in
\subsecref{long-run-id} below.

\section{Granger--Johansen representation theory}

\label{sec:common-trends}

\label{subsec:gjrt}

The starting point for the analysis of the linear cointegrated VAR
is the Granger--Johansen representation theorem (GJRT), which decomposes
$z_{t}$ into the sum of: (i) an initial condition, (ii) a stochastic
trend, and (iii) a weakly dependent process (which is stationary if
$z_{t}$ is suitably initialised). This may be rendered in various
ways: to facilitate the comparison with \thmref{gjrt} below (from
which it follows in the linear case), we shall here write this for
a linear VAR satisfying \ref{enu:lin:homeo}\ass{--3}, as 
\begin{equation}
z_{t}=\begin{bmatrix}\alpha_{\perp}^{\trans}\srp\\
\beta^{\trans}
\end{bmatrix}^{-1}\left(\begin{bmatrix}\alpha_{\perp}^{\trans}\abv{\srfn}(\vec{z}_{0})+\alpha_{\perp}^{\trans}\sum_{s=1}^{t}u_{s}\\
-\mu
\end{bmatrix}+S_{\co}^{\trans}\b{\xi}_{t}\right),\label{eq:GJ-type-rep}
\end{equation}
where: (i) $\abv{\srfn}(\vec{z}_{0})$ captures the dependence on
the initial values $\vec{z}_{0}=(z_{0}^{\trans},\ldots,z_{-k+1}^{\trans})^{\trans}$;
(ii) $\alpha_{\perp}^{\trans}\sum_{s=1}^{t}u_{s}$ is a stochastic
trend of dimension $q=\rank\alpha_{\perp}=p-r$; and (iii)
\begin{equation}
\begin{bmatrix}\mu+\beta^{\trans}z_{t}\\
\Delta\b{\zeta}_{t}
\end{bmatrix}=\b{\xi}_{t}=(I_{p(k-1)+r}+\b{\beta}^{\trans}\b{\alpha})\b{\xi}_{t-1}+\b{\beta}^{\trans}\b u_{t}\label{eq:eqerr-linear}
\end{equation}
follows a VAR, where $\b{\alpha}$ and $\b{\beta}$ are as in \eqref{lin-boldbeta}
above. The stability of this VAR follows from \propref{lin-rep-var},
because membership of $\cvar_{r}$ implies, via \ref{enu:cvar:jsr},
that all the eigenvalues of $I_{p(k-1)+r}+\b{\beta}^{\trans}\b{\alpha}$
must lie inside the unit circle. (It may also be noted from \eqref{bz-in-linear}
above that, in this case, $\Delta\b{\zeta}_{t}$ is itself a linear
function of $\{\Delta z_{t-i}\}_{i=0}^{k-2}$.) $\b{\xi}_{t}$ may
thus be rendered stationary through an appropriate choice of the initial
conditions ($\vec{z}_{0}$ or, equivalently, $\b z_{0}$).

Before providing our extension of the GJRT to the general setting
of \eqref{nlVAR}, we first note that the nonlinear counterpart of
$\b{\xi}_{t}$ will continue to follow an autoregression of the form
\eqref{eqerr-linear}, but now with time-varying coefficients, in
particular with $\b{\beta}=\b{\beta}_{t}$ now depending on the level
of $z_{t}$. While such processes cannot be stationary, we still have
a well-defined notion of \emph{stability} for such processes, which
is sufficient to ensure that $\b{\xi}_{t}$ remains of strictly smaller
stochastic order than the common trends.
\begin{defn}
\label{def:expstable}Suppose that $\{w_{t}\}_{t\in\naturals}$ follows
the time-varying VAR,
\begin{equation}
w_{t}=c_{t}+A_{t}w_{t-1}+B_{t}v_{t},\label{eq:wt-time-varying}
\end{equation}
where $A_{t}\in\set A\subset\reals^{m\times m}$, $B_{t}\in\set B\subset\reals^{m\times\ell}$,
and $c_{t}\in\set C\subset\reals^{m}$ for all $t\in\naturals$, from
some given $w_{0}$. We say that $\{w_{t}\}$ is \emph{exponentially
stable} if $\set B$ and $\set C$ are bounded, and there exists a
$C<\infty$ and a $\rho\in[0,1)$ such that 
\[
\norm{\prod_{s=1}^{t}A_{s}}\leq C\rho^{t}\sep\forall t\in\naturals.
\]
\end{defn}
A sufficient condition for exponential stability is that $\rho_{\jsr}(\set A)<1$
(e.g.\ \citealp{Jungers09}, Cor.~1.1), which motivates the restriction
on the JSR that appears in \ref{enu:cvar:jsr}. Notable consequences
are that if $v_{t}=0$ for all $t\geq\tau$, then $w_{t}\goesto0$,
while if $\{v_{t}\}$ and $w_{0}$ have uniformly bounded $2+\delta$
moments, then $\max_{1\leq t\leq n}\smlnorm{w_{t}}=o_{p}(n^{1/2})$.
(For this last, see Lemma~A.1 in \citealp{DMW22}.)

We now state our counterpart of the Granger--Johansen representation
theorem for the model \eqref{nlVAR}, which constitutes the main result
of this paper.
\begin{thm}
\label{thm:gjrt}Suppose \ref{ass:cvar} holds. Then $\co:\reals^{p}\setmap\reals^{p}$
in \eqref{co-map} is a homeomorphism, and
\begin{equation}
z_{t}=\co^{-1}\left(\begin{bmatrix}\alpha_{\perp}^{\trans}\abv{\srfn}(\vec{z}_{0})+\alpha_{\perp}^{\trans}\sum_{s=1}^{t}u_{s}\\
-\mu
\end{bmatrix}+S_{\co}^{\trans}\b{\xi}_{t}\right),\label{eq:gjrt-z}
\end{equation}
where 
\[
\abv{\srfn}(\vec{z}_{s})\defeq\fe_{0}(z_{s})-\sum_{i=1}^{k-1}\ga_{i}(z_{s-i}),
\]
$S_{\co}^{\trans}$ is the $p\times[p(k-1)+r]$ matrix given by
\[
S_{\co}^{\trans}\defeq\begin{bmatrix}-\alpha_{\perp}^{\trans} & 0\\
0 & I_{r}
\end{bmatrix}\begin{bmatrix}0_{p\times r} & I_{p} & \cdots & I_{p}\\
I_{r} & 0_{r\times p} & \cdots & 0_{r\times p}
\end{bmatrix},
\]
and 
\[
\b{\xi}_{t}\defeq\b{\mu}+\b{\eq}(z_{t})+\b D_{0}\b z_{t}=\begin{bmatrix}\mu+\eq(z_{t})\\
\Delta\b{\zeta}_{t}
\end{bmatrix}\eqdef\begin{bmatrix}\xi_{t}\\
\Delta\b{\zeta}_{t}
\end{bmatrix}
\]
follows the exponentially stable VAR
\begin{equation}
\b{\xi}_{t}=(I_{p(k-1)+r}+\b{\beta}_{t}^{\trans}\b{\alpha})\b{\xi}_{t-1}+\b{\beta}_{t}^{\trans}\b u_{t}\label{eq:gjrt-xi}
\end{equation}
with $\b{\beta}_{t}\in{\cal B}$ for all $t\in\naturals$. 
\end{thm}
\begin{rem}
\label{rem:-gjrt}\subremark{}\label{subrem:gjrt-lin} That the preceding
specialises to \eqref{GJ-type-rep} in the linear case follows from
the fact that $\ct(z)=\alpha_{\perp}^{\trans}\srp z$ and $\eq(z)=\beta^{\trans}z$
by \propref{lin-rep-var}, and that ${\cal B}=\{\b{\beta}\}$ for
$\b{\beta}$ as in \eqref{lin-boldbeta}. To illustrate more clearly
the connection between \eqref{GJ-type-rep} and conventional statements
of the GJRT, we note that the part of the r.h.s.\ of \eqref{GJ-type-rep}
that depends on the common stochastic trends may be written as
\begin{equation}
\begin{bmatrix}\alpha_{\perp}^{\trans}\srp\\
\beta^{\trans}
\end{bmatrix}^{-1}\begin{bmatrix}\alpha_{\perp}^{\trans}\sum_{s=1}^{t}u_{s}\\
0_{r\times1}
\end{bmatrix}=\beta_{\perp}(\alpha_{\perp}^{\trans}\srp\beta_{\perp})^{-1}\alpha_{\perp}^{\trans}\sum_{s=1}^{t}u_{s}\eqdef P_{\beta_{\perp}}\sum_{s=1}^{t}u_{s}\label{eq:gjrt-linear}
\end{equation}
by \lemref{lin-co-map} in \appref{examples}, which agrees precisely
with \citet[Ch.~4]{Joh95}.

\subremark{} Suppose that $\{u_{t}\}$ satisfies \assref{coerr}.
Then it follows by Lemma~A.1 in \citet{DMW22} that $\max_{1\leq t\leq n}\smlnorm{\b{\xi}_{t}}=o_{p}(n^{1/2})$,
and so is strictly of smaller order than $\sum_{s=1}^{\smlfloor{n\lambda}}u_{s}$.
This is crucial for obtaining the limiting distribution of the standardised
process $n^{-1/2}z_{\smlfloor{n\lambda}}$: though here further assumptions
are required, because of the presence of the nonlinear map $\co$
in \eqref{gjrt-z}. Results of this kind, which were obtained in the
setting of the CKSVAR by \citet{DMW22}, are the subject of a companion
paper to the present work.

\subremark{} Applying $\co$ to both sides of \eqref{gjrt-z}, we
obtain
\[
\co(z_{t})=\begin{bmatrix}\ct(z_{t})\\
\eq(z_{t})
\end{bmatrix}=\begin{bmatrix}\alpha_{\perp}^{\trans}\abv{\srfn}(\vec{z}_{0})+\alpha_{\perp}^{\trans}\sum_{s=1}^{t}u_{s}-\alpha_{\perp}^{\trans}\sum_{i=1}^{k-1}\Delta\zeta_{it}\\
\xi_{t}-\mu
\end{bmatrix}.
\]
Asymptotically, under \assref{coerr}, the dominant component of $\ct(z_{t})$
would be the $q$ common trends $\alpha_{\perp}^{\trans}\sum_{s=1}^{t}u_{s}$,
which would in turn dominate $\eq(z_{t})$, since the latter equals
the first $r$ components of the stable autoregressive process $\b{\xi}_{t}$.
In this sense, the action of $\co$ separates $z_{t}$ into its `common
trend' and `equilibrium error' components. Since $\eq(z_{t})$
is purged of these common trends -- while being restricted by the
requirement that $\co$ be a homeomorphism -- it may be said to provide
a representation of the `nonlinear cointegrating relations' that
exist between the elements of $z_{t}$.
\end{rem}

\section{Consequences of the representation theorem}

\label{sec:consequences}

\subsection{Attractor spaces}

\label{subsec:attractors}

A first application of \thmref{gjrt} is to verify the stability of
the (non-stochastic) steady state solutions to \eqref{nlVAR}. By
considering \eqref{VECM-first} when $z_{t-1}=\cdots=z_{t-k}=z$ for
some $z\in\reals^{p}$ and $u_{t}=0$, we see that for $z$ to be
a steady state equilibrium, it must satisfy
\[
0=c+\pi(z)=\alpha[\mu+\eq(z)]
\]
or equivalently $\eq(z)=-\mu$, since $\rank\alpha=r$ and $\eq:\reals^{p}\setmap\reals^{r}$.
Thus the set of steady state equilibria is given by the $q$-dimensional
manifold
\begin{equation}
\set M_{\mu}\defeq\{z\mid\pi(z)=-c\}=\{z\in\reals^{p}\mid\eq(z)=-\mu\}=\co^{-1}(\reals^{q}\times\{-\mu\})\label{eq:ctspc-mu}
\end{equation}
where the fact that $\ctspc_{\mu}$ is indeed a $q$-dimensional manifold
follows from the final equality, since $\co:\reals^{p}\setmap\reals^{p}$
is a homeomorphism.

For the purposes of analysing the stability of $\ctspc_{\mu}$, we
shall suppose that the shocks $\{u_{t}\}$ are set to zero after some
$\tau\in\naturals$: that is, $u_{s}=0$ for all $s\geq\tau+1$, and
then consider how $z_{t}$ evolves from time $\tau$ forward. In other
words, we shall fix the state of the model at time $\tau-1$ at
\begin{equation}
\vec{z}_{\tau-1}=(z_{\tau-1}^{\trans},z_{\tau-2}^{\trans},\ldots,z_{\tau-k}^{\trans})^{\trans}=(\state_{(1)}^{\trans},\state_{(2)}^{\trans},\ldots,\state_{(k)}^{\trans})^{\trans}\eqdef\state;\label{eq:tau-general-init}
\end{equation}
allow one final shock $u_{\tau}=u\in\reals^{p}$ to occur at time
$\tau$; and then evaluate the (non-stochastic) limit of $z_{t}=z_{t}(u;\state)$
as $t\goesto\infty$. We aim to show that $\ctspc_{\mu}$ is strictly
stable, in the sense of the following.
\begin{defn}
\label{def:stability}${\cal S}\subset\reals^{p}$ is \emph{stable}
if for every $(u,\state)\in\reals^{p}\times\reals^{kp}$, there exists
a $z_{\infty}\in{\cal S}$ such that $z_{t}(u,\state)\goesto z_{\infty}\in{\cal S}$.
$z_{\infty}\in\reals^{p}$ is a \emph{non-trivial attractor} if for
every $\state\in\reals^{kp}$, there exists a non-empty $\set U(z_{\infty};\state)\subset\reals^{p}$
such that $z_{t}(u,\state)\goesto z_{\infty}$ for all $u\in\set U(z_{\infty};\state)$;
we term $\set U(z_{\infty};\state)$ the \emph{domain of attraction}
for $z_{\infty}$, given $\state$. ${\cal S}$ is \emph{strictly
stable} if it is stable and contains only non-trivial attractors.
\end{defn}

In other words, if $\ctspc_{\mu}$ is strictly stable, then whatever
the current state $\vec{z}_{\tau-1}=\state$ of the model and the
given limiting $z_{\infty}\in\ctspc_{\mu}$, there is always a value
for the $\tau$-dated shock $u_{\tau}$ that would lead $z_{t}$ to
converge to $z_{\infty}$; every element in $\ctspc_{\mu}$ may thus
ultimately be `reached' from $\vec{z}_{\tau-1}$. (Note that this
is not so trivial a matter as choosing $u_{\tau}=u$ such that $z_{\tau}=z_{\infty}$
immediately, since what is required is that $z_{t}$ converge to a
\emph{steady state equilibrium} at $z_{\infty}$.) Since $\{\b{\xi}_{t}\}$
is exponentially stable by \thmref{gjrt}, it follows that $\b{\xi}_{t}\goesto0$
as $t\goesto\infty$ (see the discussion following \defref{expstable}
above). Hence, noting that $\sum_{s=\tau}^{t}u_{s}=u_{\tau}=u$, we
have by that result that
\[
z_{t}(u;\state)=\co^{-1}\left(\begin{bmatrix}\alpha_{\perp}^{\trans}\abv{\srfn}(\state)+\alpha_{\perp}^{\trans}u\\
-\mu
\end{bmatrix}+S_{\co}^{\trans}\b{\xi}_{t}\right)\goesto\co^{-1}\left(\begin{bmatrix}\alpha_{\perp}^{\trans}[\abv{\srfn}(\state)+u]\\
-\mu
\end{bmatrix}\right)\eqdef z_{\infty}(u;\state)\in\ctspc_{\mu}.
\]
Because the r.h.s.\ depends on $\state$ only through $\abv{\srfn}(\state)$,
and as $u\in\reals^{p}$ varies $\alpha_{\perp}^{\trans}[\abv{\srfn}(\state)+u]$
ranges freely over $\reals^{q}$, we can induce $z_{\infty}(u;\state)$
to take any desired value in $\ctspc_{\mu}$. We thus obtain the following.

\begin{thm}
\label{thm:stab}Suppose \assref{cvar} holds. Then
\begin{enumerate}
\item \label{enu:stab:cvg}for every $(u,\state)\in\reals^{p}\times\reals^{kp}$,
as $t\goesto\infty$
\begin{equation}
z_{t}(u;\state)\goesto z_{\infty}(u;\state)=\co^{-1}\begin{bmatrix}\alpha_{\perp}^{\trans}[\abv{\srfn}(\state)+u]\\
-\mu
\end{bmatrix}\in\ctspc_{\mu};\label{eq:zt-det-lim}
\end{equation}
\end{enumerate}
\begin{enumerate}[resume]
\item \label{enu:stab:dom}for any $(z,\state)\in\ctspc_{\mu}\times\reals^{kp}$
the $r$-dimensional affine subspace
\begin{equation}
\set U(z;\state)\defeq(\spn\alpha)+[\srfn(z)-\abv{\srfn}(\state)]\label{eq:domain-of-attraction}
\end{equation}
is the domain of attraction for $z$, given $\state$;
\end{enumerate}
and hence $\ctspc_{\mu}$ is strictly stable.
\end{thm}
\begin{rem}
\subremark{} In view of the preceding, we are justified in referring
to $\ctspc_{\mu}$ as the \emph{attractor space} for $\{z_{t}\}$.

\subremark{} If the model is in a steady state equilibrium at some
$z\in\ctspc_{\mu}$, immediately prior to the incidence of $u_{\tau}$,
so that $\state_{(i)}=z\in\ctspc_{\mu}$ for all $i\in\{1,\ldots,k-1\}$,
then
\[
\abv{\srfn}(\state)=\fe_{0}(z)-\sum_{i=1}^{k-1}\ga_{i}(z)=\srfn(z)
\]
and in this case the domain of attraction \eqref{domain-of-attraction}
reduces to $\spn\alpha$.

\subremark{} In the linear VAR, $\eq(z)=\beta^{\trans}z$, and thus
the attractor space simplifies to the $q$-dimensional affine subspace
\[
\ctspc_{\mu}=\{z\in\reals^{p}\mid\beta^{\trans}z=-\mu\},
\]
while the domain of attraction for $z_{\infty}\in\ctspc_{\mu}$ is
the $r$-dimensional affine subspace
\[
\set U(z_{\infty};\state)=(\spn\alpha)+[\srp z_{\infty}-\abv{\srfn}(\state)].
\]
\end{rem}

\subsection{Long-run identifying restrictions}

\label{subsec:long-run-id}

As a direct consequence of the preceding, we obtain the following
characterisation of the linear combinations of the shocks $u_{t}$
in \eqref{nlVAR} that do \emph{not} have permanent effects.
\begin{cor}
\label{cor:subspace-noeffect}Suppose \assref{cvar} holds. Then for
every $\state\in\reals^{p}$,
\[
\{u\in\reals^{p}\mid z_{\infty}(u;\state)=z_{\infty}(0;\state)\}=\spn\alpha
\]
i.e.~a shock has no permanent effect on $z_{t}$, if and only if
it lies in $\spn\alpha$.
\end{cor}

The significance of this result may be explained as follows. Suppose
that, as in \eqref{nlSVAR} above, we have the nonlinear structural
VAR
\begin{align*}
\fe_{0}(z_{t}) & =c+\sum_{i=1}^{k}\fe_{i}(z_{t-i})+u_{t} & u_{t} & =\Upsilon\err_{t}\tag{\ref{eq:nlSVAR}}
\end{align*}
belonging to class $\cvar_{r}$. As noted in \subsecref{identification},
the parameters of this model are identified up to the unknown orthogonal
matrix $\Upsilon\in\reals^{p\times p}$, about which the data is entirely
uninformative, and thus we seek additional restrictions that would
pin down (at least some of) the elements of $\Upsilon$. To that end,
suppose we partition the structural shocks as
\[
\err_{t}=(\err_{(1),t}^{\trans},\err_{(2),t}^{\trans})^{\trans}
\]
where $\err_{(1),t}$ takes values in $\reals^{m}$, and collects
the $m\in\{1,\ldots,r\}$ structural shocks that are regarded as having
\emph{no} permanent effect on $z_{t}$. Then \corref{subspace-noeffect}
implies that
\begin{equation}
\Upsilon\begin{bmatrix}I_{m}\\
0_{(p-m)\times m}
\end{bmatrix}\subset\spn\alpha\implies\alpha_{\perp}^{\trans}\Upsilon\begin{bmatrix}I_{m}\\
0_{(p-m)\times m}
\end{bmatrix}=0,\label{eq:lr-ident}
\end{equation}
which provides $qm$ `long-run identifying restrictions' on $\Upsilon$.\footnote{Note that \corref{subspace-noeffect}, or equivalently \eqref{lr-ident},
implies that at most $r$ elements of $\err_{t}$ may be such that
the corresponding column of $\Upsilon$ lies in $\spn\alpha$, i.e.\ \emph{up
to} $r$ structural shocks may have purely transitory effects. As
is familiar from the linear structural VAR, there is nothing here
to preclude e.g.\ \emph{all} shocks from having permanent effects;
the number (and identities) of the $m\in\{0,\ldots,r\}$ structural
shocks having purely transitory effects will thus depend on the identifying
conditions that define those shocks.}

As discussed further in \secref{without-the-crsc} below, the availability
of these long-run identifying restrictions is largely a consequence
of the common row space condition, and provides one of the principal
motivations for maintaining that condition in our model. Remarkably,
despite the nonlinearity permitted by \eqref{nlSVAR}, the identifying
restrictions \eqref{lr-ident} take exactly the same form as in a
linear SVAR (see \citealp{KL17book}, Sec.~10.2).

\subsection{Long-run multipliers}

\label{sec:lr-multipliers}

Having obtained a characterisation of the attractor space $\ctspc_{\mu}$
for $\{z_{t}\}$, we may also say something more about the permanent
effect of a shock at time $t=\tau$. The \emph{limiting impulse responses}
or \emph{long-run multipliers} can be computed by differentiating
$z_{\infty}(u;\state)$ with respect to $u$, i.e.~by computing the
Jacobian
\begin{equation}
\partial_{u}\left.z_{\infty}(u;\state)\right|_{u=0}=\partial_{v}\left.\co^{-1}\begin{bmatrix}\alpha_{\perp}^{\trans}\abv{\srfn}(\state)+v\\
-\mu
\end{bmatrix}\right|_{v=0}\alpha_{\perp}^{\trans}\eqdef\nabla(\state)\alpha_{\perp}^{\trans}.\label{eq:multipliers}
\end{equation}
In view of \eqref{nlSVAR}, postmultiplying this by $\Upsilon$ yields
the long-run multipliers with respect to the structural shocks $\err_{t}$,
at time $t=\tau$.

One concern we might have here, in the general nonlinear case, is
the apparent dependence of these long-run multipliers on $\vec{z}_{t-1}=\state\in\reals^{kp}$,
which potentially ranges over a very `large' ($kp$-dimensional)
space. However, it will be noted that $z_{\infty}(u;\state)$ only
depends on $\state$ through $\alpha_{\perp}^{\trans}\abv{\srfn}(\state)$:
and as shown in the proof of the next result, it is always possible
to find a $z_{\state}\in\ctspc_{\mu}$ such that $\ct(z_{\state})=\alpha_{\perp}^{\trans}\abv{\srfn}(\state)$.
Therefore, for the purposes of computing the full set of possible
long-run effects of shocks in the model, it suffices to consider cases
in which the model is in a steady state equilibrium (i.e.\ where
$z_{\tau-1}=\dots=z_{\tau-k}=z$ for some $z\in\ctspc_{\mu}$) immediately
prior to the incidence of the shock.

Because $\co^{-1}$ is not necessarily differentiable everywhere,
the long-run multipliers as defined in \eqref{multipliers} may not
exist for every $z\in\ctspc_{\mu}$. However, in most cases of interest,
the subset $N\subset\ctspc_{\mu}$ at which this occurs will be exceptionally
small, corresponding e.g.\ in the case of piecewise affine models
(see \subsecref{piecewise-affine} below) to points exactly on the
boundary between two regimes. It is also possible that the Jacobian
of $\co^{-1}$ may fail to be invertible at exceptional points: though
it must be invertible almost everywhere, since $\co$ is invertible.
\begin{thm}
\label{thm:multipliers}Suppose \assref{cvar} holds. Let $N\subset\ctspc_{\mu}$
denote the set of all $z\in\ctspc_{\mu}$ at which
\[
v\elmap\co^{-1}\begin{bmatrix}\ct(z)+v\\
-\mu
\end{bmatrix}
\]
is \emph{not} differentiable with respect to $v$, when $v=0$, and
$N_{0}$ denote the union of $N$ with the set of $z\in\ctspc_{\mu}\backslash N$
at which the Jacobian of the preceding is non-invertible. Then 
\begin{enumerate}
\item \label{enu:mult:set}the full set of long-run multiplier matrices
is given by
\begin{align}
\Theta_{\infty}(z)\defeq\partial_{u}\left.\co^{-1}\begin{bmatrix}\ct(z)+\alpha_{\perp}^{\trans}u\\
-\mu
\end{bmatrix}\right|_{u=0} & =\partial_{v}\left.\co^{-1}\begin{bmatrix}\ct(z)+v\\
-\mu
\end{bmatrix}\right|_{v=0}\alpha_{\perp}^{\trans}\label{eq:lrmult}
\end{align}
for $z\in\ctspc_{\mu}\backslash N$; and
\item \label{enu:mult:rank}$\rank\Theta_{\infty}(z)=q$ for every $z\in\ctspc_{\mu}\backslash N_{0}$.
\end{enumerate}
\end{thm}
\begin{rem}
\subremark{} In the linear VAR, it follows by \propref{lin-rep-var}
and \lemref{lin-co-map} that
\[
\partial_{u}\left.z_{\infty}(u;\state)\right|_{u=0}=\begin{bmatrix}\alpha_{\perp}^{\trans}\srp\\
\beta^{\trans}
\end{bmatrix}^{-1}\partial_{u}\left.\begin{bmatrix}\alpha_{\perp}^{\trans}[\abv{\srfn}(\state)+u]\\
-\mu
\end{bmatrix}\right|_{u=0}=\begin{bmatrix}\alpha_{\perp}^{\trans}\srp\\
\beta^{\trans}
\end{bmatrix}^{-1}\begin{bmatrix}\alpha_{\perp}^{\trans}\\
0
\end{bmatrix}=\beta_{\perp}(\alpha_{\perp}^{\trans}\srp\beta_{\perp})^{-1}\alpha_{\perp}^{\trans}
\]
consistent with \eqref{gjrt-linear} above. In this case, the limiting
IRF is invariant to the state of the process prior to the incidence
of the final shock\@.

\subremark{} In nonlinear VARs such as \eqref{nlVAR}, impulse responses
with respect to a shock incident at time $\tau$ are well known to
be dependent on both the current state $\vec{z}_{\tau-1}$ of the
model, and on the shocks that occur subsequent to time $\tau$. In
computing the long-run multipliers above, we have allowed for dependence
on the current state, but deliberately forced $u_{t}=0$ for $t\geq\tau+1$.
This not only simplifies the analysis but, we would argue, provides
a reasonable basis on which to determine whether the structural VAR
permits certain (small) shocks to have permanent effects on the model
variables, particularly in a nonstationary model of this kind, where
subsequent shocks may push $z_{t}$ in any region of the state space
(though it will, in some sense, still remain `attracted' to $\ctspc_{\mu}$).

Nonetheless, let us suppose that, following \citet{KPP96JoE}, one
is also interested in `generalised' impulse responses that are computed
conditional on $\vec{z}_{\tau-1}=\state$, but which average over
all (potential) histories of the shocks subsequent to $t=\tau$.
Then while we could not hope to obtain as clean a characterisation
of the limiting IRF as is given by the preceding theorem, the fact
that the major implication of \eqref{multipliers}, that
\[
\ker\partial_{u}\left.z_{\infty}(u;\state)\right|_{u=0}=\alpha
\]
is invariant to the state of the process suggests that this property
would continue to hold for the generalised impulse responses. However,
because of the possible long-range dependence of $\b{\xi}_{t}$ on
past shocks, via $\b{\beta}_{t}$ (and hence $z_{t}$), these calculations
are far from straightforward, and are therefore deferred to future
work.
\end{rem}

\section{The common row space condition}

\label{sec:without-the-crsc}

As introduced in \subsecref{framework} above, a key characteristic
of models in class $\cvar_{r}$ is the common row space condition
(CRSC), that $\im\pi=\{\pi(z)\mid z\in\reals^{p}\}$ is an $r$-dimensional
linear subspace of $\reals^{p}$; we noted there that no such restriction
is imposed on $\ker\pi$. By contrast, the existing literature on
nonlinear VECM models effectively reverses this condition by requiring
that $\ker\pi$ be a $q$-dimensional subspace of $\reals^{p}$, without
necessarily restricting $\im\pi$ (see in particular \citealp{KR10JoE}).
Loosely speaking, if we suppose that $\pi$ may be decomposed, for
each $z\in\reals^{p}$, as
\begin{equation}
\pi(z)=\alpha(z)\beta(z)^{\trans}\label{eq:factor}
\end{equation}
where $\alpha,\beta:\reals^{p}\setmap\reals^{r}$, then the CRSC permits
taking $\alpha(z)=\alpha$, if $\beta(z)$ is appropriately normalised,
but leaves $\beta(z)$ otherwise unrestricted;\footnote{If $\beta(z)$ is instead normalised in some other way, such as e.g.\ $\beta(z)^{\trans}=[I_{r},-A(z)]$,
then $\alpha(z)$ may indeed vary with $z$. In this sense, it is
not quite correct to regard the CRSC as requiring the loadings on
the equilibrium errors to be \emph{constant}, but rather that those
loadings should lie in a fixed $r$-dimensional linear subspace; there
may be nonlinear adjustment towards equilibrium, provided that it
occurs along this subspace.} whereas in the nonlinear VECM literature, $\beta(z)=\beta$ is constant,
with the result that there is a globally linear cointegrating space.
The representation theory developed in the nonlinear VECM literature
is thus strictly complementary to the present work. (Note that even
the case where $\beta(z)=\beta$ in our framework is \emph{not} encompassed
by that literature, because we allow the dynamics of the implied VECM
\eqref{VECM-first} to depend on the \emph{level} of $z_{t}$, whereas
that literature requires that these be governed entirely by the stationary
equilibrium errors $\beta^{\trans}z_{t}$ and differences $\Delta z_{t}$.)

Recall that in the linear structural VAR model, common stochastic
trends arise because some subset of the structural shocks (say, $q$
in number) have permanent effects on $z_{t}$, while the remaining
$r=p-q$ shocks have only transitory effects. This provides not only
an underpinning economic explanation for the patterns of long-run
co-movement present in the data, but also a fruitful source of long-run
identifying restrictions (following \citealp{BlanchardQuah89}, and
\citealp{KPSW91AER}). The results developed in Sections~\ref{sec:common-trends}
and \ref{sec:consequences} above show that these properties are preserved
when we generalise the cointegrated linear SVAR to a nonlinear SVAR
\eqref{nlSVAR} belonging to class $\cvar_{r}$. However, as the example
developed in this section illustrates, these properties are highly
sensitive to departures from the CRSC, even when the assumption of
a globally linear cointegrating space is maintained (i.e.\ when $\beta(z)=\beta$
for all $z\in\reals^{p}$ in \eqref{factor} above). In models where
the CRSC fails, it will generally be the case that the \emph{direction}
in which an impulse $u_{t}=\delta$ has no permanent effect on $z_{t}$
will vary with the \emph{magnitude} $\delta$. Thus, if we want to
work with nonlinear SVAR models in which structural shocks may be
distinguished on the basis of their long-run effects -- and in which,
as a corollary, common stochastic trends arise because only a subset
of these shocks have permanent effects -- then it would appear that
the CRSC cannot be easily dispensed with.

To illustrate the consequences of a failure of the CRSC, we consider
a nonlinear VAR(1), specified in VECM form as
\[
\Delta z_{t}=a(\beta^{\trans}z_{t-1})\beta^{\trans}z_{t-1}+u_{t}
\]
where for $\Lambda:\reals^{r}\setmap[0,1]$ a smooth function satisfying
$\Lambda(0)=0$ and $\lim_{\smlnorm{\xi}\goesto\infty}\Lambda(\xi)=1$,
and $\alpha,\tilde{\alpha}\in\reals^{p\times r}$, each having rank
$r$, 
\[
a(\xi)\defeq\Lambda(\xi)\tilde{\alpha}+[1-\Lambda(\xi)]\alpha,
\]
so that the model is one in which the kernel of $\pi(z)=a(\beta^{\trans}z)\beta^{\trans}$
is a fixed $q$-dimensional subspace (given by $\spn\beta$), but
the image of $\pi$ is not a linear subspace. (If we further suppose
that the eigenvalues of $I_{r}+\beta^{\trans}\alpha$ lie strictly
inside the unit circle, then the model satisfies conditions (A.2)
and (A.3) of \citet{KR10JoE}, who provide a GJRT for these models.)
This model may be regarded as smoothly combining two regimes: an `inner'
or `near equilibrium' regime in which
\begin{align*}
\Delta z_{t} & =\tilde{\alpha}\beta^{\trans}z_{t-1}+u_{t} & u_{t} & =\Upsilon\err_{t}
\end{align*}
and an `outer' or `far from equilibrium' regime where $\tilde{\alpha}$
is replaced by $\alpha$. If the eigenvalues of $I_{r}+\beta^{\trans}\tilde{\alpha}$
associated to the inner regime are also strictly inside the unit circle,
then $\ctspc_{0}=\beta_{\perp}$ is a strictly stable attractor in
the sense of \defref{stability} above.

Now suppose we were to repeat the analysis of \subsecref{attractors}
for this model: i.e.\ fixing the state $z_{\tau-1}=\state\in\reals^{p}$
of the model in time $\tau$, having the model impacted by a final
shock $u_{\tau}=u$, and then computing $\lim_{t\goesto\infty}z_{t}(u;\state)=z_{\infty}(u;\state)$
in the absence of any further shocks (so $u_{t}=0$ for all $t\geq\tau+1$).
We may then ask: what values of $u$ would leave $z_{\infty}(u;\state)$
unchanged? That is, we would like to determine the set
\[
\{u\in\reals^{p}\mid z_{\infty}(u;\state)=z_{\infty}(0;\state)\}.
\]
Let us further suppose that $\beta^{\trans}z_{\tau-1}=\beta^{\trans}\state=0$,
so that $z_{\tau-1}\in\ctspc_{0}$, and the model is in equilibrium
prior to the incidence of $u_{\tau}$. Locally to $\ctspc_{0}$, we
would expect shocks in the direction of $\spn\tilde{\alpha}$ to have
no permanent effects; whereas further from $\ctspc_{0}$, shocks in
directions lying progressively closer to $\spn\alpha$ should have
no permanent effects.

\begin{figure}
\begin{centering}
\includegraphics[viewport=150bp 180bp 692bp 415bp,clip,scale=0.8]{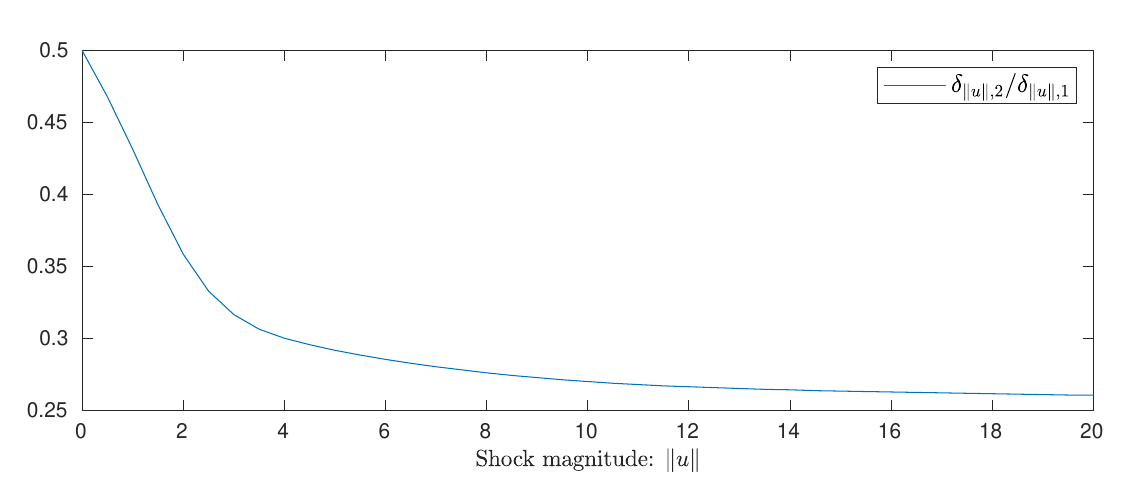}
\par\end{centering}
\caption{Directions of shocks having only transitory effects}

\label{fig:transitory}
\end{figure}

\figref{transitory} illustrates this is indeed the case for a bivariate
($p=2$) model with
\begin{align*}
\tilde{\alpha} & =-\begin{bmatrix}1\\
0.5
\end{bmatrix} & \alpha & =-\begin{bmatrix}1\\
0.25
\end{bmatrix} & \beta & =\begin{bmatrix}1\\
-1
\end{bmatrix},
\end{align*}
and $\Lambda(\xi)=2\smlabs{\Phi(\xi)-0.5}$, where $\Phi$ denotes
the Gaussian cdf. (Without loss of generality, we take $\state=0$.)
For values of $\smlnorm u$ between $0$ and $10$, the figure reports
the (unique) direction $\delta_{\smlnorm u}$ such that $u_{\tau}=\smlnorm u\delta_{\smlnorm u}$
has no permanent effect on $z_{t}$; for the sake of comparability
with $\tilde{\alpha}$ and $\alpha$, both of which have unit first
element, this is reported in terms of the ratio $\delta_{\smlnorm u,2}/\delta_{\smlnorm u,1}$.
We see here that $\delta_{\smlnorm u,2}/\delta_{\smlnorm u,1}$ takes
the value $0.5=\tilde{\alpha}_{2}$ when $\smlnorm u=0$, and tends
towards $0.25=\alpha_{1}$ as $\smlnorm u$ grows (equalling $0.26$
when $\smlnorm u=20$). As a consequence
\[
\Intsect_{u\in\reals^{p}}\spn\delta_{\smlnorm u}=\emptyset,
\]
and as such there is \emph{no} direction along which the shocks $u_{t}$
will only have transient effects. There is thus no possibility of
discriminating between the underlying structural shocks $\err_{t}$
according to whether not they have permanent effects: each must have
some such effect, to an extent that varies with the magnitude of the
shock.

\section{Application to piecewise affine VARs}

\label{subsec:piecewise-affine}

Finally, we introduce a class of regime-switching models in which
the conditions required for our results may be verified relatively
straightforwardly. Suppose that for each $i\in\{0,\ldots,k\}$,
\begin{equation}
\fe_{i}(z)=\sum_{\ell=1}^{L}\indic\{z\in\set Z^{(\ell)}\}(\bar{\phi}_{i}^{(\ell)}+\Phi_{i}^{(\ell)}z),\label{eq:pwa}
\end{equation}
where $\{\set Z^{(\ell)}\}_{\ell=1}^{L}$ is a collection of convex
sets that partition $\reals^{p}$, $\{\bar{\phi}_{i}^{(\ell)}\}_{\ell=1}^{L}\subset\reals^{p}$
and $\{\Phi_{i}^{(\ell)}\}_{\ell=1}^{L}\subset\reals^{p\times p}$.\footnote{If instead $\fe_{i}(z)=\sum_{\ell=1}^{L}\indic\{z\in\set Z_{i}^{(\ell)}\}(\bar{\phi}_{i}^{(\ell)}+\Phi_{i}^{(\ell)}z)$,
with a partition $\{\set Z_{i}^{(\ell)}\}_{\ell=1}^{L_{i}}$ that
depends on $i$, the model can nonetheless be written in terms of
$\{\fe_{i}\}_{i=0}^{k}$ of the form \eqref{pwa}, i.e.\ with a lag-independent
partition $\{\set Z^{(\ell)}\}_{\ell=1}^{L}$, by forming each $\set Z^{(\ell)}$
from a refinement of the sets in $\{\set Z_{i}^{(\ell)}\}_{\ell=1}^{L_{i}}$,
as $i$ ranges over $\{0,\ldots,k\}$.} When these parameters are such that $\fe_{i}$ is continuous for
each $i\in\{0,\ldots,k\}$, we shall say that each $\fe_{i}$ is a
\emph{piecewise affine function}, and refer to the model as a \emph{piecewise
affine VAR}. (We do \emph{not} consider cases in which $\fe_{i}$
may be discontinuous; so continuity is always implied when a function
or VAR is described as being piecewise affine.)

Piecewise affine VARs, of which the CKSVAR of \citet{SM21} is a recent
instance, provide a flexible but tractable means of introducing nonlinearity
into a vector autoregressive model. Defining 
\[
\indic^{(\ell)}(z)\defeq\indic\{z\in\set Z^{(\ell)}\},
\]
we see that
\[
\fe_{0}(z_{t})-\sum_{i=1}^{k}\fe_{i}(z_{t-i})=\sum_{\ell=1}^{L}\indic^{(\ell)}(z_{t})(\bar{\phi}_{i}^{(\ell)}+\Phi_{i}^{(\ell)}z_{t})-\sum_{i=1}^{k}\sum_{\ell=1}^{L}\indic^{(\ell)}(z_{t-i})(\bar{\phi}_{i}^{(\ell)}+\Phi_{i}^{(\ell)}z_{t-i}).
\]
This is a kind of endogenous regime-switching VAR, in which the sets
$\{\set Z^{(\ell)}\}_{\ell=1}^{L}$ demarcate $L$ distinct `regimes'.
However, unlike in typical models of this kind, there is no requirement
that the same `regime' apply e.g.\ to $z_{t}$ and $z_{t-1}$ --
each may lie in a different member of $\{\set Z^{(\ell)}\}_{\ell=1}^{L}$
-- and so it might be more correct to say that there are a total
of $L^{k}$ autoregressive regimes, once all the possible patterns
of membership of $\{z_{t-i}\}_{i=0}^{k}$ in the sets $\{\set Z^{(\ell)}\}_{\ell=1}^{L}$
are allowed for. Nonetheless, it turns out that a fruitful approach
to analysing the long-run behaviour of these systems focuses attention
on those $L$ `linear submodels' that arise when each of $\{z_{t-i}\}_{i=0}^{k}$
lie in the \emph{same} $\set Z^{(\ell)}$, to each $\ell\in\{1,\ldots,L\}$
of which we may associate the autoregressive polynomial
\begin{equation}
\Phi^{(\ell)}(\lambda)\defeq\Phi_{0}^{(\ell)}-\sum_{i=1}^{k}\Phi_{i}^{(\ell)}\lambda^{i}.\label{eq:Phi-ell}
\end{equation}

\subsection{Simplification of the JSR condition}

The conditions for membership of $\cvar_{r}$, for piecewise affine
VARs, will parallel \ref{enu:lin:homeo}\ass{--3}. Before stating
these, we first give a result that reduces \ref{enu:cvar:jsr} to
a condition on the JSR of a collection of only $L$ matrices (rather
than a collection of uncountably many matrices). To state it, observe
that 
\begin{align}
\pi(z) & =-\sum_{\ell=1}^{L}\indic^{(\ell)}(z)\left[\left(\bar{\phi}_{0}^{(\ell)}-\sum_{i=1}^{k}\bar{\phi}_{i}^{(\ell)}\right)+\left(\Phi_{0}^{(\ell)}-\sum_{i=1}^{k}\Phi_{i}^{(\ell)}\right)z\right]\nonumber \\
 & \eqdef\sum_{\ell=1}^{L}\indic^{(\ell)}(z)(\bar{\pi}^{(\ell)}+\Pi^{(\ell)}z).\label{eq:pi-piecewise}
\end{align}
In this context, we note that \ref{enu:cvar:rank} effectively requires
that
\begin{equation}
\Pi^{(\ell)}=-\Phi^{(\ell)}(1)=\alpha\beta^{(\ell)\trans}\label{eq:Pi-ell}
\end{equation}
where $\alpha,\beta^{(\ell)}\in\reals^{p\times r}$ with $\rank\alpha=\rank\beta^{(\ell)}=r$.
Further, let $\Gamma_{j}^{(\ell)}\defeq-\sum_{j=i+1}^{k}\Phi_{j}^{(\ell)}$,
$\b{\Gamma}^{(\ell)}\defeq[\Gamma_{i}^{(\ell)}]_{i=1}^{k-1}$, and
\begin{align}
\b{\beta}^{(\ell)\trans} & \defeq\begin{bmatrix}\beta^{(\ell)\trans}(\Phi_{0}^{(\ell)})^{-1} & 0\\
\b{\Gamma}^{(\ell)}(\Phi_{0}^{(\ell)})^{-1} & D
\end{bmatrix} & \b{\alpha} & \defeq\begin{bmatrix}\alpha & E^{\trans}\\
0 & I_{p(k-1)}
\end{bmatrix}.\label{eq:beta_l}
\end{align}

\begin{lem}
\label{lem:affine-conditions}Suppose that $(c,\{f_{i}\}_{i=0}^{k})$
is a piecewise affine model, for which there exists an $\alpha\in\reals^{p\times r}$
with $\rank\alpha=r$ such that
\begin{enumerate}
\item $c=\alpha\cn$, $\bar{\pi}^{(\ell)}=\alpha\bar{\cn}^{(\ell)}$ and
$\Pi^{(\ell)}=\alpha\beta^{(\ell)\trans}$, for $\cn,\bar{\cn}^{(\ell)},\beta^{(\ell)}\in\reals^{p\times r}$,
for all $\ell\in\{1,\ldots,L\}$; and
\item $\fe_{0}:\reals^{p}\setmap\reals^{p}$ is a homeomorphism.
\end{enumerate}
Then $(c,\{f_{i}\}_{i=0}^{k})$ satisfies \eqref{gradient} with ${\cal B}$
equal to $\ch\{\b{\beta}^{(\ell)}\}_{\ell=1}^{L}$, and so
\[
\rho_{\jsr}(\{I_{p(k-1)+r}+\b{\beta}^{\trans}\b{\alpha}\mid\b{\beta}\in{\cal B}\})=\rho_{\jsr}(\{I_{p(k-1)+r}+\b{\beta}^{(\ell)\trans}\b{\alpha}\}_{\ell=1}^{L}).
\]
\end{lem}
Since the $p(k-1)+r$ eigenvalues of $I_{p(k-1)+r}+\b{\beta}^{(\ell)\trans}\b{\alpha}$
correspond to the inverses of the non-unit roots of $\Phi^{(\ell)}(\lambda)$,
a necessary, though not sufficient condition for
\[
\rho_{\jsr}(\{I_{p(k-1)+r}+\b{\beta}^{(\ell)\trans}\b{\alpha}\}_{\ell=1}^{L})<1
\]
is that the eigenvalues of $I_{p(k-1)+r}+\b{\beta}^{(\ell)\trans}\b{\alpha}$
should lie strictly inside the unit circle, for all $\ell\in\{1,\ldots,L\}$.
Thus, if $\rank\Phi^{(\ell)}(1)=q=p-r$, as per \eqref{Pi-ell}, then
$\Phi^{(\ell)}(\lambda)$ will have $q$ roots at unity, and all others
strictly outside the unit circle (this follows e.g.\ from arguments
given in the proof of \propref{lin-rep-var}). In other words, for
a piecewise affine model to belong to class $\cvar_{r}(\alpha,\abv b,\abv{\rho})$
(for some $\abv{\rho}<1$), it is essentially necessary that each
of its $L$ linear submodels satisfy the usual conditions for a linear
VAR to give rise to $q$ common trends and $r$ cointegrating relations
(i.e.\ conditions \ref{enu:lin:rank} and \ref{enu:lin:roots} above).

\subsection{Membership of $\protect\cvar_{r}$}

It remains to consider \ref{enu:cvar:homeo}, i.e.\ the requirement
that $\fe_{0}$ be a homeomorphism. For the purposes of verifying
this condition, two important special cases of the piecewise affine
VAR are:
\begin{itemize}
\item the \emph{piecewise linear VAR} (PLVAR), in which there exists a basis
$\{a_{i}\}_{i=1}^{p}$ for $\reals^{p}$ such that each $\set Z^{(\ell)}$
can be written as a union of cones of the form
\[
\set C_{{\cal I}}\defeq\{z\in\reals^{p}\mid a_{i}^{\trans}z\geq0\sep\forall i\in{\cal I}\text{ and }a_{i}^{\trans}z<0\sep\forall i\notin{\cal I}\}
\]
where ${\cal I}$ ranges over the subsets of $\{1,\ldots,p\}$, and
$\bar{\phi}_{i}^{(\ell)}=0$ for all $i$ and $\ell$; and
\item the \emph{threshold affine VAR} (TAVAR), in which there exists an
$a\in\reals^{p}\backslash\{0\}$ and thresholds $\{\tau_{\ell}\}_{\ell=0}^{L}$
with $\tau_{\ell}<\tau_{\ell+1}$, $\tau_{0}=-\infty$ and $\tau_{L}=+\infty$,
such that
\[
\set Z^{(\ell)}=\{z\in\reals^{p}\mid a^{\trans}z\in(\tau_{\ell-1},\tau_{\ell}]\},
\]
i.e.\ the sets $\{\set Z^{(\ell)}\}$ take the forms of `bands'
in $\reals^{p}$. (In typical examples, $a=e_{p,i}$, i.e.\ it picks
out one `threshold variable' from the elements of $z_{t}$.)
\end{itemize}
In these models, the results of \citet{GLM80Ecta} provide necessary
and sufficient conditions for $\fe_{0}$ to be invertible, which can
be expressed in terms of the determinants of $\{\Phi_{0}^{(\ell)}\}_{\ell=1}^{L}$.
We thus have the following
\begin{prop}
\label{prop:affine}Suppose is $(c,\{\fe_{i}\}_{i=0}^{k})$ is either
a piecewise linear or threshold affine VAR, such that:
\begin{enumerate}[label=\ass{PWA.\arabic*}, leftmargin=1.75cm]
\item \label{enu:pwa:homeo}$\sgn\det\Phi_{0}^{(\ell)}=\sgn\det\Phi_{0}^{(1)}\neq0$
for all $\ell\in\{1,\ldots,L\}$;
\item \label{enu:pwa:rank}$c=\alpha\cn$, $\bar{\pi}^{(\ell)}=\alpha\bar{\cn}^{(\ell)}$
, and $\Pi^{(\ell)}=\alpha\beta^{(\ell)\trans}$, where $\alpha,\beta^{(\ell)}\in\reals^{p\times r}$
have rank $r$, for all $\ell\in\{1,\ldots,L\}$; and
\item \label{enu:pwa:jsr}$\rho_{\jsr}(\{I_{p(k-1)+r}+\b{\beta}^{(\ell)\trans}\b{\alpha}\}_{\ell=1}^{L})\leq\abv{\rho}$.
\end{enumerate}
Then $(c,\{f_{i}\}_{i=0}^{k})$ belongs to class $\cvar_{r}(\alpha,\abv b,\abv{\rho})$,
for $\abv b\geq\max_{1\leq\ell\leq L}\smlnorm{\b{\beta}^{(\ell)}}$,
with
\begin{equation}
\co(z)=\begin{bmatrix}\ct(z)\\
\eq(z)
\end{bmatrix}=\sum_{\ell=1}^{L}\indic^{(\ell)}(z)\left\{ \begin{bmatrix}\alpha_{\perp}^{\trans}\bar{\ct}^{(\ell)}\\
\bar{\cn}^{(\ell)}
\end{bmatrix}+\begin{bmatrix}\alpha_{\perp}^{\trans}\srp^{(\ell)}\\
\beta^{(\ell)\trans}
\end{bmatrix}z\right\} \eqdef\sum_{\ell=1}^{L}\indic^{(\ell)}(z)\{\bar{\co}^{(\ell)}+\CO^{(\ell)}z\}\label{eq:co-pwa}
\end{equation}
where 
\begin{align*}
\bar{\ct}^{(\ell)} & \defeq\bar{\phi}_{0}^{(\ell)}-\sum_{i=1}^{k-1}\bar{\gamma}_{i}^{(\ell)} & \srp^{(\ell)} & \defeq\Phi_{0}^{(\ell)}-\sum_{i=1}^{k-1}\Gamma_{i}^{(\ell)}
\end{align*}
for $\bar{\gamma}_{i}^{(\ell)}\defeq-\sum_{j=i+1}^{k}\bar{\phi}_{j}^{(\ell)}$.
Moreover, the same conclusion holds if $(c,\{\fe_{i}\}_{i=0}^{k})$
is a general piecewise affine model satisfying \ref{enu:pwa:rank}
and \ref{enu:pwa:jsr} above, if $\fe_{0}$ is a homeomorphism.
\end{prop}
Since it may be verified that a linear VAR satisfying \ref{enu:lin:homeo}\ass{--3}
is itself a piecewise linear model satisfying \ref{enu:pwa:homeo}\ass{--3}
above, \propref{lin-rep-var} may be construed as a special case of
the preceding (though for expository reasons, a separate proof of
that result is given in \appref{examples}).

\subsection{Smooth transitions}

The model \eqref{pwa} may also be extended to allow for smooth transitions
between the $L$ regimes. In the literature on smooth transition (vector)
autoregressive models, the conventional approach (e.g.\ \citealp[Sec.~3.3]{HT13})
is to replace the indicator functions $\indic\{z\in\set Z^{(\ell)}\}$
by smooth maps $\Lambda^{(\ell)}(z)$, so that now
\[
\fe_{i}^{\mathrm{ST}}(z)=\sum_{\ell=1}^{L}\Lambda^{(\ell)}(z)(\bar{\phi}_{i}^{(\ell)}+\Phi_{i}^{(\ell)}z),
\]
where $\Lambda^{(\ell)}(z)\in[0,1]$ and $\sum_{\ell=1}^{L}\Lambda^{(\ell)}(z)=1$
for all $z\in\reals^{p}$, so that $\fe_{i}^{\mathrm{ST}}(z)$ is
a always a smooth, convex combination of the affine functions $\{z\elmap\bar{\phi}_{i}^{(\ell)}+\Phi_{i}^{(\ell)}z\}_{\ell=1}^{L}$.
While models of this kind may be accommodated within our framework
(under certain regularity conditions), the fact that the \emph{gradient}
of $\fe_{i}^{\mathrm{ST}}$ is \emph{not} a convex combination of
those underlying affine regimes makes it difficult to reduce \ref{enu:cvar:jsr}
to a bound on the JSR of a finite collection of matrices, in the manner
of \propref{affine}. As an alternative specification that allows
for smooth transitions between regimes, while also keeping \ref{enu:cvar:jsr}
tractable, consider
\begin{equation}
\fe_{i,K}(z)\defeq\int_{\reals^{p}}\fe_{i}(z+u)K(u)\diff u\label{eq:smoothed}
\end{equation}
where $\fe_{i}$ is a (continuous) piecewise affine function as in
\eqref{pwa} above, and $K$ is a smooth kernel with mean zero. Then
we have the following.
\begin{prop}
\label{prop:smoothed}Suppose that:
\begin{enumerate}
\item $(c,\{\fe_{i}\}_{i=0}^{k})$ is either a piecewise linear or threshold
affine VAR satisfying \ref{enu:pwa:homeo}\ass{--3};
\item \label{enu:f0K}$\Phi_{0}^{(\ell)}=\Phi_{0}$ for all $\ell\in\{1,\ldots,L\}$,
for some nonsingular $\Phi_{0}$;
\item $K:\reals\setmap\reals$ is continuous and non-negative, with $\int_{\reals^{p}}K(u)\diff u=1$
and $\int_{\reals^{p}}uK(u)\diff u=0$;
\item $\fe_{i,K}$ is constructed by smoothing $\fe_{i}$ with $K$ as in
\eqref{smoothed}, for $i\in\{0,\ldots,k\}$.
\end{enumerate}
Then $(c,\{\fe_{i,K}\}_{i=0}^{k})$ belongs to class $\cvar_{r}(\alpha,\abv b,\abv{\rho})$,
for $\abv b\geq\max_{1\leq\ell\leq L}\smlnorm{\b{\beta}^{(\ell)}}$,
for $\b{\beta}^{(\ell)}$ as in \eqref{beta_l} above.
\end{prop}
The requirement in \enuref{f0K} that $\fe_{0}$ be linear is needed
principally to facilitate the verification of \enuref{cvar:jsr},
because of the role played by $\fe_{K,0}^{-1}$ there. We expect
that it should be possible to extend the preceding to allow $\fe_{0}$
to be a general piecewise affine function, albeit possibly at the
cost of additional regularity conditions. (Indeed, it may be shown
that the smooth counterpart $\fe_{0,K}$ of $\fe_{0}$ is a homeomorphism
if every $\Phi\in\ch\{\Phi_{0}^{(\ell)}\}_{\ell=1}^{L}$ is invertible,
thereby satisfying \enuref{cvar:homeo}.)

\section{Conclusion}

\label{sec:conclusion}

This paper has extended the Granger--Johansen representation theorem
to a flexible class of additively time-separable, nonlinear SVAR models.
This shows that such models may be applied directly to time series
in which (common) stochastic trends are present, without the need
for pre-filtering, thus avoiding the potential for misspecification
that this entails. As an important corollary to our results, we show
that these models are capable of supporting the same kinds of long-run
identifying restrictions as are available in linear cointegrated SVARs.
A companion paper to the present work provides limit theory for $n^{-1/2}z_{\smlfloor{n\lambda}}$,
and a further discussion of nonlinear cointegration, under the assumptions
maintained here.

\bibliographystyle{ecta}
\bibliography{cksvar}

\appendix

\section{Observational equivalence and identification}

\label{app:identification}

In this appendix, we develop a characterisation of observational equivalence
in the nonlinear VAR model \eqref{nlVAR}, reproduced here as
\begin{equation}
\fe_{0}(z_{t})=c+\sum_{i=1}^{k}\fe_{i}(z_{t-i})+u_{t},\tag{\ref{eq:nlVAR}}
\end{equation}
where $\{u_{t}\}$ is i.i.d.\ with $\expect u_{t}=0$ and $\expect u_{t}u_{t}^{\trans}=I_{p}$.
We approach this problem by casting the model as a special case of
the nonlinear simultaneous equations model considered by \citet{Matz09Ecta},
in which additional structure is provided by the additive separability
across the variables $\{z_{t-i}\}_{i=0}^{k}$. This structure permits
us, with the aid of her Theorem~3.3, to obtain a sharper characterisation
of observational equivalence than is available in general, non-separable
models; this is stated as \propref{matzkin} below. The application
of this result to the structural nonlinear VAR model \eqref{nlVAR}
is then provided by our \thmref{shockid}.

In obtaining \thmref{shockid}, we have found it convenient to impose
smoothness regularity conditions on the functions $\{\fe_{i}\}_{i=0}^{p}$,
and the density $\den$ of $u_{t}$, that are the same as those maintained
in \citet{Matz09Ecta}. These could plausibly be relaxed, without
affecting the substance of \thmref{shockid}, so as to accommodate
the kinds of piecewise smooth functions that parametrise the piecewise
affine VARs discussed in \subsecref{piecewise-affine}. In any case,
since \thmref{shockid} is only intended to provide an illustrative
justification for the formulation of the structural nonlinear VAR
given in \eqref{nlSVAR}, we do not maintain the conditions required
by \thmref{shockid} elsewhere in this paper.

\subsection{In a general simultaneous equations model}

\label{app:id-simeq}

We first derive a characterisation of observational equivalence in
the setting of an additively separable, nonlinear simultaneous equations
model, drawing on \citet{Matz09Ecta}. To facilitate the application
of her results, in this section (only) we adopt her notation, setting
aside that used in the rest of this paper. We consider a model defined
by
\begin{equation}
U=r(Y,X_{1},X_{2})=r_{0}(Y)+r_{1}(X_{1})+r_{2}(X_{2})\label{eq:nlsem}
\end{equation}
where $U$ and $Y$ are random vectors taking values in $\reals^{G}$,
and $X_{i}$ is a random vector taking values in $\reals^{K_{i}}$,
for $i\in\{1,2\}$. Let $f_{U}$ denote the density of $U$. We suppose
$(Y,X_{1},X_{2})$ are observed, and that $U$ is independent of $X_{1}$
and $X_{2}$. The model is parametrised by the functions $r_{0}:\reals^{G}\setmap\reals^{G}$,
$r_{i}:\reals^{K_{i}}\setmap\reals^{G}$ for $i\in\{1,2\}$, and the
density $f_{U}$; let $\Gamma_{0}\ni r_{0}$, $\Gamma_{i}\ni r_{i}$
for $i\in\{1,2\}$ and $\Phi\ni f_{U}$ denote the sets of functions
(and densities) that together constitute the model parameter space.
Our first assumption records the restrictions imposed on that parameter
space. For a function $g:\reals^{m}\setmap\reals^{n}$, let $Dg(z_{0})=[(\partial g_{i}/\partial z_{j})(z_{0})]$
denote the $(n\times m)$ Jacobian of $g(z)$ at $z=z_{0}$.

\assumpname{PAR}
\begin{assumption}
\label{ass:ps}~For every $\tilde{r}_{i}\in\Gamma_{i}$, $i\in\{0,1,2\}$:
\begin{enumerate}[label=\ass{\arabic*.}, ref=\ass{\arabic*}]
\item \label{enu:ps:twicediff}$\tilde{r}_{i}$ is twice continuously differentiable
on $\reals^{G}$.
\end{enumerate}
For every $\tilde{r}_{0}\in\Gamma_{0}$:
\begin{enumerate}[resume, resume*]
\item \label{enu:ps:bij}$\tilde{r}_{0}:\reals^{G}\setmap\reals^{G}$ is
bijective; and
\item \label{enu:ps:jac}$\det D\tilde{r}_{0}(y)>0$ for every $y\in\reals^{G}$.
\end{enumerate}
For every $f_{\tilde{U}}\in\Phi$:
\begin{enumerate}[resume, resume*]
\item \label{enu:ps:densdiff}$f_{\tilde{U}}$ is continuously differentiable
on $\reals^{G}$;
\item \label{enu:ps:supp}the support of $f_{\tilde{U}}$ is $\reals^{G}$;
and
\item \label{enu:ps:norm}$\int_{\reals^{G}}uf_{\tilde{U}}(u)\diff u=0$
and $\int_{\reals^{G}}uu^{\trans}f_{\tilde{U}}(u)\diff u=I_{G}$.
\end{enumerate}
\end{assumption}

Assumptions \assref{ps}\ass{.}\ref{enu:ps:twicediff}\ass{--}\ref{enu:ps:supp}
are exactly the conditions maintained (on the sets $\Gamma$ and $\Phi$)
by \citet[p.~948]{Matz09Ecta}, here specialised to the additively
separable structure of \eqref{nlsem}; while \assref{ps}\ass{.}\ref{enu:ps:norm}
is a location and scale normalisation that requires the unobservable
vector $U$ to have mean zero and identity covariance matrix. For
given $\tilde{r}_{i}\in\Gamma_{i}$, for each $i\in\{0,1,2\}$, define
\[
\tilde{r}(y,x_{1},x_{2})\defeq\tilde{r}_{0}(y)+\tilde{r}_{1}(x_{1})+\tilde{r}_{2}(x_{2})
\]
Following the discussion on pp.\ 948--50 of \citet{Matz09Ecta},
under \assref{ps} the conditional density of $Y$ given $(X_{1},X_{2})=(x_{1},x_{2})$,
associated with the parameters $(\tilde{r},f_{\tilde{U}})$, is
\[
f_{Y\mid(X_{1},X_{2})=(x_{1},x_{2})}[y;(\tilde{r},f_{\tilde{U}})]=f_{\tilde{U}}[\tilde{r}(y,x_{1},x_{2})]\smlabs{\det D\tilde{r}_{0}(y)}
\]
We say that $(\{\tilde{r}_{i}\}_{i=1}^{3},f_{\tilde{U}})$ is \emph{observationally
equivalent} to $(\{r_{i}\}_{i=1}^{3},f_{U})$ if both imply identical
conditional densities (for $Y$ given $(X_{1},X_{2})$; cf.\ (3.1)
in \citealp{Matz09Ecta}).

We shall also suppose that the true data generating mechanism, as
encoded by the functions $(r_{0},r_{1},r_{2},f_{U})$, satisfies the
following. (Note, in particular, that these conditions are \emph{not}
imposed on the entirety of the model parameter space.) Our conditions
on $f_{U}$ are similar to those entertained by \citet{Matz09Ecta}
in the context of the example exposited in Section~4.2 of her paper,
and permit some of her arguments to be transposed to the present setting.

\assumpname{DGP}
\begin{assumption}
\label{ass:sem} $(r_{0},r_{1},r_{2},f_{U})$ are such that:
\begin{enumerate}[label=\ass{\arabic*.}, ref=\ass{\arabic*}]
\item \label{enu:sem:surj}the image of $(x_{1},x_{2})\elmap r_{1}(x_{1})+r_{2}(x_{2})$
is $\reals^{G}$;
\item \label{enu:sem:rank}$\rank Dr_{i}(x_{i})=G$ for all $x_{i}\in\reals^{K_{i}}$,
for $i\in\{1,2\}$;
\item \label{enu:sem:mode}there exists a $u^{\ast}$ such that $D(\log f_{U})(u^{\ast})=0$;
and
\item \label{enu:sem:span}there exist $\{\abv u_{i}\}_{i=1}^{G}\subset\reals^{G}$
such that $\{\pi_{i}\}_{i=1}^{G}$ span $\reals^{G}$, where
\[
\pi_{i}\defeq D(\log f_{U})(\abv u_{i})^{\trans}.
\]
\end{enumerate}
\end{assumption}

We may now state the main result of this section.
\begin{prop}
\label{prop:matzkin}Suppose that \assref{ps} and \assref{sem} hold.
Let $\tilde{r}_{i}\in\Gamma_{i}$ for $i\in\{0,1,2\}$. Then there
exists an $f_{\tilde{U}}\in\Phi$, such that $(\{\tilde{r}_{i}\}_{i=0}^{2},f_{\tilde{U}})$
is observationally equivalent to $(\{r_{i}\}_{i=0}^{2},f_{U})$, if
and only if there exists an orthogonal matrix $Q\in\reals^{G\times G}$
such that
\[
\tilde{r}_{0}(y)+\tilde{r}_{1}(x_{1})+\tilde{r}_{2}(x_{2})=Q[r_{0}(y)+r_{1}(x_{1})+r_{2}(x_{2})]
\]
for all $(y,x_{1},x_{2})\in\reals^{G}\times\reals^{K_{1}}\times\reals^{K_{2}}$.
\end{prop}
\begin{proof}
The forward implication is trivial; it remains to prove the reverse
implication. Accordingly, suppose that there exists an $f_{\tilde{U}}\in\Phi$
such that $(\{\tilde{r}_{i}\}_{i=1}^{3},f_{\tilde{U}})$ is observationally
equivalent to $(\{r_{i}\}_{i=1}^{3},f_{U})$. By the additive separability
of the model in $Y$ and $X=(X_{1},X_{2})$, it follows that $\Delta_{x}$
(as defined on p.\ 957 of \citealp{Matz09Ecta}) is identically zero,
and that $\Delta_{y}$ depends only on $y$ (i.e.\ and not on $x_{1}$
or $x_{2}$). Therefore by Theorem~3.3 of \citet{Matz09Ecta}, the
$(G+K_{1}+K_{2})\times(G+1)$ matrix 
\begin{equation}
M(y,x_{1},x_{2})\defeq\begin{bmatrix}D\tilde{r}_{0}(y)^{\trans} & \Delta_{y}(y)+Dr_{0}(y)^{\trans}s(y,x_{1},x_{2})\\
D\tilde{r}_{1}(x_{1})^{\trans} & Dr_{1}(x_{1})^{\trans}s(y,x_{1},x_{2})\\
D\tilde{r}_{2}(x_{2})^{\trans} & Dr_{2}(x_{2})^{\trans}s(y,x_{1},x_{2})
\end{bmatrix}\label{eq:matkin-matrix}
\end{equation}
has rank $G$, for all $(y,x_{1},x_{2})\in\reals^{G+K_{1}+K_{2}}$,
where
\[
s(y,x_{1},x_{2})\defeq D(\log f_{U})[r(y,x_{1},x_{2})]^{\trans}.
\]
The proof proceeds in five steps.

(i) Claim: $\rank D\tilde{r}_{j}(x_{j})=G$, for all $x_{j}\in\reals^{K_{j}}$,
for $j\in\{1,2\}$. Let $(x_{1},x_{2})$ be given. By \assref{ps}\ass{.}\ref{enu:ps:bij}
and \assref{sem}\ass{.}\ref{enu:sem:span}, there exists a $y^{i}\in\reals^{G}$
(depending on $x_{1}$ and $x_{2}$) such that
\begin{equation}
r_{0}(y^{i})=\abv u_{i}-r_{1}(x_{1})-r_{2}(x_{2})\label{eq:ysupi}
\end{equation}
whence $s(y^{i},x_{1},x_{2})=\pi_{i}$. Since $D\tilde{r}_{0}(y^{i})^{\trans}$
has rank $G$, the matrix 
\[
M(y^{i},x_{1},x_{2})=\begin{bmatrix}D\tilde{r}_{0}(y^{i})^{\trans} & \Delta_{y}(y^{i})+Dr_{0}(y^{i})^{\trans}\pi_{i}\\
D\tilde{r}_{1}(x_{1})^{\trans} & Dr_{1}(x_{1})^{\trans}\pi_{i}\\
D\tilde{r}_{2}(x_{2})^{\trans} & Dr_{2}(x_{2})^{\trans}\pi_{i}
\end{bmatrix}
\]
can have rank $G$ only if $Dr_{1}(x_{1})^{\trans}\pi_{i}$ lies in
the span of $D\tilde{r}_{1}(x_{1})^{\trans}$. Repeating this for
every $i\in\{1,\ldots,G\}$, and constructing $\Pi\defeq[\pi_{1},\ldots,\pi_{G}]$,
which is invertible, it follows that
\[
\spn\{Dr_{1}(x_{1})^{\trans}\}=\spn\{Dr_{1}(x_{1})^{\trans}\Pi\}\subset\spn\{D\tilde{r}_{1}(x_{1})^{\trans}\}
\]
where $\spn A$ denotes the column span of the matrix $A$. Hence
$D\tilde{r}_{1}(x_{1})$ indeed has rank $G$; an identical argument
yields the same conclusion for $D\tilde{r}_{2}(x_{2})$.

(ii) Claim: $\Delta_{y}(y)$ is identically zero. Let $y\in\reals^{G}$
be given. By \assref{sem}\ass{.}\ref{enu:sem:surj} and \assref{sem}\ass{.}\ref{enu:sem:mode},
there exists $(x_{1}^{\ast},x_{2}^{\ast})\in\reals^{K_{1}+K_{2}}$
(depending on $y$) such that
\[
r_{1}(x_{1}^{\ast})+r_{2}(x_{2}^{\ast})=u^{\ast}-r_{0}(y)
\]
whence $s(y,x_{1}^{\ast},x_{2}^{\ast})=0$. It follows that
\[
M(y,x_{1}^{\ast},x_{2}^{\ast})=\begin{bmatrix}D\tilde{r}_{0}(y)^{\trans} & \Delta_{y}(y)\\
D\tilde{r}_{1}(x_{1}^{\ast})^{\trans} & 0\\
D\tilde{r}_{2}(x_{2}^{\ast})^{\trans} & 0
\end{bmatrix}.
\]
Since $\rank D\tilde{r}_{1}(x_{1}^{\ast})^{\trans}=G$, the preceding
can only have rank $G$ if $\Delta_{y}(y)=0$.

(iii). Claim: $D\tilde{r}_{j}(x_{j})^{\trans}=Dr_{j}(x_{j})^{\trans}Q$
for all $x_{j}\in\reals^{K_{j}}$, for $j\in\{1,2\}$, for some invertible
$Q\in\reals^{G\times G}$. Let $(x_{1},x_{2})$ be given, and let
$y^{i}$ be as in \eqref{ysupi} above. Since $\rank D\tilde{r}_{1}(x_{1})^{\trans}=G$,
for $M(y^{i},x_{1},x_{2})$ to have rank $G$, it must be the case
that the $(K_{1}+K_{2})\times(G+1)$ submatrix
\[
\begin{bmatrix}D\tilde{r}_{1}(x_{1})^{\trans} & Dr_{1}(x_{1})^{\trans}\pi_{i}\\
D\tilde{r}_{2}(x_{2})^{\trans} & Dr_{2}(x_{2})^{\trans}\pi_{i}
\end{bmatrix}
\]
also has rank $G$. Since $\rank Dr_{1}(x_{1})^{\trans}=G$, the final
column of the preceding matrix is nonzero. Therefore, there must exist
a (nonzero) $\lambda_{i}(x_{1},x_{2})\in\reals^{G}$, possibly depending
on $(x_{1},x_{2})$, such that
\[
\begin{bmatrix}D\tilde{r}_{1}(x_{1})^{\trans} & Dr_{1}(x_{1})^{\trans}\pi_{i}\\
D\tilde{r}_{2}(x_{2})^{\trans} & Dr_{2}(x_{2})^{\trans}\pi_{i}
\end{bmatrix}\begin{bmatrix}\lambda_{i}(x_{1},x_{2})\\
-1
\end{bmatrix}=0.
\]
Now suppose we keep the value of $x_{1}$ fixed, while allowing $x_{2}$
to vary freely over $\reals^{K_{2}}$: while doing so, we adjust the
value of $y^{i}$ such that \eqref{ysupi} is maintained. It follows
that
\[
D\tilde{r}_{1}(x_{1})^{\trans}\lambda_{i}(x_{1},x_{2})=Dr_{1}(x_{1})^{\trans}\pi_{i}
\]
for all $x_{2}\in\reals^{K_{2}}$. Since $\rank D\tilde{r}_{1}(x_{1})^{\trans}=G$,
and the r.h.s.\ does not depend on $x_{2}$, it follows that $\lambda_{i}(x_{1},x_{2})$
cannot depend on $x_{2}$. A symmetrical argument (holding $x_{2}$
fixed while varying $x_{1}\in\reals^{K_{2}}$) yields that $\lambda_{i}(x_{1},x_{2})$
cannot depend on $x_{1}$ either. Hence $\lambda_{i}(x_{1},x_{2})=\lambda_{i}$
is constant.

Repeating the above for every $i\in\{1,\ldots,G\}$, and collecting
$\Lambda=[\lambda_{1},\ldots,\lambda_{G}]$, we obtain
\[
D\tilde{r}_{j}(x_{j})^{\trans}\Lambda=Dr_{j}(x_{j})^{\trans}\Pi
\]
for all $x_{j}\in\reals^{K_{j}}$, for $i\in\{1,2\}$. Since $\rank D\tilde{r}_{j}(x_{j})^{\trans}=\rank Dr_{j}(x_{j})^{\trans}=\rank\Gamma=G$,
it follows that $\Lambda$ is invertible, whence the claim holds with
$Q\defeq\Gamma\Lambda^{-1}$.

(iv). Claim: $D\tilde{r}_{0}(y)^{\trans}=Dr_{0}(y)^{\trans}Q$ for
all $y\in\reals^{G}$. Let $y$ be given. By \assref{sem}\ass{.}\ref{enu:sem:surj}
and \assref{sem}\ass{.}\ref{enu:sem:span}, there exist $(x_{1}^{i},x_{2}^{i})\in\reals^{K_{1}+K_{2}}$
such that
\[
r_{1}(x_{1}^{i})+r_{2}(x_{2}^{i})=\abv u_{i}-r_{0}(y)
\]
whence $s(y,x_{1}^{i},x_{2}^{i})=\pi_{i}$. We know from the preceding
arguments that $\lambda_{i}$ is the unique element of $\reals^{G}$
such that
\[
M(y,x_{1}^{i},x_{2}^{i})\begin{bmatrix}\lambda_{i}\\
-1
\end{bmatrix}=\begin{bmatrix}D\tilde{r}_{0}(y)^{\trans} & Dr_{0}(y)^{\trans}\pi_{i}\\
D\tilde{r}_{1}(x_{1}^{i})^{\trans} & Dr_{1}(x_{1}^{i})^{\trans}\pi_{i}\\
D\tilde{r}_{2}(x_{2}^{i})^{\trans} & Dr_{2}(x_{2}^{i})^{\trans}\pi_{i}
\end{bmatrix}\begin{bmatrix}\lambda_{i}\\
-1
\end{bmatrix}=\begin{bmatrix}c\\
0_{K_{1}}\\
0_{K_{2}}
\end{bmatrix}
\]
for some $c\in\reals^{G}$. Since $\rank M(y,x_{1}^{i},x_{2}^{i})=G$,
we must have $c=0$, and therefore
\[
D\tilde{r}_{0}(y)^{\trans}\lambda_{i}=Dr_{0}(y)^{\trans}\pi_{i}.
\]
Since $y$ was arbitrary, this holds for all $y\in\reals^{G}$. Repeating
this for each $i\in\{1,\ldots,G\}$ proves the claim.

(v). It follows from the preceding, in particular from the third and
fourth claims, that
\[
\tilde{r}_{0}(y)+\tilde{r}_{1}(x_{1})+\tilde{r}_{2}(x_{2})=d+Q[r_{0}(y)+r_{1}(x_{1})+r_{2}(x_{2})]
\]
for some $d\in\reals^{G}$, for all $(y,x_{1},x_{2})\in\reals^{G+K_{1}+K_{2}}$.
Deduce that
\[
\tilde{U}=\tilde{r}(Y)+\tilde{r}_{1}(X_{1})+\tilde{r}_{2}(X_{2})=d+Q[r_{0}(Y)+r_{1}(X_{1})+r_{2}(X_{2})]=d+QU,
\]
where $\tilde{U}$ and $U$ have densities $f_{\tilde{U}}$ and $f_{U}$,
respectively. By \assref{ps}\ass{.}\ref{enu:ps:norm},
\[
0=\expect\tilde{U}=d+Q\expect U=d,
\]
whence $d=0$, and
\[
I_{G}=\expect\tilde{U}\tilde{U}^{\trans}=Q\expect UU^{\trans}Q^{\trans}=QQ^{\trans},
\]
whence $Q$ is an orthogonal matrix.
\end{proof}

\subsection{Application to the nonlinear structural VAR}

\label{app:id-VAR}

We now return to the framework of \subsecref{identification}. For
each $i\in\{0,\ldots,k\}$, let $\fespc_{i}$ denote a collection
of functions $\reals^{p}\setmap\reals^{p}$, and $\denspc$ a collection
of density functions on $\reals^{p}$. The parameter space for the
nonlinear VAR \eqref{nlVAR} is formally defined as
\[
\set P\defeq\{(\tilde{c},\{\tilde{\fe}_{i}\}_{i=0}^{k},\tilde{\den})\mid\tilde{c}\in\reals\sep\tilde{\fe}_{i}\in\fespc_{i}\text{ for }i\in\{0,\ldots,k\}\sep\tilde{\den}\in\denspc\},
\]
so that under the parameters $(\tilde{c},\{\tilde{\fe}_{i}\}_{i=0}^{k},\tilde{\den})\in\set P$,
the model is a nonlinear VAR($k$),
\[
\tilde{\fe}_{0}(z_{t})=\tilde{c}+\sum_{i=1}^{k}\tilde{\fe}_{i}(z_{t-i})+\tilde{u}_{t}
\]
where $\{u_{t}\}$ is i.i.d.\ with density $\tilde{\den}$.
\begin{thm}
\label{thm:shockid}Suppose that $\{\fespc_{i}\}_{i=0}^{k}$ and $\denspc$
are such that every:
\begin{enumerate}
\item \label{enu:shockid:smooth}$\tilde{\fe}_{i}\in\fespc_{i}$ is twice
continuously differentiable, with $\tilde{\fe}_{i}(0)=0$, for all
$i\in\{0,\ldots,k\}$;
\item \label{enu:shockid:bij}$\tilde{\fe}_{0}\in\fespc_{0}$ is a bijection
$\reals^{p}\setmap\reals^{p}$, with $\det D\tilde{\fe}_{0}(z)>0$
for every $z\in\reals^{p}$; and
\item \label{enu:shockid:dens}$\tilde{\den}\in\denspc$ is continuously
differentiable with support $\reals^{p}$, with 
\begin{align*}
\int_{\reals^{p}}u\tilde{\den}(u)\diff u & =0 & \int_{\reals^{p}}uu^{\trans}\tilde{\den}(u)\diff u & =I_{p}.
\end{align*}
\end{enumerate}
Suppose also that the parameters $(\{\fe_{i}\}_{i=0}^{k},\den)$ that
generated $\{z_{t}\}$ are such that
\begin{enumerate}[resume]
\item \label{enu:shockid:surj}the image of the map $(\bar{z}_{1},\ldots,\bar{z}_{k})\elmap\sum_{i=1}^{k}\fe_{i}(\bar{z}_{i})$
is $\reals^{p}$;
\item \label{enu:shockid:jac}for some ${\cal I}_{1},{\cal I}_{2}\subset\{1,\ldots,k\}$
with ${\cal I}_{1}\intsect{\cal I}_{2}=\emptyset$, the Jacobian of
$\{\bar{z}_{i}\}_{i\in{\cal I}_{j}}\elmap\sum_{i\in{\cal I}_{j}}\fe_{i}(\abv z_{i})$
has rank $p$ everywhere, for $j\in\{1,2\}$; and
\item \label{enu:shockid:densrank}there exists $\{\abv u_{i}\}_{i=0}^{p}\subset\reals^{p}$
such that $D(\log\den)(\abv u_{0})=0$ and $\spn\{D(\log\den)(\abv u_{i})^{\trans}\}_{i=1}^{p}=\reals^{p}$.
\end{enumerate}
Then if $\tilde{\fe}_{i}\in\fespc_{i}$ for $i\in\{0,\ldots,k\}$,
there exists a $\tilde{\den}\in\denspc$ such that $(\tilde{c},\{\tilde{\fe}_{i}\}_{i=0}^{k},\tilde{\den})$
is observationally equivalent to $(c,\{\fe_{i}\}_{i=0}^{k},\den)$,
if and only if
\begin{align}
\tilde{c} & =\Upsilon c & \tilde{\fe}_{i}(z) & =\Upsilon\fe_{i}(z)\sep\forall z\in\reals^{p}\sep i\in\{0,\ldots,k\}\label{eq:orthogtr}
\end{align}
for some orthogonal matrix $\Upsilon\in\reals^{p\times p}$.
\end{thm}

Arguably the most significant of the preceding conditions are \ref{enu:shockid:surj}
and \ref{enu:shockid:jac}, which we briefly remark on here. Note
that these refer only to the parameters that actually generated $\{z_{t}\}$,
and not the entirety of the model parameter space. Condition~\ref{enu:shockid:surj}
requires that the conditional mean of $\fe_{0}(z_{t})$, as given
by $\sum_{i=1}^{k}\fe_{i}(z_{t-i})$, should be able to take any value
in $\reals^{p}$: we may interpret this is as requiring that there
be sufficient `global' dependence of the r.h.s.\ on the lags of
$z_{t}$. This excludes, for example, the situation where the r.h.s.\ is
identically zero (excepting $u_{t}$), in which case $\fe_{0}$ would
clearly not be identified in any meaningful sense. Condition~\ref{enu:shockid:jac}
is a kind of rank condition, which requires that we be able to partition
the $k$ lags of $z_{t}$ in such a way, between the sets ${\cal I}_{1}$
and ${\cal I}_{2}$, that the Jacobians of both $\sum_{i\in{\cal I}_{1}}\fe_{i}(\abv z_{i})$
and $\sum_{i\in{\cal I}_{2}}\fe_{i}(\abv z_{i})$ have rank $p$:
there is in this sense sufficient `local' dependence of the r.h.s.\ on
the lags of $z_{t}$.
\begin{proof}[Proof of \thmref{shockid}]
The forward implication is trivial; it remains to prove the reverse
implication. Without loss of generality, we may take ${\cal I}_{1}$
and ${\cal I}_{2}$ to be such that ${\cal I}_{1}\union{\cal I}_{2}=\{1,\ldots,k\}$.

Fix $t\in\naturals$. The idea is to apply \propref{matzkin} by taking
\begin{align}
Y & =z_{t}, & X_{j} & =\{z_{t-i}\}_{i\in{\cal I}_{j}}\text{ for }j\in\{1,2\}.\label{eq:matzkinmapping}
\end{align}
We also define
\begin{align*}
\bar{Y} & \defeq\bar{z}_{t} & \bar{X}_{j} & \defeq\{\bar{z}_{t-i}\}_{i\in{\cal I}_{j}}\text{ for }j\in\{1,2\},
\end{align*}
where $\bar{z}_{t-i}\in\reals^{p}$ is non-random, for $i\in\{0,\ldots,k\}$.
For each $(\tilde{c},\{\tilde{\fe}_{i}\}_{i=0}^{k},\tilde{\den})\in\set P$,
define the maps $\{\tilde{r}_{i}\}_{i=0}^{2}$ such that
\begin{align}
\tilde{r}_{0}(\bar{Y}) & =\tilde{\fe}_{0}(\bar{z}_{t})-\tilde{c}, & \tilde{r}_{1}(\bar{X}_{1}) & =-\sum_{i\in{\cal I}_{1}}\tilde{\fe}_{i}(\bar{z}_{t-i}), & \tilde{r}_{2}(\bar{X}_{2}) & =-\sum_{i\in{\cal I}_{2}}\tilde{\fe}_{i}(\bar{z}_{t-i});\label{eq:rtildes}
\end{align}
for all $\bar{z}_{t-i}\in\reals^{p}$, for all $i\in\{0,\ldots,k\}$,
so that under the parameters $(\tilde{c},\{\tilde{\fe}_{i}\}_{i=0}^{k},\tilde{\den})$,
\begin{equation}
\tilde{U}=\tilde{r}_{0}(Y)+\tilde{r}_{1}(X_{1})+\tilde{r}_{2}(X_{2})=\tilde{\fe}_{0}(z_{t})-\tilde{c}-\sum_{i=1}^{k}\tilde{\fe}_{i}(z_{t-i})=\tilde{u}_{t}\label{eq:utilde}
\end{equation}
has density $\tilde{\den}$. Let $\{r_{j}\}_{j=0}^{2}$ be defined
as in \eqref{rtildes}, when $\tilde{\fe}_{i}=\fe_{i}$ for $i\in\{0,\ldots,k\}$.

Conditions \ref{enu:shockid:smooth} and \ref{enu:shockid:bij} above
ensure that each $\tilde{r}_{j}$ satisfies \ref{ass:ps}\ass{.}\ref{enu:ps:twicediff}\ass{--}\ref{enu:ps:jac},
for $j\in\{0,1,2\}$. By condition \ref{enu:shockid:dens}, the density
of the associated $\tilde{U}$ in \eqref{utilde} satisfies \ref{ass:ps}\ass{.}\ref{enu:ps:densdiff}\ass{--}\ref{enu:ps:norm}.
Regarding the actual values of the model parameters, conditions \ref{enu:shockid:surj}
and \ref{enu:shockid:jac} imply that $\{r_{j}\}_{j=0}^{2}$ satisfy
\ref{ass:sem}\ass{.}\ref{enu:sem:surj}--\ref{enu:sem:rank}, while
\ref{enu:shockid:densrank} implies that the density of $U$ satisfies
\ref{ass:sem}\ass{.}\ref{enu:sem:mode}\ass{--}\ref{enu:sem:span}.
Thus the requirements of \propref{matzkin} are satisfied.

Observe that the conditional density of $Y$ given $(X_{1},X_{2})$
is the same as the conditional density of $z_{t}$ given $(z_{t-1},\ldots,z_{t-k})$,
under the mapping \eqref{matzkinmapping}. Thus if $(\tilde{c},\{\tilde{\fe}_{i}\}_{i=0}^{k},\tilde{\den})$
and $(c,\{\fe_{i}\}_{i=0}^{k},\den)$ are observationally equivalent,
so too are the associated $(\{\tilde{r}_{j}\}_{j=0}^{2},f_{\tilde{U}})$
and $(\{r_{j}\}_{j=0}^{2},f_{U})$. Hence by \propref{matzkin}, in
that case there exists an orthogonal $\Upsilon\in\reals^{p\times p}$
such that
\begin{align*}
\tilde{\fe}_{0}(\bar{z}_{t})-\tilde{c}-\sum_{i=1}^{k}\tilde{\fe}_{i}(\bar{z}_{t-i}) & =\tilde{r}_{0}(\bar{Y})+\tilde{r}_{1}(\bar{X}_{1})+\tilde{r}_{2}(\bar{X}_{2})\\
 & =\Upsilon[r_{0}(\bar{Y})+r_{1}(\bar{X}_{1})+r_{2}(\bar{X}_{2})]=\Upsilon\left[\fe_{0}(\bar{z}_{t})-c-\sum_{i=1}^{k}\fe_{i}(\bar{z}_{t-i})\right]
\end{align*}
for all $\bar{z}_{t-i}\in\reals^{p}$, for all $i\in\{0,\ldots,k\}$.
\end{proof}

\section{Auxiliary lemmas}

\label{app:auxiliary}

As noted in \remref{vecm-representation}\ref{subrem:higher-order},
we shall not generally indicate how our results may specialise in
cases where $k\in\{1,2\}$; the reader is accordingly advised to treat
these models as particular degenerate instances of the VAR($k$) model
with $k\geq3$.
\begin{lem}
\label{lem:lipcondition}Suppose $f:\reals^{d_{x}}\setmap\reals^{d_{y}}$.
Then the following are equivalent:
\begin{enumerate}
\item $f$ is Lipschitz, with Lipschitz constant $K$;
\item for every $x,x^{\prime}\in\reals^{d_{x}}$, there exists a $\beta\in\reals^{d_{x}\times d_{y}}$
with $\smlnorm{\beta}\leq K$ such that
\begin{equation}
f(x)-f(x^{\prime})=\beta^{\trans}(x-x^{\prime})\label{eq:f-x}
\end{equation}
\end{enumerate}
\end{lem}
Recall the definitions of $\{\ga_{i}\}_{i=0}^{k-1}$, $\pi$ and $\co$
given in \eqref{ge}, \eqref{pi-def} and \eqref{co-map}.
\begin{lem}
\label{lem:unstruct}Suppose $(c,\{\fe_{i}\}_{i=0}^{k})$ belongs
to class $\cvar_{r}(\alpha,\abv b,\abv{\rho})$. Then the model $(c^{\ast},\{\fe_{i}^{\ast}\}_{i=0}^{k})$,
with $c^{\ast}=c$ and $\fe_{i}^{\ast}=\fe_{i}\compose\fe_{0}^{-1}$
for $i\in\{0,\ldots,k\}$ belongs to class $\cvar_{r}^{\ast}(\alpha,\abv b,\abv{\rho})$,
with $\ga_{i}^{\ast}=\ga_{i}\compose\fe_{0}^{-1}$ for $i\in\{0,\ldots,k-1\}$,
$\pi^{\ast}=\pi\compose\fe_{0}^{-1}$, and
\[
\co^{\ast}=\begin{bmatrix}\ct^{\ast}\\
\eq^{\ast}
\end{bmatrix}=\begin{bmatrix}\ct\compose\fe_{0}^{-1}\\
\eq\compose\fe_{0}^{-1}
\end{bmatrix}=\co\compose\fe_{0}^{-1}.
\]
Moreover, we may take ${\cal B^{\ast}}={\cal B}$, where ${\cal B}$
and ${\cal B}^{\ast}$ satisfy \ref{enu:cvar:jsr} for models $(c,\{\fe_{i}\}_{i=0}^{k})$
and $(c^{\ast},\{\fe_{i}^{\ast}\}_{i=0}^{k})$ respectively.
\end{lem}
Recall the definition of $\b{\alpha}_{\perp}$ from \eqref{bold-alpha}.
\begin{lem}
\label{lem:invert}Let $\alpha\in\reals^{p\times r}$ with $\rank\alpha=r$,
$\abv b<\infty$ and $\abv{\rho}<1$. Then 
\[
\sup_{\b{\beta}\in{\cal B}}\norm{\begin{bmatrix}\b{\alpha}_{\perp}^{\trans}\\
\b{\beta}^{\trans}
\end{bmatrix}^{-1}}<\infty
\]
\end{lem}
Recall the definition of a bi-Lipschitz function given in \subsecref{high-level}.
\begin{lem}
\label{lem:co-cts}Suppose $(c,\{f_{i}\}_{i=0}^{k}\}$ belongs to
class $\cvar_{r}(\alpha,\abv b,\abv{\rho})$, for some $\abv b<\infty$
and $\abv{\rho}<1$. Then $\fe_{i}$ is continuous for $i\in\{1,\ldots,k-1\}$,
and the map $\co:\reals^{p}\setmap\reals^{p}$ in \eqref{co-map}
is a homeomorphism. If in addition $\fe_{0}$ is bi-Lipschitz, with
bi-Lipschitz constant $C_{\fe_{0}}$, then $\co$ is also bi-Lipschitz,
with bi-Lipschitz constant that depends only on $\alpha$, $\abv b$,
$\abv{\rho}$ and $C_{\fe_{0}}$.
\end{lem}

\section{Proofs of auxiliary lemmas}

\label{app:auxproof}

Since some of the proofs rely on the representation given in \lemref{vecm},
the proof of this result is given first.
\begin{proof}[Proof of \lemref{vecm}]
 By direct calculation,
\begin{align*}
\sum_{i=1}^{k}\fe_{i}(z_{t-i})+\ga_{0}(z_{t-1}) & =-\sum_{i=1}^{k}[\fe_{i}(z_{t-1})-\fe_{i}(z_{t-i})]\\
 & =-\sum_{i=1}^{k}\sum_{j=1}^{i-1}\Delta\fe_{i}(z_{t-j})=-\sum_{j=1}^{k-1}\sum_{i=j+1}^{k}\Delta\fe_{i}(z_{t-j})\\
 & =\sum_{j=1}^{k-1}\Delta\left[-\sum_{i=j+1}^{k}\fe_{i}(z_{t-j})\right]=\sum_{j=1}^{k-1}\Delta\ga_{j}(z_{t-j}),
\end{align*}
whence it follows from \eqref{nlVAR} and the definition of $\pi(z)$
that
\begin{align*}
\Delta\fe_{0}(z_{t}) & =c-\fe_{0}(z_{t-1})+\sum_{i=1}^{k}\fe_{i}(z_{t-i})+u_{t}\\
 & =c-[\fe_{0}(z_{t-1})+\ga_{0}(z_{t-1})]+\sum_{i=1}^{k-1}\Delta\ga_{i}(z_{t-i})+u_{t}\\
 & =c+\pi(z_{t-1})+\sum_{i=1}^{k-1}\Delta\ga_{i}(z_{t-i})+u_{t}
\end{align*}
which yields \eqref{VECM-first}. Further, since $\zeta_{j,t}=\sum_{i=j}^{k-1}\ga_{i}(z_{t-i+j})$,
\begin{align*}
\sum_{i=1}^{k-1}\Delta\ga_{i}(z_{t-i}) & =\sum_{i=1}^{k-1}\ga_{i}(z_{t-i})-\sum_{i=1}^{k-1}\ga_{i}(z_{t-1-i})\\
 & =\ga_{1}(z_{t-1})+\sum_{i=2}^{k-1}\ga_{i}(z_{t-i})-\sum_{i=1}^{k-1}\ga_{i}(z_{t-1-i})\\
 & =\ga_{1}(z_{t-1})-\sum_{i=1}^{k-1}\ga_{i}(z_{(t-2)-i+1})+\sum_{i=2}^{k-1}\ga_{i}(z_{(t-2)-i+2})\\
 & =\ga_{1}(z_{t-1})-\zeta_{1,t-2}+\zeta_{2,t-2}=\ga_{1}(z_{t-1})+D_{1}\b{\zeta}_{t-2}
\end{align*}
so that
\[
\Delta\fe_{0}(z_{t})=c+[\pi(z_{t-1})+\ga_{1}(z_{t-1})]+D_{1}\b{\zeta}_{t-2}+u_{t}
\]
as per the first $p$ elements of \eqref{VECM-system}. Finally, we
note that for $j\in\{1,\ldots,k-2\}$,
\begin{align*}
\Delta\zeta_{j,t-1} & =\zeta_{j,t-1}-\zeta_{j,t-2}=\sum_{i=j}^{k-1}\ga_{i}(z_{(t-1)-i+j})-\zeta_{j,t-2}\\
 & =\ga_{j}(z_{t-1})+\sum_{i=j+1}^{k-1}\ga_{i}(z_{(t-1)-i+j})-\zeta_{j,t-2}\\
 & =\ga_{j}(z_{t-1})+\sum_{i=j+1}^{k-1}\ga_{i}(z_{(t-2)-i+(j+1)})-\zeta_{j,t-2}\\
 & =\ga_{j}(z_{t-1})-\zeta_{j,t-2}+\zeta_{j+1,t-2}
\end{align*}
and
\[
\Delta\zeta_{k-1,t-1}=\ga_{k-1}(z_{t-1})-\zeta_{k-1,t-2},
\]
so that stacking these equations as $j$ ranges over $\{1,\ldots,k-1\}$
yields the remaining elements of \eqref{VECM-system}.
\end{proof}
\begin{proof}[Proof of \lemref{lipcondition}]
(ii) trivially implies (i); it remains to prove the reverse implication.
Let $x,x^{\prime}\in\reals^{d_{y}}$ with $x\neq x^{\prime}$ be given,
and let $\delta_{x}\defeq x-x^{\prime}\neq0$ and $\delta_{f}\defeq f(x)-f(x^{\prime})$,
so that $\smlnorm{\delta_{f}}\leq K\smlnorm{\delta_{x}}$ by (i).
Taking $\beta^{\trans}\defeq\delta_{f}(\delta_{x}^{\trans}\delta_{x})^{-1}\delta_{x}^{\trans}$
satisfies \eqref{f-x}. Finally, we note
\[
\smlnorm{\beta^{\trans}}=\smlnorm{\beta}=\sup_{\smlnorm z=1}\smlnorm{\delta_{f}(\delta_{x}^{\trans}\delta_{x})^{-1}\delta_{x}^{\trans}z}=\smlnorm{\delta_{f}}\smlabs{\delta_{x}^{\trans}\delta_{x}}^{-1}\sup_{\smlnorm z=1}\smlabs{\delta_{x}^{\trans}z}=_{(1)}\smlnorm{\delta_{f}}/\smlnorm{\delta_{x}}\leq K,
\]
where $=_{(1)}$ holds because the maximum is achieved by taking $z=\delta_{x}/\smlnorm{\delta_{x}}$. 
\end{proof}
\begin{proof}[Proof of \lemref{unstruct}]
The claimed forms of $\{\ga_{i}^{\ast}\}_{i=1}^{k-1}$, $\pi^{\ast}$
and $\co^{\ast}$ follow immediately from the definitions of these
objects, e.g.\ from \eqref{ge} we have
\[
\ga_{j}^{\ast}=-\sum_{i=j+1}^{k}\fe_{i}^{\ast}=-\sum_{i=j+1}^{k}\fe_{i}\compose\fe_{0}^{-1}=\left(-\sum_{i=j+1}^{k}\fe_{i}\right)\compose\fe_{0}^{-1}=\ga_{j}\compose\fe_{0}^{-1},
\]
etc. It remains to verify the three conditions required for $(c^{\ast},\{\fe_{i}^{\ast}\}_{i=0}^{k})$
to belong to class $\cvar_{r}^{\ast}(\alpha,\abv b,\abv{\rho})$,
as per \defref{cvar}.

\textbf{\enuref{cvar:homeo}.} $\fe_{0}^{\ast}$ is the identity by
construction, and thus trivially a homeomorphism.

\textbf{\enuref{cvar:rank}.} Holds trivially since $c^{\ast}=c=\alpha\mu$
and 
\[
\pi^{\ast}(z)=\pi[\fe_{0}^{-1}(z)]=\alpha\eq[\fe_{0}^{-1}(z)]=\alpha\eq^{\ast}(z).
\]
It follows, moreover, that $\alpha^{\ast}=\alpha$.

\textbf{\enuref{cvar:jsr}.} We have $\b{\eq}^{\ast}\compose(\fe_{0}^{\ast})^{-1}=\b{\eq}^{\ast}=\b{\eq}\compose\fe_{0}^{-1}$,
and so
\[
[\b{\eq}^{\ast}\compose(\fe_{0}^{\ast})^{-1}](z)+\b D_{0}\b z=(\b{\eq}\compose\fe_{0}^{-1})(z)+\b D_{0}\b z
\]
for every $\b z=(z^{\trans},\b{\zeta}^{\trans})^{\trans}\in\reals^{kp}$.
Since the r.h.s.\ satisfies \eqref{gradient}, and $\alpha^{\ast}=\alpha$,
we may deduce that \enuref{cvar:jsr} holds for $(c^{\ast},\{\fe_{i}^{\ast}\}_{i=0}^{k})$
with the same ${\cal B}$, $\abv b$ and $\abv{\rho}$ as for $(c,\{\fe_{i}\}_{i=0}^{k})$.
\end{proof}
\begin{proof}[Proof of \lemref{invert}]
Let $\delta>0$ be such that $\abv{\rho}+\delta<1$, and ${\cal B}\in\cvar_{r}(\alpha,\abv b,\abv{\rho})$.
For any $\b{\beta}\in{\cal B}$,
\[
\begin{bmatrix}\b{\alpha}_{\perp}^{\trans}\\
\b{\beta}^{\trans}
\end{bmatrix}=\begin{bmatrix}I_{r} & 0\\
\b{\beta}^{\trans}\b{\alpha}_{\perp}(\b{\alpha}_{\perp}^{\trans}\b{\alpha}_{\perp})^{-1} & \b{\beta}^{\trans}\b{\alpha}
\end{bmatrix}\begin{bmatrix}\b{\alpha}_{\perp}^{\trans}\\
(\b{\alpha}^{\trans}\b{\alpha})^{-1}\b{\alpha}^{\trans}
\end{bmatrix},
\]
where $\b{\alpha}$ and $\b{\alpha}_{\perp}$ are as defined in \eqref{bold-alpha}.
By construction, $\b{\alpha}\in\reals^{kp\times[p(k-1)+r]}$ and $\b{\alpha}_{\perp}\in\reals^{kp\times q}$
have full column rank, and $\b{\alpha}_{\perp}^{\trans}\b{\alpha}=0$.
They thus span the whole of $\reals^{kp}$, whence the second r.h.s.\ matrix
is invertible. It will be shown below that $\b{\beta}^{\trans}\b{\alpha}$
is invertible, and hence so too is the first r.h.s.\ matrix. Therefore
\begin{align*}
\begin{bmatrix}\b{\alpha}_{\perp}^{\trans}\\
\b{\beta}^{\trans}
\end{bmatrix}^{-1} & =\begin{bmatrix}\b{\alpha}_{\perp}^{\trans}\\
(\b{\alpha}^{\trans}\b{\alpha})^{-1}\b{\alpha}^{\trans}
\end{bmatrix}^{-1}\begin{bmatrix}I_{r} & 0\\
-(\b{\beta}^{\trans}\b{\alpha})^{-1}\b{\beta}^{\trans}\b{\alpha}_{\perp}(\b{\alpha}_{\perp}^{\trans}\b{\alpha}_{\perp})^{-1} & (\b{\beta}^{\trans}\b{\alpha})^{-1}
\end{bmatrix}\\
 & =\begin{bmatrix}\b{\alpha}_{\perp}^{\trans}\\
(\b{\alpha}^{\trans}\b{\alpha})^{-1}\b{\alpha}^{\trans}
\end{bmatrix}^{-1}\begin{bmatrix}I_{r} & 0\\
0 & (\b{\beta}^{\trans}\b{\alpha})^{-1}
\end{bmatrix}\begin{bmatrix}I_{r} & 0\\
-\b{\beta}^{\trans}\b{\alpha}_{\perp}(\b{\alpha}_{\perp}^{\trans}\b{\alpha}_{\perp})^{-1} & I_{p(k-1)}
\end{bmatrix}\\
 & \eqdef M_{1}M_{2}(\b{\beta})M_{3}(\b{\beta}).
\end{align*}

Let $\abv B(0,\abv b)$ denote the closed ball centred at zero, of
radius $\abv b$, in $\reals^{pk\times[p(k-1)+r]}$, and
\[
{\cal K}\defeq\{I_{p(k-1)+r}+\b{\beta}^{\trans}\b{\alpha}\mid\b{\beta}\in\abv B(0,\abv b)\}.
\]
Then $\{I_{p(k-1)+r}+\b{\beta}^{\trans}\b{\alpha}\mid\b{\beta}\in{\cal B}\}$
is a closed subset of ${\cal K}$. Suppose that ${\cal B}\subset\reals^{m\times m}$
is compact, with $\rho_{\jsr}({\cal B})\leq\abv{\rho}$. It follows
by \citet[Prop.~1.4]{Jungers09} that that there exists a norm $\smlnorm{\cdot}_{\delta}$
on $\reals^{p(k-1)+r}$ such that
\[
\max_{\b{\beta}\in{\cal B}}\smlnorm{I_{p(k-1)+r}+\b{\beta}^{\trans}\b{\alpha}}_{\delta}\leq\abv{\rho}+\delta<1.
\]
Hence, by the equivalence of $\smlnorm{\cdot}_{\delta}$ and $\smlnorm{\cdot}$,
there exists a $C_{0}<\infty$ such that for all $\b{\beta}\in{\cal B}$,
\begin{multline*}
\smlnorm{(\b{\beta}^{\trans}\b{\alpha})^{-1}}\leq_{(1)}C_{0}\smlnorm{(\b{\beta}^{\trans}\b{\alpha})^{-1}}_{\delta}=C_{0}\norm{\sum_{k=0}^{\infty}(I+\b{\beta}^{\trans}\b{\alpha})^{k}}_{\delta}\\
\leq C_{0}\sum_{k=0}^{\infty}\smlnorm{I+\b{\beta}^{\trans}\b{\alpha}}_{\delta}^{k}\leq\frac{C_{0}}{1-(\abv{\rho}+\delta)}<\infty,
\end{multline*}
where $\leq_{(1)}$ holds by Theorem~5.6.18 in \citet{HJ13book}.
Thus $\b{\beta}^{\trans}\b{\alpha}$ is indeed invertible as claimed.
Since the final bound is independent of ${\cal B}$, it follows that
\[
\sup_{\b{\beta}\in{\cal B}}\smlnorm{M_{2}(\b{\beta})}<\infty,
\]
while also
\[
\sup_{\b{\beta}\in{\cal B}}\smlnorm{M_{3}(\b{\beta})}\leq\sup_{\b{\beta}\in\abv B(0,\abv b)}\smlnorm{M_{3}(\b{\beta})}<\infty.
\]
 Since $M_{1}$ does not depend on $\b{\beta}$, the claim follows.
\end{proof}
\begin{proof}[Proof of \lemref{co-cts}]
We first suppose that $\fe_{0}=I_{p}$, i.e.\ the model is in class
$\cvar_{r}^{\ast}(\alpha,\abv b,\abv{\rho})$. We shall show in this
case that each $\{\fe_{i}\}_{i=1}^{k}$ is Lipschitz; and that $\co:\reals^{p}\setmap\reals^{p}$
is bi-Lipschitz, and therefore a homeomorphism.

To that end, define $\b{\co}:\reals^{kp}\setmap\reals^{kp}$ as
\begin{equation}
\b{\co}(\b z)\defeq\begin{bmatrix}\b{\alpha}_{\perp}^{\trans}\b z\\
\b{\eq}(z)+\b D_{0}\b z
\end{bmatrix}=\begin{bmatrix}\alpha_{\perp}^{\trans}[z-E^{\trans}\b{\zeta}]\\
\eq(z)\\
\b{\ga}(z)+D\b{\zeta}
\end{bmatrix}\label{eq:bold-chi}
\end{equation}
where $\b z=(z^{\trans},\b{\zeta}^{\trans})^{\trans}$ for $z\in\reals^{p}$
and $\b{\zeta}\in\reals^{p(k-1)}$. By evaluating \eqref{bold-chi}
when $\b{\zeta}=0$, and applying \lemref{lipcondition}, we see that
\ref{enu:cvar:jsr} implies that $\eq$ and $\b{\ga}$ are Lipschitz,
with Lipschitz constant depending only on $\abv b$. Hence $\ga_{i}$
is Lipschitz for $i\in\{1,\ldots,k-1\}$, and moreover so too is
\[
\ga_{0}(z)=-\pi(z)-\fe_{0}(z)=-\alpha\eq(z)-z.
\]
Since $\fe_{i}=\ga_{i}-\ga_{i-1}$ by \eqref{ge}, it follows that
$\fe_{i}$ is Lipschitz for $i\in\{1,\ldots,k\}$.

We now turn to $\co$. Let $z_{0}\in\reals^{p}$, $\b{\zeta}_{0}\defeq-D^{-1}\b{\ga}(z_{0})$,
and $\b z_{0}\defeq(z_{0}^{\trans},\b{\zeta}_{0}^{\trans})^{\trans}$.
Then
\[
E^{\trans}\b{\zeta}_{0}=-E^{\trans}D^{-1}\b{\ga}(z_{0})=\begin{bmatrix}I_{p} & 0_{p\times p} & \cdots & 0_{p\times p}\end{bmatrix}\begin{bmatrix}I_{p} & I_{p} & I_{p} & I_{p}\\
 & I_{p} & I_{p} & I_{p}\\
 &  & \ddots & \vdots\\
 &  &  & I_{p}
\end{bmatrix}\begin{bmatrix}\ga_{1}(z_{0})\\
\ga_{2}(z_{0})\\
\vdots\\
\ga_{k-1}(z_{0})
\end{bmatrix}=\sum_{i=1}^{k-1}\ga_{i}(z_{0})
\]
and hence
\[
\b{\co}(\b z_{0})=\b{\co}(z_{0},\b{\zeta}_{0})=\begin{bmatrix}\alpha_{\perp}^{\trans}[z_{0}-\sum_{i=1}^{k-1}\ga_{i}(z_{0})]\\
\eq(z_{0})\\
0_{p(k-1)}
\end{bmatrix}=\begin{bmatrix}\ct(z_{0})\\
\eq(z_{0})\\
0_{p(k-1)}
\end{bmatrix}=\begin{bmatrix}\co(z_{0})\\
0_{p(k-1)}
\end{bmatrix}.
\]
Let $z_{1}\in\reals^{p}$, $\b{\zeta}_{1}\defeq-D^{-1}\b{\ga}(z_{1})$
and $\b z_{1}\defeq(z_{1}^{\trans},\b{\zeta}_{1}^{\trans})^{\trans}$.
Then by the preceding (recalling $\smlnorm{\cdot}$ denotes the Euclidean
norm),
\[
\smlnorm{\b{\co}(\b z_{0})-\b{\co}(\b z_{1})}=\smlnorm{\co(z_{0})-\co(z_{1})},
\]
while by \ref{enu:cvar:jsr} (with $\fe_{0}=I_{p}$), there exists
a $\b{\beta}\in{\cal B}$ such that
\begin{equation}
\b{\co}(\b z_{0})-\b{\co}(\b z_{1})=\begin{bmatrix}\b{\alpha}_{\perp}^{\trans}\\
\b{\beta}^{\trans}
\end{bmatrix}(\b z_{0}-\b z_{1}).\label{eq:boldco}
\end{equation}
Since $\smlnorm{\b{\beta}}\leq\abv b$ for all $\b{\beta}\in{\cal B}$,
there exists a $C_{0}<\infty$ (depending only on $\abv b$) such
that
\begin{align}
\smlnorm{\co(z_{0})-\co(z_{1})} & =\smlnorm{\b{\co}(\b z_{0})-\b{\co}(\b z_{1})}\leq C_{0}\smlnorm{\b z_{0}-\b z_{1}}.\label{eq:co-diff-1}
\end{align}
As noted above, $\b{\ga}$ is Lipschitz: hence there exists a $C_{1}<\infty$
(depending only on $\abv b$) such that 
\begin{equation}
\smlnorm{\b z_{0}-\b z_{1}}=\norm{\begin{bmatrix}z_{0}-z_{1}\\
-D^{-1}[\b{\ga}(z_{0})-\b{\ga}(z_{1})]
\end{bmatrix}}\leq C_{1}\smlnorm{z_{0}-z_{1}}.\label{eq:co-diff-2}
\end{equation}
By \lemref{invert}, the matrix on the r.h.s.\ of \eqref{boldco}
is invertible; hence
\[
\smlnorm{\b z_{0}-\b z_{1}}\leq\norm{\begin{bmatrix}\b{\alpha}_{\perp}^{\trans}\\
\b{\beta}^{\trans}
\end{bmatrix}^{-1}}\smlnorm{\b{\co}(\b z_{0})-\b{\co}(\b z_{1})}.
\]
That lemma further implies that there exists a $C_{2}<\infty$, depending
only on $\alpha$, $\abv b$ and $\abv{\rho}$, such that 
\begin{align}
\smlnorm{\co(z_{0})-\co(z_{1})} & =\smlnorm{\b{\co}(\b z_{0})-\b{\co}(\b z_{1})}\nonumber \\
 & \geq\inf_{\b{\beta}\in{\cal B}}\norm{\begin{bmatrix}\b{\alpha}_{\perp}^{\trans}\\
\b{\beta}^{\trans}
\end{bmatrix}^{-1}}^{-1}\smlnorm{\b z_{0}-\b z_{1}}\nonumber \\
 & =\left(\sup_{\b{\beta}\in{\cal B}}\norm{\begin{bmatrix}\b{\alpha}_{\perp}^{\trans}\\
\b{\beta}^{\trans}
\end{bmatrix}^{-1}}\right)^{-1}\smlnorm{\b z_{0}-\b z_{1}}\nonumber \\
 & \geq C_{2}^{-1}\smlnorm{\b z_{0}-\b z_{1}}\nonumber \\
 & \geq_{(1)}C_{2}^{-1}\smlnorm{z_{0}-z_{1}}\label{eq:co-diff-3}
\end{align}
where $\geq_{(1)}$ follows since
\[
\smlnorm{\b z_{0}-\b z_{1}}^{2}=\smlnorm{z_{0}-z_{1}}^{2}+\smlnorm{\b{\zeta}_{0}-\b{\zeta}_{1}}^{2}\geq\smlnorm{z_{0}-z_{1}}^{2}.
\]
Deduce from \eqref{boldco}--\eqref{co-diff-3} that there exists
a $C<\infty$, depending only on $\alpha$, $\abv b$ and $\abv{\rho}$,
such that
\[
C^{-1}\smlnorm{z_{0}-z_{1}}\leq\smlnorm{\co(z_{0})-\co(z_{1})}\leq C\smlnorm{z_{0}-z_{1}}
\]
for all $z_{0},z_{1}\in\reals^{p}$. Hence $\co$ is bi-Lipschitz,
as required. 

Finally, suppose that $\fe_{0}$ is not the identity, but merely a
homeomorphism (as per \ref{enu:cvar:homeo}). By \lemref{unstruct},
the model $(c^{\ast},\{\fe_{i}^{\ast}\}_{i=0}^{k})$, with $c^{\ast}=c$
and $\fe_{i}^{\ast}=\fe_{i}\compose\fe_{0}^{-1}$ for $i\in\{0,\ldots,k\}$
is in class $\cvar_{r}^{\ast}(\alpha,\abv b,\abv{\rho})$, and it
follows by the preceding that each $\{\fe_{i}^{\ast}\}_{i=1}^{k}$
is Lipschitz, and the associated $\co^{\ast}$ is bi-Lipschitz. Clearly
$\fe_{i}=\fe_{i}^{\ast}\compose\fe_{0}$ is continuous for each $i\in\{1,\ldots,k\}$.
Moreover $\co^{\ast}=\co\compose\fe_{0}^{-1}$, whence $\co=\co^{\ast}\compose\fe_{0}$
is also a homeomorphism. If $\fe_{0}$ is bi-Lipschitz with bi-Lipschitz
constant $C_{\fe_{0}}$, then $\co=\co^{\ast}\compose\fe_{0}$ is
bi-Lipschitz, with bi-Lipschitz constant depending only on $C_{\fe_{0}}$,
$\alpha$, $\abv b$ and $\abv{\rho}$.
\end{proof}

\section{Proofs of high-level results}

This section provides proofs of our main results under the high-level
conditions provided by the definition of class $\cvar_{r}$ (see \defref{cvar}). 
\begin{proof}[Proof of \thmref{gjrt}]
 We note that $\co$ is a homeomorphism by \lemref{co-cts}.

We first prove the result when $\fe_{0}=I_{p}$. By \lemref{vecm},
we can rewrite the model as
\begin{align}
\Delta\b z_{t} & =\b c+\b{\pi}(z_{t-1})+\b D\b z_{t-1}+\b u_{t}\label{eq:DZ}
\end{align}
where
\begin{equation}
\b{\pi}(z)=\begin{bmatrix}\pi(z)+\ga_{1}(z)\\
\b{\ga}(z)
\end{bmatrix}=\begin{bmatrix}\alpha & E^{\trans}\\
0 & I_{p(k-1)}
\end{bmatrix}\begin{bmatrix}\eq(z)\\
\b{\ga}(z)
\end{bmatrix}=\b{\alpha}\b{\eq}(z),\label{eq:PI}
\end{equation}
$\b D=\b{\alpha}\b D_{0}$, and
\begin{equation}
\b c=\begin{bmatrix}c\\
0_{p(k-1)}
\end{bmatrix}=\begin{bmatrix}\alpha & E^{\trans}\\
0 & I_{p(k-1)}
\end{bmatrix}\begin{bmatrix}\cn\\
0_{p(k-1)}
\end{bmatrix}\eqdef\b{\alpha}\b{\cn}.\label{eq:Crep}
\end{equation}

To obtain \eqref{gjrt-z}, we first evaluate each of $\eq(z_{t})$
and $\ct(z_{t})$ in turn. Regarding the former, we observe that by
\eqref{DZ}--\eqref{Crep},
\begin{align*}
\Delta\b z_{t} & =\b{\alpha}[\b{\cn}+\b{\eq}(z_{t-1})+\b D_{0}\b z_{t-1}]+\b u_{t}=\b{\alpha}\b{\xi}_{t-1}+\b u_{t}.
\end{align*}
Since the VAR belongs to class $\cvar_{r}$, for each $t\in\naturals$
there exists a $\b{\beta}_{t}\in{\cal B}$ such that
\[
\Delta\b{\xi}_{t}=[\b{\eq}(z_{t})+\b D_{0}\b z_{t}]-[\b{\eq}(z_{t-1})+\b D_{0}\b z_{t-1}]=\b{\beta}_{t}^{\trans}\Delta\b z_{t},
\]
and hence
\begin{equation}
\b{\xi}_{t}=(I_{p(k-1)+r}+\b{\beta}_{t}^{\trans}\b{\alpha})\b{\xi}_{t-1}+\b{\beta}_{t}^{\trans}\b u_{t}\label{eq:xi-lom}
\end{equation}
as per \eqref{gjrt-xi}. By definition of $\b{\xi}_{t}$ and $\xi_{t}$,
\begin{equation}
\b{\xi}_{t}=\b{\mu}+\b{\eq}(z_{t-1})+\b D_{0}\b z_{t-1}=\begin{bmatrix}\mu+\eq(z_{t})\\
\b{\ga}(z_{t})+D\b{\zeta}_{t-1}
\end{bmatrix}=\begin{bmatrix}\xi_{t}\\
\Delta\b{\zeta_{t}}
\end{bmatrix}\label{eq:xi-components}
\end{equation}
where the final equality follows from the final $p(k-1)$ equations
in \eqref{VECM-system}, whence
\begin{equation}
\eq(z_{t})=-\mu+S_{\eq}^{\trans}\b{\xi}_{t}\label{eq:eqrep}
\end{equation}
for $S_{\eq}^{\trans}\defeq[I_{r},0_{r\times p(k-1)}].$

We next turn to $\ct(z_{t})$. By \eqref{DZ}--\eqref{Crep}, and
the definition of $\b{\alpha}_{\perp}$ in \eqref{bold-alpha},
\[
\Delta\alpha_{\perp}^{\trans}(z_{s}-\zeta_{1,s-1})=\alpha_{\perp}^{\trans}\begin{bmatrix}I_{p}, & -E^{\trans}\end{bmatrix}\begin{bmatrix}\Delta z_{s}\\
\Delta\b{\zeta}_{s-1}
\end{bmatrix}=\b{\alpha}_{\perp}^{\trans}\begin{bmatrix}\Delta z_{s}\\
\Delta\b{\zeta}_{s-1}
\end{bmatrix}=\b{\alpha}_{\perp}^{\trans}\b u_{s}=\alpha_{\perp}^{\trans}u_{s},
\]
whence
\begin{align}
\alpha_{\perp}^{\trans}(z_{t}-\zeta_{1,t-1}) & =\alpha_{\perp}^{\trans}(z_{0}-\zeta_{1,-1})+\sum_{s=1}^{t}\Delta\alpha_{\perp}^{\trans}(z_{s}-\zeta_{1,s-1})\nonumber \\
 & =\alpha_{\perp}^{\trans}\left[z_{0}-\sum_{i=1}^{k-1}\ga_{i}(z_{-i})\right]+\alpha_{\perp}^{\trans}\sum_{s=1}^{t}u_{s}\nonumber \\
 & =\alpha_{\perp}^{\trans}\abv{\srfn}(\vec{z}_{0})+\alpha_{\perp}^{\trans}\sum_{s=1}^{t}u_{s}.\label{eq:ctsum}
\end{align}
We need to relate the l.h.s.\ to
\[
\ct(z_{t})=\alpha_{\perp}^{\trans}\left[z_{t}-\sum_{i=1}^{k-1}\ga_{i}(z_{t})\right]=\alpha_{\perp}^{\trans}[z_{t}+E^{\trans}D^{-1}\b{\ga}(z_{t})],
\]
where the second equality follows since
\begin{equation}
-D^{-1}=\begin{bmatrix}I_{p} & -I_{p}\\
 & \ddots & \ddots\\
 &  & I_{p} & -I_{p}\\
 &  &  & I_{p}
\end{bmatrix}^{-1}=\begin{bmatrix}I_{p} & I_{p} & I_{p} & I_{p}\\
 & \ddots & \vdots & \vdots\\
 &  & I_{p} & I_{p}\\
 &  &  & I_{p}
\end{bmatrix}.\label{eq:Dinv}
\end{equation}
To that end, we note that
\begin{align*}
\ct(z_{t})-\alpha_{\perp}^{\trans}(z_{t}-\zeta_{1,t-1}) & =\alpha_{\perp}^{\trans}[E^{\trans}D^{-1}\b{\ga}(z_{t})+\zeta_{1,t-1}]\\
 & =\alpha_{\perp}^{\trans}E^{\trans}[D^{-1}\b{\ga}(z_{t})+\b{\zeta}_{t-1}]\\
 & =\alpha_{\perp}^{\trans}E^{\trans}D^{-1}[\b{\ga}(z_{t})+D\b{\zeta}_{t-1}]\\
 & =\alpha_{\perp}^{\trans}S_{\ct}^{\trans}\b{\xi}_{t}
\end{align*}
for where the final equality follows by \eqref{xi-components}, for
$S_{\ct}^{\trans}$ the $p\times p(k-1)+r$ matrix of the form
\[
S_{\ct}^{\trans}\defeq\begin{bmatrix}0_{p\times r} & E^{\trans}D^{-1}\end{bmatrix}=-\begin{bmatrix}0_{p\times r} & I_{p} & \cdots & I_{p}\end{bmatrix}.
\]
Hence by \eqref{ctsum},
\begin{equation}
\ct(z_{t})=\alpha_{\perp}^{\trans}\abv{\srfn}(\vec{z}_{0})+\alpha_{\perp}^{\trans}\sum_{s=1}^{t}u_{s}+\alpha_{\perp}^{\trans}S_{\ct}^{\trans}\b{\xi}_{t}\label{eq:ctrep}
\end{equation}
It then follows from \eqref{eqrep} and \eqref{ctrep} that
\[
\co(z_{t})=\begin{bmatrix}\ct(z_{t})\\
\eq(z_{t})
\end{bmatrix}=\begin{bmatrix}\alpha_{\perp}^{\trans}\abv{\srfn}(\vec{z}_{0})+\alpha_{\perp}^{\trans}\sum_{s=1}^{t}u_{s}\\
-\mu
\end{bmatrix}+\begin{bmatrix}\alpha_{\perp}^{\trans}S_{\ct}^{\trans}\\
S_{\eq}^{\trans}
\end{bmatrix}\b{\xi}_{t}
\]
which yields \eqref{gjrt-z} with $S_{\co}^{\trans}=[S_{\ct}\alpha_{\perp},S_{\theta}]^{\trans}$.

Finally, suppose that $\fe_{0}$ is a general homeomorphism. Observe
that $z_{t}^{\ast}\defeq\fe_{0}(z_{t})$ satisfies
\[
z_{t}^{\ast}=c+\sum_{i=1}^{k}(\fe_{i}\compose\fe_{0}^{-1})(z_{t-i}^{\ast})+u_{t}\eqdef c^{\ast}+\sum_{i=1}^{k}\fe_{i}^{\ast}(z_{t-i}^{\ast})+u_{t}
\]
i.e.\ it is generated by a model $(c^{\ast},\{\fe_{i}^{\ast}\}_{i=0}^{k})$
with $\fe_{0}^{\ast}=I_{p}$. By \lemref{unstruct}, this model belongs
to class $\cvar_{r}^{\ast}$, with the same $\alpha$ and $\mu$.
Therefore, by the arguments given above, \eqref{gjrt-z} holds in
the form
\[
\co^{\ast}(z_{t}^{\ast})=\begin{bmatrix}\alpha_{\perp}^{\trans}\abv{\srfn}^{\ast}(\vec{z}_{0}^{\ast})+\alpha_{\perp}^{\trans}\sum_{s=1}^{t}u_{s}\\
-\mu
\end{bmatrix}+S_{\co}^{\trans}\b{\xi}_{t}^{\ast}
\]
It additionally follows from \lemref{unstruct} that the l.h.s.\ is
equal to $\co(z_{t})$, while
\[
\abv{\srfn}^{\ast}(\vec{z}_{0}^{\ast})=z_{0}^{\ast}-\sum_{i=1}^{k-1}\ga_{i}^{\ast}(z_{-i}^{\ast})=\fe_{0}(z_{0})-\sum_{i=1}^{k-1}\ga_{i}(z_{-i})=\abv{\srfn}(\vec{z}_{0})
\]
and, recalling \eqref{xi-components} above,
\begin{equation}
\b{\xi}_{t}^{\ast}=\begin{bmatrix}\mu+\eq^{\ast}(z_{t}^{\ast})\\
\Delta\b{\zeta}_{t}^{\ast}
\end{bmatrix}=\begin{bmatrix}\mu+\eq(z_{t})\\
\Delta\b{\zeta}_{t}
\end{bmatrix}=\b{\xi}_{t}\label{eq:xi-ast-xi}
\end{equation}
since $\zeta_{j,t}^{\ast}=\sum_{i=j}^{k-1}\ga_{i}^{\ast}(z_{t-i+j}^{\ast})=\sum_{i=j}^{k-1}\ga_{i}(z_{t-i+j})=\zeta_{j,t}$.
Since \eqref{xi-lom} holds as
\[
\b{\xi}_{t}^{\ast}=(I_{p(k-1)+r}+\b{\beta}_{t}^{\trans}\b{\alpha})\b{\xi}_{t-1}^{\ast}+\b{\beta}_{t}^{\trans}\b u_{t}
\]
where $\b{\beta}_{t}\in{\cal B}^{\ast}$ for each $t\in\naturals$,
it then follows by \eqref{xi-ast-xi} and \lemref{unstruct} that
\[
\b{\xi}_{t}=(I_{p(k-1)+r}+\b{\beta}_{t}^{\trans}\b{\alpha})\b{\xi}_{t-1}+\b{\beta}_{t}^{\trans}\b u_{t}
\]
where $\b{\beta}_{t}\in{\cal B}$ for each $t\in\naturals$, since
${\cal B}={\cal B}^{\ast}$.
\end{proof}
\begin{proof}[Proof of \thmref{stab}]
 Without loss of generality, we may take $\tau=1$. We note that
$\co$ is a homeomorphism by \lemref{co-cts}.

\textbf{\ref{enu:stab:cvg}.} Regarding \eqref{zt-det-lim}, we have
from \thmref{gjrt} with $\vec{z}_{0}=\state$ and $\sum_{s=1}^{t}u_{s}=u$
that
\begin{align*}
\co[z_{t}(u;\state)] & =\begin{bmatrix}\alpha_{\perp}^{\trans}\abv{\srfn}(\state)+\alpha_{\perp}^{\trans}u\\
-\mu
\end{bmatrix}+S_{\co}^{\trans}\b{\xi}_{t},
\end{align*}
where $\b{\xi}_{t}$ satisfies \eqref{gjrt-z} with $\b u_{t}=0$
for all $t\geq2$; whence
\[
\b{\xi}_{t}=\left[\prod_{s=2}^{t}(I_{p(k-1)+r}+\b{\beta}_{s}^{\trans}\b{\alpha})\right]\b{\xi}_{1}
\]
for some $\{\b{\beta}_{s}\}\subset{\cal B}$. Since $\rho_{\jsr}(\{I_{p(k-1)+r}+\b{\beta}^{\trans}\b{\alpha}\mid\b{\beta}\in{\cal B}\})\leq\abv{\rho}<1$,
it follows by \citet[Prop.~1.4]{Jungers09} that there exists a norm
$\smlnorm{\cdot}_{\ast}$ such that
\[
\smlnorm{\b{\xi}_{t}}_{\ast}\leq\left[\prod_{s=2}^{t}\smlnorm{I_{p(k-1)+r}+\b{\beta}_{s}^{\trans}\b{\alpha}}_{\ast}\right]\smlnorm{\b{\xi}_{1}}_{\ast}\leq\abv{\rho}^{t-1}\smlnorm{\b{\xi}_{1}}_{\ast}\goesto0
\]
as $t\goesto\infty$. \eqref{zt-det-lim} then follows.

\textbf{\ref{enu:stab:dom}.} By \eqref{zt-det-lim}, $u\in\set U(z;\state)$
if and only if
\[
\begin{bmatrix}\alpha_{\perp}^{\trans}\srfn(z)\\
\eq(z)
\end{bmatrix}=\co(z)=\begin{bmatrix}\alpha_{\perp}^{\trans}[\abv{\srfn}(\state)+u]\\
-\mu
\end{bmatrix},
\]
whence
\begin{align*}
\set U(z;\state) & =\{u\in\reals^{p}\mid\alpha_{\perp}^{\trans}[u+\abv{\srfn}(\state)-\srfn(z)]=0\}\\
 & =(\spn\alpha)+[\srfn(z)-\abv{\srfn}(\state)].
\end{align*}
Since $\set U(z;\state)$ is non-empty for every $(z,\state)\in\ctspc_{\mu}\times\reals^{kp}$,
it follows that $\ctspc_{\mu}$ is strictly stable.
\end{proof}
\begin{proof}[Proof of \thmref{multipliers}]
\textbf{\ref{enu:mult:set}.} We note that $z_{\infty}(u;\state)$
depends on $\state$ only through $\alpha_{\perp}^{\trans}\abv{\srfn}(\state)$,
and that setting
\[
z_{\state}\defeq\co^{-1}\begin{bmatrix}\alpha_{\perp}^{\trans}\abv{\srfn}(\state)\\
-\mu
\end{bmatrix}
\]
yields a $z_{\state}\in\ctspc_{\mu}$ with $\ct(z_{\state})=\alpha_{\perp}^{\trans}\abv{\srfn}(\state)$.
Since the l.h.s.\ of 
\[
z_{\infty}(u;\state)=\co^{-1}\begin{bmatrix}\alpha_{\perp}^{\trans}\abv{\srfn}(\state)+\alpha_{\perp}^{\trans}u\\
-\mu
\end{bmatrix}=\co^{-1}\begin{bmatrix}\ct(z)+\alpha_{\perp}^{\trans}u\\
-\mu
\end{bmatrix}
\]
is differentiable with respect to $u$ (at $u=0$) if and only if
the r.h.s.\ is, the long-run multipliers are only well defined when
$z_{\state}\in\ctspc_{\mu}\backslash N$, and the full set of long-run
multipliers can thus be recovered by computing 
\[
\Theta_{\infty}(z)\defeq\partial_{u}\left.\co^{-1}\begin{bmatrix}\ct(z)+\alpha_{\perp}^{\trans}u\\
-\mu
\end{bmatrix}\right|_{u=0}
\]
for each $z\in\ctspc_{\mu}\backslash N$. The second equality in \eqref{lrmult}
then follows by the chain rule.

\textbf{\ref{enu:mult:rank}.} If $z\in\ctspc_{\mu}\backslash N_{0}$,
then
\[
\partial_{v}\left.\co^{-1}\begin{bmatrix}\ct(z)+v\\
-\mu
\end{bmatrix}\right|_{v=0}
\]
is invertible; the result follows since $\rank\alpha_{\perp}=q$.
\end{proof}

\section{Proofs relating to examples}

\label{app:examples}

\subsection{Linear systems}
\begin{proof}[Proof of \propref{lin-rep-var}]
We verify the conditions required for a model to be in class $\cvar_{r}$
(see \defref{cvar}); the other claims follow either in the course
of the subsequent arguments, or directly from the relevant definitions.

\textbf{\ref{enu:cvar:homeo}.} Since $\Phi_{0}$ is invertible, $\fe_{0}(z)=\Phi_{0}z$
is trivially a homeomorphism.

\textbf{\ref{enu:cvar:rank}.} Since $\rank\Phi(1)=r$, there exist
$\alpha,\beta\in\reals^{p\times r}$ having full column rank, such
that $\alpha\beta^{\trans}=-\Phi(1)$. Then $c\in\spn\Phi(1)=\alpha$
implies that $c=\alpha\cn$ for some $\cn\in\reals^{r}$. 

\textbf{\ref{enu:cvar:jsr}.} First recognise that 
\[
\pi(z)=-\Phi(1)z=\alpha\beta^{\trans}z,
\]
so that $\eq(z)=\beta^{\trans}$, and that 
\[
\ga_{j}(z)=-\sum_{i=j+1}^{k}\fe_{i}(z)=-\sum_{i=j+1}^{k}\Phi_{i}z=\Gamma_{j}z.
\]
Letting $\b{\Gamma}\defeq[\Gamma_{i}]_{i=1}^{k-1}$, so that $\b{\ga}(z)=\b{\Gamma}z,$
we thus have
\[
(\b{\eq}\compose\fe_{0}^{-1})(z)+\b D_{0}\b z=\begin{bmatrix}\beta^{\trans}\\
\b{\Gamma}
\end{bmatrix}\Phi_{0}^{-1}z+\b D_{0}\b z=\begin{bmatrix}\beta^{\trans}\Phi_{0}^{-1} & 0\\
\b{\Gamma}\Phi_{0}^{-1} & D
\end{bmatrix}\begin{bmatrix}z\\
\b{\zeta}
\end{bmatrix}\eqdef\b{\beta}^{\trans}\b z
\]
which is linear, and so \eqref{jsr} will be satisfied if the eigenvalues
of $I_{p(k-1)+r}+\b{\beta}^{\trans}\b{\alpha}$ are less than some
$\abv{\rho}<1$ in modulus.

By the Weinstein--Aronszajn identity, the non-unit eigenvalues of
$I_{pk}+\b{\alpha}\b{\beta}^{\trans}$ coincide with those of $I_{p(k-1)+r}+\b{\beta}^{\trans}\b{\alpha}$.
Therefore, consider
\begin{equation}
I_{pk}+\b{\alpha}\b{\beta}^{\trans}=I_{pk}+\begin{bmatrix}\alpha & E^{\trans}\\
0 & I_{p(k-1)}
\end{bmatrix}\begin{bmatrix}\beta^{\trans}\Phi_{0}^{-1} & 0\\
\b{\Gamma}\Phi_{0}^{-1} & D
\end{bmatrix}=\begin{bmatrix}I_{p}+(\alpha\beta^{\trans}+\Gamma_{1})\Phi_{0}^{-1} & E^{\trans}D\\
\b{\Gamma}\Phi_{0}^{-1} & I_{p(k-1)}+D
\end{bmatrix}\label{eq:bbeta-lin}
\end{equation}
whence
\begin{align*}
\begin{bmatrix}I_{p} & 0\\
0 & D
\end{bmatrix}(I_{pk}+\b{\alpha}\b{\beta}^{\trans})\begin{bmatrix}I_{p} & 0\\
0 & D^{-1}
\end{bmatrix} & =\begin{bmatrix}I_{p} & 0\\
0 & D
\end{bmatrix}\begin{bmatrix}I_{p}+(\alpha\beta^{\trans}+\Gamma_{1})\Phi_{0}^{-1} & E^{\trans}\\
\b{\Gamma}\Phi_{0}^{-1} & D^{-1}+I_{p(k-1)}
\end{bmatrix}\\
 & =\begin{bmatrix}I_{p}+(\alpha\beta^{\trans}+\Gamma_{1})\Phi_{0}^{-1} & E^{\trans}\\
D\b{\Gamma}\Phi_{0}^{-1} & I_{p(k-1)}+D
\end{bmatrix}.
\end{align*}
Noting that 
\[
\alpha\beta^{\trans}+\Gamma_{1}=-\Phi(1)+\Gamma_{1}=-\left(\Phi_{0}-\sum_{i=1}^{k}\Phi_{i}\right)+\left(-\sum_{i=2}^{k}\Phi_{i}\right)=\Phi_{1}-\Phi_{0}
\]
and
\[
D\b{\Gamma}=\begin{bmatrix}-I_{p} & I_{p}\\
 & \ddots & \ddots\\
 &  & -I_{p} & I_{p}\\
 &  &  & -I_{p}
\end{bmatrix}\begin{bmatrix}\Gamma_{1}\\
\vdots\\
\Gamma_{k-2}\\
\Gamma_{k-1}
\end{bmatrix}=\begin{bmatrix}\Gamma_{2}-\Gamma_{1}\\
\vdots\\
\Gamma_{k-1}-\Gamma_{k-2}\\
-\Gamma_{k-1}
\end{bmatrix}=\begin{bmatrix}\Phi_{2}\\
\vdots\\
\Phi_{k-1}\\
\Phi_{k}
\end{bmatrix}\eqdef\b{\Phi}_{(2)}
\]
it follows that
\[
\begin{bmatrix}I_{p}+(\alpha\beta^{\trans}+\Gamma_{1})\Phi_{0}^{-1} & E^{\trans}\\
D\b{\Gamma}\Phi_{0}^{-1} & I_{p(k-1)}+D
\end{bmatrix}=\begin{bmatrix}\Phi_{1}\Phi_{0}^{-1} & E^{\trans}\\
\b{\Phi}_{(2)}\Phi_{0}^{-1} & I_{p(k-1)}+D
\end{bmatrix}.
\]
The transpose of this matrix is the companion form for a VAR with
autoregressive polynomial
\[
I_{p}-(\Phi_{0}^{-1})^{\trans}\sum_{i=1}^{k}\Phi_{i}^{\trans}\lambda^{i}=(\Phi_{0}^{-1})^{\trans}\left[\Phi_{0}^{\trans}-\sum_{i=1}^{k}\Phi_{i}^{\trans}\lambda^{i}\right]=(\Phi_{0}^{-1})^{\trans}\Phi(\lambda)^{\trans},
\]
the roots of which are exactly the roots of $\Phi(\lambda)$.

It follows that the eigenvalues of $I_{p(k-1)+r}+\b{\beta}^{\trans}\b{\alpha}$
coincide with the inverses of the \emph{non-unit} roots of $\Phi(\lambda)$,
which by assumption lie strictly inside the unit circle, and are thus
(in modulus) bounded above by some $\abv{\rho}<1$, as required.
\end{proof}
The following is used in \subsecref{gjrt}.
\begin{lem}
\label{lem:lin-co-map}Suppose $(c,\{f_{i}\}_{i=0}^{k})$ is a linear
VAR satisfying \ref{enu:lin:homeo}\ass{--3}. Then for $y=(y_{(1)},y_{(2)})\in\reals^{q}\times\reals^{r}$,
\begin{align*}
\co^{-1}(y) & =\beta_{\perp}(\alpha_{\perp}^{\trans}\srp\beta_{\perp})^{-1}y_{(1)}+\{I-\beta_{\perp}(\alpha_{\perp}^{\trans}\srp\beta_{\perp})^{-1}\alpha_{\perp}^{\trans}\srp\}\beta(\beta^{\trans}\beta)^{-1}y_{(2)}.
\end{align*}
\end{lem}
\begin{proof}
By \propref{lin-rep-var}, we have that $(c,\{f_{i}\}_{i=0}^{k})\in\cvar_{r}$
and 
\[
\co(z)=\begin{bmatrix}\alpha_{\perp}^{\trans}\srp\\
\beta^{\trans}
\end{bmatrix}z.
\]
It therefore follows from \lemref{co-cts} that the matrix on the
r.h.s.\ is invertible. Since
\[
\begin{bmatrix}\alpha_{\perp}^{\trans}\srp\\
\beta^{\trans}
\end{bmatrix}\begin{bmatrix}\beta_{\perp} & \beta\end{bmatrix}=\begin{bmatrix}\alpha_{\perp}^{\trans}\srp\beta_{\perp} & \alpha_{\perp}^{\trans}\srp\beta\\
0 & \beta^{\trans}\beta
\end{bmatrix},
\]
it follows that the r.h.s.\ matrix is also invertible, and hence
\begin{align*}
\begin{bmatrix}\alpha_{\perp}^{\trans}\srp\\
\beta^{\trans}
\end{bmatrix}^{-1} & =\begin{bmatrix}\beta_{\perp} & \beta\end{bmatrix}\begin{bmatrix}(\alpha_{\perp}^{\trans}\srp\beta_{\perp})^{-1} & -(\alpha_{\perp}^{\trans}\srp\beta_{\perp})^{-1}\alpha_{\perp}^{\trans}\srp\beta(\beta^{\trans}\beta)^{-1}\\
0 & (\beta^{\trans}\beta)^{-1}
\end{bmatrix}\\
 & =\begin{bmatrix}\beta_{\perp}(\alpha_{\perp}^{\trans}\srp\beta_{\perp})^{-1} & M\end{bmatrix}
\end{align*}
where
\[
M=\{I-\beta_{\perp}(\alpha_{\perp}^{\trans}\srp\beta_{\perp})^{-1}\alpha_{\perp}^{\trans}\srp\}\beta(\beta^{\trans}\beta)^{-1}.
\]
Thus
\[
\begin{bmatrix}\alpha_{\perp}^{\trans}\srp\\
\beta^{\trans}
\end{bmatrix}^{-1}\begin{bmatrix}y_{(1)}\\
y_{(2)}
\end{bmatrix}=\beta_{\perp}(\alpha_{\perp}^{\trans}\srp\beta_{\perp})^{-1}y_{(1)}+\{I-\beta_{\perp}(\alpha_{\perp}^{\trans}\srp\beta_{\perp})^{-1}\alpha_{\perp}^{\trans}\srp\}\beta(\beta^{\trans}\beta)^{-1}y_{(2)}.\qedhere
\]
\end{proof}

\subsection{Piecewise affine systems}

Recall from \subsecref{piecewise-affine} that $f:\reals^{p}\setmap\reals^{p}$
is piecewise affine if it is a continuous function of the form
\begin{equation}
f(x)=\sum_{\ell=1}^{L}\indic\{x\in\set X^{(\ell)}\}(\phi^{(\ell)}+\Phi^{(\ell)}x),\label{eq:piecewiseaffine}
\end{equation}
where the sets $\{\set X^{(\ell)}\}_{\ell=1}^{L}$ are convex and
partition $\reals^{p}$. We say that such an $f$ is:
\begin{itemize}
\item \emph{piecewise linear} if there exists a basis $\{a_{i}\}_{i=1}^{p}$
for $\reals^{p}$ such that each $\set Z^{(\ell)}$ can be written
as a union of cones of the form
\[
\set C_{{\cal I}}\defeq\{z\in\reals^{p}\mid a_{i}^{\trans}z\geq0\sep\forall i\in{\cal I}\text{ and }a_{i}^{\trans}z<0\sep\forall i\notin{\cal I}\}
\]
where ${\cal I}$ ranges over the subsets of $\{1,\ldots,p\}$, and
$\bar{\phi}_{i}^{(\ell)}=0$ for all $i$ and $\ell$;
\item \emph{threshold affine} if there exists an $a\in\reals^{p}\backslash\{0\}$
and thresholds $\{\tau_{\ell}\}_{\ell=0}^{L}$ with $\tau_{\ell}<\tau_{\ell+1}$,
$\tau_{0}=-\infty$ and $\tau_{L}=+\infty$, such that
\begin{equation}
\set Z^{(\ell)}=\{z\in\reals^{p}\mid a^{\trans}z\in(\tau_{\ell-1},\tau_{\ell}]\}.\label{eq:thresholdaffine}
\end{equation}
\end{itemize}
We now give three auxiliary results relating to functions of these
kinds, each of whose proofs follow immediately.
\begin{lem}
\label{lem:pwa-conhull}Suppose $f:\reals^{p}\setmap\reals^{p}$ is
the piecewise affine function in \eqref{piecewiseaffine}. Then:
\begin{enumerate}
\item \label{enu:cohull}for every $x^{\prime},x^{\prime\prime}\in\reals^{p}$,
there exists a $\Phi\in\ch\{\Phi^{(\ell)}\}_{\ell=1}^{L}$ such that
\[
f(x^{\prime\prime})-f(x^{\prime})=\Phi(x^{\prime\prime}-x^{\prime});
\]
\item \label{enu:inverse}if $f$ is invertible, then $\set Y^{(\ell)}\defeq f(\set X^{(\ell)})$
partition $\reals^{p}$, and 
\[
f^{-1}(y)=\sum_{\ell=1}^{L}\indic\{y\in\set Y^{(\ell)}\}(\Phi^{(\ell)})^{-1}(y-\phi^{(\ell)}).
\]
\end{enumerate}
Suppose further that $f$ is threshold affine function, i.e.\ \eqref{thresholdaffine}
holds. Then 
\begin{enumerate}[resume]
\item \label{enu:inv-ta}if $f$ is invertible,
\[
f^{-1}(y)=\sum_{\ell=1}^{L}\indic\{b^{\trans}y\in(\nu_{\ell-1},\nu_{\ell}]\}(\Phi^{(\ell)})^{-1}(y-\phi^{(\ell)})
\]
where $b^{\trans}\defeq a^{\trans}(\Phi^{(1)})^{-1}$, and for $\ell\in\{0,\ldots,L\}$
\[
\nu_{\ell}\defeq\frac{\det\Phi^{(\ell)}}{\det\Phi^{(1)}}\tau_{\ell}+a^{\trans}(\Phi^{(1)})^{-1}\phi^{(\ell)}
\]
\end{enumerate}
\end{lem}
\begin{proof}
\textbf{\ref{enu:cohull}.} Let $\phi(x)\defeq\sum_{\ell=1}^{L}\indic\{x\in\set X^{(\ell)}\}\phi^{(\ell)}$
and $\Phi(x)\defeq\sum_{\ell=1}^{L}\indic\{x\in\set X^{(\ell)}\}\Phi^{(\ell)}$,
so that these are constant on each $\set X^{(\ell)}$, and $f(x)=\phi(x)+\Phi(x)x$.
Now let $x^{\prime},x^{\prime\prime}\in\reals^{p}$; with this notation,
\begin{align}
f(x^{\prime\prime})-f(x^{\prime}) & =[\phi(x^{\prime\prime})-\phi(x^{\prime})]+[\Phi(x^{\prime\prime})x^{\prime\prime}-\Phi(x^{\prime})x^{\prime}]\nonumber \\
 & =[\phi(x^{\prime\prime})-\phi(x^{\prime})]+\Phi(x^{\prime})(x^{\prime\prime}-x^{\prime})+[\Phi(x^{\prime\prime})-\Phi(x^{\prime})]x^{\prime\prime}.\label{eq:dfdecomp}
\end{align}
Define
\[
x(\delta)\defeq(1-\delta)x^{\prime}+\delta x^{\prime\prime}
\]
for $\delta\in[0,1]$. Since $f$ is continuous, so too is $\delta\elmap f[x(\delta)]$.
Because $\phi$ and $\Phi$ are piecewise constant, and $\{\set X^{(\ell)}\}_{\ell=1}^{L}$
is a convex partition of $\reals^{p}$, it follows that $\delta\elmap\phi[x(\delta)]$
and $\delta\elmap\Phi[x(\delta)]$ have $m\in\{0,\ldots,L-1\}$ points
of discontinuity, located at some $\{\delta_{i}\}_{i=1}^{m}$ with
$\delta_{i}<\delta_{i+1}$ for all $i$. Let $\{x_{i}\}_{i=1}^{m-1}$
be chosen such that $x_{i}=x(\delta)$ for some $\delta\in(\delta_{i},\delta_{i+1})$,
and set $x_{0}\defeq x^{\prime}$ and $x_{m}\defeq x^{\prime\prime}$.
By the continuity of $\delta\elmap f[x(\delta)]$ at each $\delta=\delta_{i}$,
we must have
\begin{equation}
0=\lim_{\delta\dto\delta_{i}}f[x(\delta)]-\lim_{\delta\uto\delta_{i}}f[x(\delta)]=[\phi(x_{i})-\phi(x_{i-1})]+[\Phi(x_{i})-\Phi(x_{i-1})]x(\delta_{i}).\label{eq:ctyatdty}
\end{equation}
for $i\in\{1,\ldots,m\}$. Noting also that
\begin{equation}
x^{\prime\prime}-x(\delta_{i})=x^{\prime\prime}-[(1-\delta_{i})x^{\prime}+\delta_{i}x^{\prime\prime}]=(1-\delta_{i})(x^{\prime\prime}-x^{\prime}),\label{eq:xmxd}
\end{equation}
we may write the final term on the r.h.s.\ of \eqref{dfdecomp} as
\begin{align*}
[\Phi(x^{\prime\prime})-\Phi(x^{\prime})]x^{\prime\prime} & =[\Phi(x_{m})-\Phi(x_{0})]x^{\prime\prime}\\
 & =\sum_{i=1}^{m}[\Phi(x_{i})-\Phi(x_{i-1})]x^{\prime\prime}\\
 & =_{(1)}\sum_{i=1}^{m}[\Phi(x_{i})-\Phi(x_{i-1})][(1-\delta_{i})(x^{\prime\prime}-x^{\prime})+x(\delta_{i})]\\
 & =_{(2)}\sum_{i=1}^{m}(1-\delta_{i})[\Phi(x_{i})-\Phi(x_{i-1})](x^{\prime\prime}-x^{\prime})-\sum_{i=1}^{m}[\phi(x_{i})-\phi(x_{i-1})],
\end{align*}
where $=_{(1)}$ follows from \eqref{xmxd}, and $=_{(2)}$ from \eqref{ctyatdty}.
We note that
\[
\sum_{i=1}^{m}[\phi(x_{i})-\phi(x_{i-1})]=\phi(x_{m})-\phi(x_{0})
\]
and that setting $\delta_{0}\defeq0$ and $\delta_{m+1}=1$, we have
\begin{align*}
 & \sum_{i=1}^{m}(1-\delta_{i})[\Phi(x_{i})-\Phi(x_{i-1})]\\
 & \qquad\qquad\qquad=(1-\delta_{m})\Phi(x_{m})+\sum_{i=1}^{m-1}[(1-\delta_{i})-(1-\delta_{i+1})]\Phi(x_{i})-(1-\delta_{1})\Phi(x_{0})\\
 & \qquad\qquad\qquad=\sum_{i=0}^{m}[(1-\delta_{i})-(1-\delta_{i+1})]\Phi(x_{i})-\Phi(x_{0})\\
 & \qquad\qquad\qquad=\sum_{i=0}^{m}(\delta_{i+1}-\delta_{i})\Phi(x_{i})-\Phi(x_{0})
\end{align*}
and whence
\begin{align*}
[\Phi(x^{\prime\prime})-\Phi(x^{\prime})]x^{\prime\prime} & =\left[\sum_{i=0}^{m}(\delta_{i+1}-\delta_{i})\Phi(x_{i})\right](x^{\prime\prime}-x^{\prime})-\Phi(x_{0})(x^{\prime\prime}-x^{\prime})-[\phi(x_{m})-\phi(x_{0})]\\
 & =\left[\sum_{i=0}^{m}\lambda_{i}\Phi(x_{i})\right](x^{\prime\prime}-x^{\prime})-\Phi(x^{\prime})(x^{\prime\prime}-x^{\prime})-[\phi(x^{\prime\prime})-\phi(x^{\prime})]
\end{align*}
where $\lambda_{i}\defeq\delta_{i+1}-\delta_{i}>0$ and so $\sum_{i=0}^{m}\lambda_{i}=\sum_{i=0}^{m}(\delta_{i+1}-\delta_{i})=\delta_{m+1}-\delta_{0}=1$.
It follows from \eqref{dfdecomp} that
\[
f(x^{\prime\prime})-f(x^{\prime})=\left[\sum_{i=0}^{m}\lambda_{i}\Phi(x_{i})\right](x^{\prime\prime}-x^{\prime}).
\]
Finally, noting that for each $i\in\{1,\ldots,m\}$, there exists
an $\ell_{i}\in\{1,\ldots,L\}$ such that $\Phi(x_{i})=\Phi^{(\ell_{i})}$,
we have $\sum_{i=0}^{m}\lambda_{i}\Phi(x_{i})\in\ch\{\Phi^{(\ell)}\}_{\ell=1}^{L}$
as required.

\textbf{\ref{enu:inverse}.} Since $f$ is surjective, and the sets
$\{\set X^{(\ell)}\}$ partition $\reals^{p}$,
\[
\Union_{\ell=1}^{L}\set Y^{(\ell)}=\Union_{\ell=1}^{L}f(\set X^{(\ell)})=f\left(\Union_{\ell=1}^{L}\set X^{(\ell)}\right)=f(\reals^{p})=\reals^{p},
\]
while $\set Y^{(\ell)}\intsect\set Y^{(\ell^{\prime})}\neq\emptyset$
for every $\ell\neq\ell^{\prime}$, since $f$ is injective. The claimed
form for $f^{-1}(y)$ then follows from the linearity of $f:\set X^{(\ell)}\setmap\set Y^{(\ell)}$,
for each $\ell\in\{1,\ldots,L\}$.

\textbf{\ref{enu:inv-ta}.} Because $f$ in invertible, it follows
from Theorem~4 in \citet{GLM80Ecta} that $c_{\ell}\defeq(\det\Phi^{(1)})/(\det\Phi^{(\ell)})>0$
for $\ell\in\{1,\ldots,L\}$. Since $f$ is continuous at the thresholds,
\begin{equation}
\phi^{(\ell-1)}+\Phi^{(\ell-1)}x=\phi^{(\ell)}+\Phi^{(\ell)}x\label{eq:fcty}
\end{equation}
for all $x\in\reals^{p}$ such that $a^{\trans}x=\tau_{\ell-1}$.
Hence, there exists an $M^{(\ell)}\in\reals^{p\times p}$ such that
\[
\Phi^{(\ell)}-\Phi^{(\ell-1)}=M^{(\ell)}P_{a}
\]
where $P_{a}=a(a^{\trans}a)^{-1}a^{\trans}$, whence
\begin{equation}
\Phi^{(\ell)}=\Phi^{(1)}+\sum_{i=2}^{\ell}M^{(i)}P_{a}\eqdef\Phi^{(1)}+n^{(\ell)}a^{\trans},\label{eq:Phiell}
\end{equation}
and so by the Sherman--Morrison--Woodbury formula and Cauchy's
formula for a rank-one perturbation (see (0.7.4.1) and (0.8.5.11)
in \citealp{HJ13book}), that
\begin{equation}
c_{\ell}^{-1}=1+a^{\trans}(\Phi^{(1)})^{-1}n^{(\ell)}\label{eq:c-ell}
\end{equation}
and
\begin{align}
a^{\trans}(\Phi^{(\ell)})^{-1} & =\{1-[1+a^{\trans}(\Phi^{(1)})^{-1}n^{(\ell)}]^{-1}a^{\trans}(\Phi^{(1)})^{-1}n^{(\ell)}\}a^{\trans}(\Phi^{(1)})^{-1}\nonumber \\
 & =\frac{1}{1+a^{\trans}(\Phi^{(1)})^{-1}n^{(\ell)}}a^{\trans}(\Phi^{(1)})^{-1}=c_{\ell}a^{\trans}(\Phi^{(1)})^{-1}.\label{eq:aTPhi}
\end{align}

Now by the invertibility of $f$ and the result of part~\ref{enu:inverse},
it may be verified that for $\ell\in\{1,\ldots,L\}$
\begin{align*}
y\in\set Y^{(\ell)}\iff f^{-1}(y)\in\set X^{(\ell)} & \iff a^{\trans}f^{-1}(y)\in(\tau_{\ell-1},\tau_{\ell}]\\
 & \iff a^{\trans}(\Phi^{(\ell)})^{-1}(y-\phi^{(\ell)})\in(\tau_{\ell-1},\tau_{\ell}].
\end{align*}
By \eqref{aTPhi}, the final condition is equivalent to 
\[
c_{\ell}a^{\trans}(\Phi^{(1)})^{-1}(y-\phi^{(\ell)})\in(\tau_{\ell-1},\tau_{\ell}]\iff b^{\trans}y\in(c_{\ell}^{-1}\tau_{\ell-1}+b^{\trans}\phi^{(\ell)},c_{\ell}^{-1}\tau_{\ell}+b^{\trans}\phi^{(\ell)}]
\]
where $b^{\trans}\defeq a^{\trans}(\Phi^{(1)})^{-1}$. Observe $c_{\ell}^{-1}\tau_{\ell}+b^{\trans}\phi^{(\ell)}=\nu_{\ell}$.
Finally, we note that for $x\in\reals^{p}$ such that $a^{\trans}x=\tau_{\ell-1}$,
\begin{align*}
b^{\trans}\phi^{(\ell)} & =b^{\trans}\phi^{(\ell-1)}+b^{\trans}(\Phi^{(\ell-1)}-\Phi^{(\ell)})x\\
 & =_{(1)}b^{\trans}\phi^{(\ell-1)}+a^{\trans}(\Phi^{(1)})^{-1}(n^{(\ell-1)}-n^{(\ell)})\tau_{\ell-1}\\
 & =_{(2)}b^{\trans}\phi^{(\ell-1)}+(c_{\ell-1}^{-1}-c_{\ell}^{-1})\tau_{\ell-1}
\end{align*}
where $=_{(1)}$ holds by \eqref{fcty} and \eqref{Phiell}, and $=_{(2)}$
by \eqref{c-ell}, whence
\[
c_{\ell}^{-1}\tau_{\ell-1}+b^{\trans}\phi^{(\ell)}=c_{\ell-1}^{-1}\tau_{\ell-1}+b^{\trans}\phi^{(\ell-1)}=\nu_{\ell-1}.\qedhere
\]
\end{proof}
\begin{lem}
\label{lem:pwa-homeo}Suppose
\begin{enumerate}
\item $\fe(x)=\sum_{\ell=1}^{L}\indic\{x\in\set X^{(\ell)}\}(\phi^{(\ell)}+\Phi^{(\ell)}x)$
is either a
\begin{enumerate}[ref=(i)(\alph*)]
\item \label{enu:pwl}piecewise linear function, or
\item \label{enu:ta}threshold affine function; and
\end{enumerate}
\item $\sgn\det\Phi^{(\ell)}=\sgn\det\Phi^{(1)}\neq0$ for all $\ell\in\{1,\ldots,L\}$.
\end{enumerate}
Then $\fe:\reals^{p}\setmap\reals^{p}$ is a homeomorphism.
\end{lem}
\begin{proof}
By either Theorem\ 1 and 4 in \citet{GLM80Ecta}, which are applicable
in cases \ref{enu:pwl} and \ref{enu:ta} respectively, $\fe$ is
invertible. Being continuous by assumption, its restriction to any
compact subset $A$ of $\reals^{p}$ is accordingly a homeomorphism
onto its image, i.e.\ $\fe^{-1}$ is continuous on $\fe(A)$, for
any $A\subset\reals^{p}$ compact.

Let $B(z,\rho)$ and $\abv B(z,\rho)$ respectively denote the open
and closed balls of radius $\rho$, centred at $z$, in $\reals^{p}$.
We shall show below that, in each of cases \ref{enu:pwl} and \ref{enu:ta},
that
\begin{equation}
\delta_{n}\defeq\inf_{\smlnorm z\geq n}\smlnorm{\fe(z)}\goesto\infty\label{eq:delta-diverge}
\end{equation}
as $n\goesto\infty$. By construction, if $z\in\reals^{p}$ is such
that $\smlnorm{\fe(z)}<\delta_{n}$, then we must have $z\in B(0,n)$:
deduce by the invertibility of $f$ that $B(0,\delta_{n})\subset\fe[B(0,n)]$.
Therefore for every $y\in\reals^{p}$ and $\epsilon>0$, there exists
an $n_{0}\in\naturals$ such that
\[
B(y,\epsilon)\subset B(0,\delta_{n})\subset\fe[B(0,n)]\subset\fe[\abv B(0,n)]
\]
for all $n\ge n_{0}$. But as noted above, $\fe^{-1}$ is continuous
on $\fe[\abv B(0,n)]$; therefore it is continuous at $y$. Since
$y\in\reals^{p}$ was arbitrary, $\fe^{-1}$ is continuous on the
whole of $\reals^{p}$, whence $\fe$ is a homeomorphism.

It remains to prove \eqref{delta-diverge}. Suppose \ref{enu:pwl}
holds. Then $\fe(0)=0$, and $\fe$ is non-negative homogeneous of
degree one: so that for any $\lambda\geq0$, $\fe(\lambda z)=\lambda\fe(z)$.
Hence
\[
\delta_{n}=\inf_{\smlnorm z\geq n}\smlnorm{\fe(z)}=\inf_{\lambda\geq n}\inf_{\smlnorm z=1}\smlnorm{\fe(\lambda z)}=n\inf_{\smlnorm z=1}\smlnorm{\fe(z)}.
\]
Since $\fe^{-1}$ is continuous on $\fe[\abv B(0,1)]$, and $\fe(0)=0$,
we must have $\inf_{\smlnorm z=1}\smlnorm{\fe(z)}>0$. Hence $\delta_{n}\goesto\infty$
as required. Alternately, suppose \ref{enu:ta} holds. Then
\begin{align*}
\delta_{n}=\inf_{\smlnorm x\geq n}\smlnorm{\fe(z)} & =\inf_{\smlnorm x\geq n}\norm{\sum_{\ell=1}^{L}\indic\{a^{\trans}z\in(\tau_{\ell-1},\tau_{\ell}]\}(\bar{\phi}_{0}^{(\ell)}+\Phi_{0}^{(\ell)}z)}\\
 & \geq\min_{\ell\in\{1,\ldots,L\}}\inf_{\smlnorm z\geq n}\smlnorm{\Phi_{0}^{(\ell)}z}-\max_{\ell\in\{1,\ldots,L\}}\smlnorm{\bar{\phi}_{0}^{(\ell)}}\\
 & =n\min_{\ell\in\{1,\ldots,L\}}\inf_{\smlnorm z=1}\smlnorm{\Phi_{0}^{(\ell)}z}-\max_{\ell\in\{1,\ldots,L\}}\smlnorm{\bar{\phi}_{0}^{(\ell)}}
\end{align*}
Because $\det\Phi_{0}^{(\ell)}\neq0$ for all $\ell\in\{1,\ldots,L\}$,
each $\Phi_{0}^{(\ell)}$ has full rank, and so there exists a $C>0$
such that $\inf_{\smlnorm z=1}\smlnorm{\Phi_{0}^{(\ell)}z}\geq C$
for all $\ell\in\{1,\ldots,L\}$. Hence
\[
\delta_{n}\geq nC-\max_{\ell\in\{1,\ldots,L\}}\smlnorm{\bar{\phi}_{0}^{(\ell)}}\goesto\infty.\qedhere
\]
\end{proof}

\begin{proof}[Proof of \lemref{affine-conditions}]
 By \eqref{pi-piecewise} and the maintained assumptions,
\begin{equation}
\pi(z)=\sum_{\ell=1}^{L}\indic^{(\ell)}(z)(\bar{\pi}^{(\ell)}+\Pi^{(\ell)}z)=\alpha\sum_{\ell=1}^{L}\indic^{(\ell)}(z)(\bar{\mu}^{(\ell)}+\beta^{(\ell)\trans}z)=\alpha\eq(z)\label{eq:pi-pwa}
\end{equation}
and for $j\in\{0,\ldots,k-1\}$,
\begin{align}
\ga_{j}(z)=-\sum_{i=j+1}^{k}\fe_{i}(z) & =\sum_{\ell=1}^{L}\indic^{(\ell)}(z)\left[-\sum_{i=j+1}^{k}\bar{\phi}_{i}^{(\ell)}-\sum_{i=j+1}^{k}\Phi_{i}^{(\ell)}z\right]\nonumber \\
 & =\sum_{\ell=1}^{L}\indic^{(\ell)}(z)[\bar{\gamma}_{j}^{(\ell)}+\Gamma_{j}^{(\ell)}z].\label{eq:ga-pwa}
\end{align}
where $\bar{\gamma}_{j}^{(\ell)}\defeq-\sum_{i=j+1}^{k}\bar{\phi}_{i}^{(\ell)}$
and $\Gamma_{j}^{(\ell)}=-\sum_{i=j+1}^{k}\Phi_{i}^{(\ell)}$. Letting
$\bar{\b{\gamma}}^{(\ell)}\defeq[\gamma_{j}^{(\ell)}]_{j=1}^{k-1}$
and $\b{\Gamma}^{(\ell)}\defeq[\Gamma_{j}^{(\ell)}]_{j=1}^{k-1}$,
it follows from the preceding that
\begin{equation}
\b{\eq}(z)=\begin{bmatrix}\eq(z)\\
\b{\ga}(z)
\end{bmatrix}=\sum_{\ell=1}^{L}\indic^{(\ell)}(z)\left(\begin{bmatrix}\bar{\mu}^{(\ell)}\\
\bar{\b{\gamma}}^{(\ell)}
\end{bmatrix}+\begin{bmatrix}\beta^{(\ell)\trans}\\
\b{\Gamma}^{(\ell)}
\end{bmatrix}z\right)\eqdef m(z)+M(z)z,\label{eq:beq-decomp}
\end{equation}
so that $\b{\eq}$ is piecewise affine, with $m(z)$ and $M(z)$ being
constant on each $\set Z^{(\ell)}$. By the assumed continuity of
$\{\fe_{i}\}_{i=0}^{k}$, $\eq$ and $\{\ga_{i}\}_{i=1}^{k-1}$ are
continuous, and hence so too is $\b{\eq}$.

In view of \eqref{gradient}, we need to consider $\b{\eq}\compose\fe_{0}^{-1}$,
where by assumption $\fe_{0}^{-1}$ exists and is continuous. Recalling
that 
\[
\fe_{0}(z)=\sum_{\ell=1}^{L}\indic\{z\in\set Z^{(\ell)}\}(\bar{\phi}_{0}^{(\ell)}+\Phi_{0}^{(\ell)}z),
\]
it follows from \lemref{pwa-conhull}\ref{enu:inverse} that
\[
\fe_{0}^{-1}(y)=\sum_{\ell=1}^{L}\indic\{y\in\set Y^{(\ell)}\}(\Phi_{0}^{(\ell)})^{-1}(y-\bar{\phi}_{0}^{(\ell)})
\]
where $\set Y^{(\ell)}\defeq\fe_{0}(\set Z^{(\ell)})$. Since $m$
and $M$ in \eqref{beq-decomp} are constant on on each $\set Z^{(\ell)}$,
and $\fe_{0}^{-1}(\set Y^{(\ell)})=\set Z^{(\ell)}$, the compositions
$m\compose\fe_{0}^{-1}$ and $M\compose\fe_{0}^{-1}$ must themselves
be constant on each $\set Y^{(\ell)}$. Letting $m^{(\ell)}\defeq m(z)$
and $M^{(\ell)}\defeq M(z)$ for $z\in\set Z^{(\ell)}$, we thus have
\begin{align*}
\b{\eq}\compose\fe_{0}^{-1}(y) & =\sum_{\ell=1}^{L}\indic\{y\in\set Y^{(\ell)}\}[m^{(\ell)}+M^{(\ell)}(\Phi_{0}^{(\ell)})^{-1}(y-\bar{\phi}_{0}^{(\ell)})]\\
 & =\sum_{\ell=1}^{L}\indic\{y\in\set Y^{(\ell)}\}[(m^{(\ell)}-M^{(\ell)}(\Phi_{0}^{(\ell)})^{-1}\bar{\phi}_{0}^{(\ell)})+M^{(\ell)}(\Phi_{0}^{(\ell)})^{-1}y]\\
 & \eqdef\sum_{\ell=1}^{L}\indic\{y\in\set Y^{(\ell)}\}[\vartheta^{(\ell)}+\Theta^{(\ell)}y].
\end{align*}
Letting $\b y\defeq(y^{\trans},\b x^{\trans})$ for $\b x\in\reals^{p(k-1)}$,
it follows that
\[
\b{\eq}\compose\fe_{0}^{-1}(y)+\b D_{0}\b y
\]
is piecewise affine on $\{\set Y^{(\ell)}\times\reals^{p(k-1)}\}_{\ell=1}^{L}$,
which provide a convex partition of $\reals^{kp}$. Therefore by \lemref{pwa-conhull},
for any $\b y,\b y^{\prime}\in\reals^{kp}$,
\[
[\b{\eq}\compose\fe_{0}^{-1}(y)+\b D_{0}\b y]-[\b{\eq}\compose\fe_{0}^{-1}(y^{\prime})+\b D_{0}\b y^{\prime}]=\b{\beta}^{\trans}(\b y-\b y^{\prime}),
\]
where $\b{\beta}^{\trans}$ lies in the convex hull formed from
\[
\Theta^{(\ell)}E^{\trans}+\b D_{0}=M^{(\ell)}(\Phi_{0}^{(\ell)})^{-1}E^{\trans}+\b D_{0}=\begin{bmatrix}\beta^{(\ell)\trans}(\Phi_{0}^{(\ell)})^{-1} & 0\\
\b{\Gamma}^{(\ell)}(\Phi_{0}^{(\ell)})^{-1} & D
\end{bmatrix}=\b{\beta}^{(\ell)\trans}
\]
for $\ell\in\{1,\ldots,L\}$. Deduce that \eqref{gradient} holds
with ${\cal B}=\ch\{\b{\beta}^{(\ell)}\}_{\ell=1}^{L}$, and therefore
\begin{align*}
\rho_{\jsr}(\{I_{p(k-1)+r}+\b{\beta}^{\trans}\b{\alpha}\mid\b{\beta}\in{\cal B}\}) & =\rho_{\jsr}(\ch\{I_{p(k-1)+r}+\b{\beta}^{(\ell)\trans}\b{\alpha}\}_{\ell=1}^{L})\\
 & =\rho_{\jsr}(\{I_{p(k-1)+r}+\b{\beta}^{(\ell)\trans}\b{\alpha}\}_{\ell=1}^{L}),
\end{align*}
where the second equality follows by \citet[Prop.~1.8]{Jungers09}.
\end{proof}
\begin{proof}[Proof of \propref{affine}]
 We verify the conditions required for membership of $\cvar_{r}(\alpha,\abv b,\abv{\rho})$.

\textbf{\ref{enu:cvar:homeo}.} In view of \ref{enu:pwa:homeo}, \lemref{pwa-homeo}
yields that $\fe_{0}$ is a homeomorphism. (Note that the subsequent
arguments only require that $(c,\{\fe_{i}\}_{i=0}^{k})$ be piecewise
affine, so that the conclusions of the theorem also hold when it is
directly assumed that $\fe_{0}$ is a homeomorphism.)

\textbf{\ref{enu:cvar:rank}.} From \eqref{pi-piecewise} and \ref{enu:pwa:rank},
$c=\alpha\mu$ and
\[
\pi(z)=\sum_{\ell=1}^{L}\indic^{(\ell)}(z)(\bar{\pi}^{(\ell)}+\Pi^{(\ell)}z)=\alpha\sum_{\ell=1}^{L}\indic^{(\ell)}(z)(\bar{\mu}^{(\ell)}+\beta^{(\ell)\trans}z)\eqdef\alpha\eq(z).
\]

\textbf{\ref{enu:cvar:jsr}.} \lemref{affine-conditions} yields that
\eqref{gradient} is satisfied for ${\cal B}\subset\reals^{p(k-1)+r}$
such that
\[
\rho_{\jsr}(\{I_{p(k-1)+r}+\b{\beta}^{\trans}\b{\alpha}\mid\b{\beta}\in{\cal B}\})=\rho_{\jsr}(\{I_{p(k-1)+r}+\b{\beta}^{(\ell)\trans}\b{\alpha}\}_{\ell=1}^{L}),
\]
where the r.h.s.\ is bounded by $\abv{\rho}<1$, by \ref{enu:pwa:jsr}.

Finally, the claimed form for $\co$ follows by noting that $\eq(z)$
must be as in \eqref{pi-pwa} in the proof of \lemref{affine-conditions},
while \eqref{ct-def} here implies that
\begin{align*}
\ct(z)=\alpha_{\perp}^{\trans}\left[\fe_{0}(z)-\sum_{i=1}^{k-1}\ga_{i}(z)\right] & =\sum_{\ell=1}^{L}\indic^{(\ell)}(z)\alpha_{\perp}^{\trans}\left\{ \bar{\phi}_{0}^{(\ell)}+\Phi_{0}^{(\ell)}z-\sum_{i=1}^{k}(\bar{\gamma}_{i}^{(\ell)}+\Gamma_{i}^{(\ell)}z)\right\} \\
 & =\sum_{\ell=1}^{L}\indic^{(\ell)}(z)\alpha_{\perp}^{\trans}(\bar{\srfn}^{(\ell)}+\srp^{(\ell)}z),
\end{align*}
we we have used the expression for $\ga_{i}$ given in \eqref{ga-pwa}
in the proof of \lemref{affine-conditions}.
\end{proof}
\begin{proof}[Proof of \propref{smoothed}]
 We verify each of the requirements of \defref{cvar} in turn.

\textbf{\enuref{cvar:homeo}.} $\fe_{0,K}(z)=\Phi_{0}z=\fe_{0}$ is
trivially a homeomorphism, since $\Phi_{0}$ is invertible.

\textbf{\enuref{cvar:rank}.} Recalling \eqref{pi-def}, we have
\begin{multline*}
\pi_{K}(z)=-\fe_{0,K}(z)+\sum_{i=1}^{k}\fe_{i,K}(z)=\int_{\reals^{p}}\left[-\fe_{0}(z^{\prime})+\sum_{i=1}^{k}\fe_{i}(z^{\prime})\right]K(z^{\prime}-z)\diff z^{\prime}\\
=\int_{\reals^{p}}\pi(z^{\prime})K(z^{\prime}-z)\diff z^{\prime}=\alpha\int_{\reals^{p}}\eq(z^{\prime})K(z^{\prime}-z)\diff z^{\prime}\eqdef\alpha\eq_{K}(z)
\end{multline*}

\textbf{\enuref{cvar:jsr}.} Fix $\b z^{\prime}=(z^{\prime\trans},\b{\zeta}^{\prime\trans})^{\trans}$
and $\b z^{\prime\prime}=(z^{\prime\prime\trans},\b{\zeta}^{\prime\prime\trans})^{\trans}$
in $\reals^{p}\times\reals^{p(k-1)}$. Since $\b{\eq}_{K}(z)=\int_{\reals^{p}}\b{\eq}(z^{\prime})K(z^{\prime}-z)\diff z^{\prime}$
and $\fe_{0,K}=\fe_{0}$ is linear, we have
\begin{align*}
\b{\eq}_{K}\compose\fe_{0,K}^{-1}(z^{\prime})-\b{\eq}_{K}\compose\fe_{0,K}^{-1}(z^{\prime\prime}) & =\int_{\reals^{p}}\{\b{\eq}[\fe_{0}^{-1}(z^{\prime})+u]-\b{\eq}[\fe_{0}^{-1}(z^{\prime\prime})+u]\}K(u)\diff u\\
 & =\int_{\reals^{p}}\{(\b{\eq}\compose\fe_{0}^{-1})[z^{\prime}+\fe_{0}(u)]-(\b{\eq}\compose\fe_{0}^{-1})[z^{\prime\prime}+\fe_{0}(u)]\}K(u)\diff u
\end{align*}
whence
\begin{align*}
 & [\b{\eq}_{K}\compose\fe_{0,K}^{-1}(z^{\prime})+\b D_{0}\b z^{\prime}]-[\b{\eq}_{K}\compose\fe_{0,K}^{-1}(z^{\prime\prime})+\b D_{0}\b z^{\prime\prime}]\\
 & \qquad\qquad=\int_{\reals^{p}}[\{(\b{\eq}\compose\fe_{0}^{-1})[z^{\prime}+\fe_{0}(u)]+\b D_{0}\b z^{\prime}\}-\{(\b{\eq}\compose\fe_{0}^{-1})[z^{\prime\prime}+\fe_{0}(u)]+\b D_{0}\b z^{\prime\prime}\}]K(u)\diff u\\
 & \qquad\qquad=\int_{\reals^{p}}\b{\beta}_{u}^{\trans}(\b z^{\prime}-\b z^{\prime\prime})K(u)\diff u=\left[\int_{\reals^{p}}\b{\beta}_{u}K(u)\diff u\right]^{\trans}(\b z^{\prime}-\b z^{\prime\prime})\eqdef\abv{\b{\beta}}^{\trans}(\b z^{\prime}-\b z^{\prime\prime}),
\end{align*}
where $\b{\beta}_{u}\in{\cal B}$ for each $u\in\reals^{p}$, for
${\cal B}$ the set of gradients associated to $(c,\{\fe_{i}\}_{i=0}^{k})$.
Since ${\cal B}$ is convex by \lemref{affine-conditions}, $\abv{\b{\beta}}\in{\cal B}$.
Hence \enuref{cvar:jsr} holds with ${\cal B}_{K}={\cal B}$.
\end{proof}

\end{document}